\documentclass[11 pt, a4paper]{book}
\usepackage[a4paper, inner=4.1 cm, outer=4.1 cm,top=4.4 cm, bottom=3.6 cm]{geometry}
\usepackage[latin1]{inputenc}
\usepackage{float}
\usepackage{hyperref}
\usepackage{amsmath}
\usepackage{amsfonts}
\usepackage{amssymb}
\usepackage{amsthm}
\usepackage{graphicx}

\theoremstyle{plain}
\newtheorem{thm}{Theorem}
\newtheorem{cor}{Corollary}
\newtheorem{lem}{Lemma}
\theoremstyle{definition}
\newtheorem{defn}{Definition}
\theoremstyle{remark}

\usepackage{fancyhdr}
\fancypagestyle{plain}{%
\fancyhead{} 
}
\author{David Vitali e Andrea Mari}
\title{Optomechanical Devices: CV Quantum Systems}
\begin{document}
\begin{titlepage}

\begin{center}

\Large{{\huge U}NIVERSIT\'A \hspace{0.4 cm}  DEGLI \hspace{0.4 cm} {\huge S}TUDI  \hspace{0.4 cm} DI \hspace{0.4 cm} {\huge C}AMERINO} \tiny\\
\rule{13 cm}{1pt} \\

\vspace{0,4cm}

{\Large{F}\small{ACOLT\'A DI} \Large{S}\small{CIENZE E} \Large{T}\small{ECNOLOGIE}}
\vspace{0,4cm}

{\Large{\bfseries Corso di Laurea Specialistica in Fisica}}
\\ \vspace{0,3 cm}
\begin{large} (Classe 20/S)\end{large}
\vspace{0,4cm}

{\large \textit {Dipartimento di Fisica}}
\vspace{0.3cm}

\includegraphics[width=2cm]{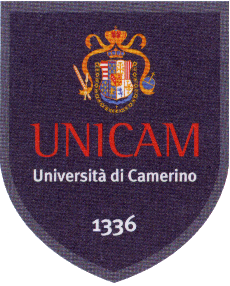}

\vspace{0.3 cm}

\Large{\bfseries OPTOMECHANICAL DEVICES}
\vspace{0,3 cm}

\Large{\bfseries ENTANGLEMENT AND TELEPORTATION}
\vspace{1 cm}
\begin{center}
 Abstract
\vspace{0.3 cm}
\end{center}

\begin{minipage}{15 cm}
\small 
In this thesis we study continuous variable systems, focusing on optomechanical devices where a laser field interacts with a macroscopic mechanical oscillator through radiation pressure. We consider the optimization problem of quantum teleportation by means of Gaussian local operations and we study the relation between the amount of entanglement shared by Alice and Bob and the maximum fidelity they can reach for the teleportation of coherent states. We also theoretically predict the possibility of creating a robust entaglement between a mechanical mode of a Fabry-Perot mirror and the Stokes sidband of the output field. We show that by driving two different optical modes supported by the cavity, a stable steady state is reachable in which the optomechanical entanglement is significantly improved.
%
\\ \end{minipage}

\vspace{1 cm}
\large Master's Thesis on \\
         Quantum Optics

\vspace{0.5 cm}

\large{Relatore:\hspace{9.2 cm}Candidato:}

\large{Prof. David Vitali\hspace{7.2 cm}Andrea Mari}

\vspace{0.5 cm}
\rule{13 cm}{1pt} \\
\normalsize{\Large A}NNO {\Large A}CCADEMICO {\Large 2007-2008}
\end{center}

\end{titlepage}

{
\pagestyle{empty}
\ 
\vspace{\stretch{1}}
\begin{flushright}
\textit{Alla prof.ssa Doriana Fabiani, \\
per tutto quello che mi ha dato \\
semplicemente facendo il suo lavoro.}
\end{flushright}
\ 
\vspace{\stretch{2}}
\clearpage
}

\frontmatter
\tableofcontents
\chapter{Introduction}

The birth and the development of quantum mechanics is a remarkable example of the scientific method introduced in the 17th century by Galileo Galilei. He defined this method as the composition of two fundamental steps: \textit{``sensate esperienze''} (sensory experience) and \textit{``necessarie dimostrazioni''} (necessary demonstrations)\cite{Galileo}. 

Unexplained phenomena like photoelectric emission, black body radiation, atomic emission lines, \textit{etc.}, encouraged the definition of a new theory called \textit{quantum mechanics}. This is the first step. However, simultaneously with the birth of the theory, scientists started to ask what are all the possible implications of the new quantum mechanical laws. This is the second step. They discovered that these  implications are very interesting and unusual, e.g. a system can be in a superposition of states, the result of a measurement can be predicted only as a probability, position and momentum of a particle cannot be simultaneously measured with infinite precision, the effect of a measurement on a system can instantaneously affect the state of another spacelike separated system, \textit{etc.}. Moreover, during the last century, other interesting questions were posed: ``Is quantum mechanics a complete theory? Or is it a manifestation of an underlying \textit{hidden variable} structure?'', ``Can we use quantum mechanics to transmit information or to perform computations, with better performances with respect to classical methods?'', ``If quantum mechanical laws are true for microscopic systems, why should not be true also in the macroscopic world?'', \textit{etc.}. This thesis is focused on the second step of the scientific method which Galilei called ``necessarie dimostrazioni''. In fact we consider some important implications and predictions of quantum mechanics and we study how they can be experimentally verified. 

In particular, in the first part of the thesis, we consider the quantum mechanical phenomenon of \textit{entanglement}, according to which the state of a system is so much correlated to the state of another spacelike separated system, that an operation performed on the first one causes a distant action on the second one. 

Within the formalism of quantum optics, we introduce also the technique of \textit{quantum teleportation}. This is a smart way of using entanglement and classical communication, to transmit an exact copy of a quantum state between two spacelike separated interlocutors (usually called Alice and Bob). We determine also which are the best local operations, that Alice and Bob can perform at their sites in order to optimize the teleportation of a coherent state. For a fixed amount of entanglement (negativity), we find upper and lower bounds for the maximum of the teleportation fidelity and we prove that this maximum is reached iff the correlation matrix of the shared entangled state is in a particular \textit{normal form}.

Usually it is said that quantum mechanics is the study of physical systems whose dimensions are very small such as atoms, electrons, subatomic particles, photons, \textit{etc.}. However this is not completely true, since quantum mechanics is a new way of looking nature and it is a theory which describes an electron as well as a planet. One of the aim of this work is indeed to give theoretical predictions of quantum effects in macroscopic systems. In particular, in the second part of the thesis, we will investigate optically driven mechanical oscillators: a new emerging research field named \textit{quantum optomechanics}. 

We consider an optical cavity in which one of the two mirrors is free to oscillate around its equilibrium position. The motion of the mirror is coupled with the cavity light field by radiation pressure. We will show that this coupling can optically cool a vibrational mode of the mirror and can entangle the state of the macroscopic mirror with the intracavity field. Moreover, we show that appropriately filtering the cavity output field, different optical modes can be extracted. These modes possess a significant amount of entanglement with the mirror, which can be also higher than the entanglement of the intracavity field. 

In the last chapter we consider also a cavity driven by two different lasers. We show that, by choosing opposite detunings and equal laser powers, a stable regime can be reached. In this regime, optomechanical entanglement between the laser fields and the mirror is improved. Also the two optical output fields, if appropriately filtered, are robustly entangled even at room temperature. 

Finally we apply the classification criterion of \cite{giedke} to show that, both in the monochromatic and in the bichromatic setup, we can find two optical output sidebands which together with the mirror vibrational mode are in a fully tripartite-entangled state.

\mainmatter

\part[\small{CONTINUOUS VARIABLE QUANTUM INFORMATION}]{CONTINUOUS VARIABLE QUANTUM INFORMATION}

\begin{chapter}{Quantum mechanics}
It is well known from quantum mechanics, that a pure state of a quantum system can be represented by a \textit{ray} of a Hilbert space associated to the system.
This ray is written in the Dirac notation \cite{Dirac} with a \textit{ket} vector $|\psi \rangle$ and it is said to be \textit{pure} since the uncertainty that it possesses is due only to the Heisenberg principle and not to our lack of information about it.
However, in general and especially in quantum optics, we do not know exactly which is the state vector of a system, but we know only the probability $P(a)$ that the system is in the state $|\psi_a\rangle$. In this case the system is said to be in a \textit{mixed} state.
For this reason it is better to use a mathematical object, firstly introduced by Von Neumann \cite{VonNeumann}, which is more powerful than the usual ket vector: the density operator.
\section{Density operator}
\begin{defn}
We associate to every state a density operator $\hat \rho$ defined through the following convex sum of projectors
\begin{equation}
\hat\rho= \sum_a {P(a) |\psi_a\rangle\langle \psi_a|}, \label{eq 1.1}
\end{equation}
where $P(a)$ is the probability that the system is in the pure state vector $|\psi_a\rangle$ and therefore we must require $\sum_a P(a)=1$.
\end{defn}
This new formalism associates to each possible state of a system an operator $\hat\rho$ which contains all the quantum and statistical informations that we know about the system.

The expected value of a generic observable $\hat O$ is given by
\begin{equation}
\langle \hat O \rangle=\sum_a {P(a) \langle \psi_a| \hat O |\psi_a\rangle } \label{eq 1.2}
\end{equation}
and we can write it in a compact form as
\begin{equation}
\langle \hat O \rangle=Tr\{ \hat \rho\hat O \}
\footnote{The trace of an operator $\hat A$ is defined as $ Tr \{ \hat A \} = \sum_n{ \langle n| \hat A |n \rangle}$,  where  $\{ |n \rangle \}$ is an arbitrary orthonormal basis of the Hilbert space.}
.\label{eq 1.3}
\end{equation} 

Three important properties follow form the definition,
\begin{eqnarray}
&&\hat \rho=\hat \rho^\dag{}, \label{eq 1.5}\\
 &&Tr\{ \hat \rho\}=1,  \label{eq 1.6}\\
 &&\hat \rho \ge 0.  \label{eq 1.7}
\end{eqnarray}

\begin{defn}
For every state $\hat \rho$ we can define \cite{VonNeumann} the \textit{Von Neumann entropy} $H(\hat\rho)$ as
\begin{equation}
H(\hat\rho)=-Tr\{ \hat\rho \ln \hat \rho \} \label{eq 1.8}.
\end{equation}
\end{defn}

The Von Neumann entropy can be considered as a measure of the disorder of a given quantum state. Moreover, at thermal equilibrium, it coincides with the classical thermodynamic entropy $S$ of the system (up to a multiplicative factor given by the Boltzmann constant $K_B$).

\section{Pure and mixed States}
\begin{defn}
Given a density operator $\hat\rho= \sum_a {P(a) |\psi_a\rangle\langle \psi_a|}$, if the set of vectors $\{|\psi_a\rangle\}$ is made of only one element, the state is \textit{pure}, otherwise it is mixed.
\end{defn}

\begin{thm}
A state  $\hat \rho$ is pure iff $\hat \rho^2=\hat \rho$.
\end{thm}
\begin{proof}
$\Rightarrow \hat \rho^2=|\psi\rangle\langle \psi||\psi\rangle\langle \psi|=\hat \rho.$\\
$\Leftarrow$ $\hat \rho$ is Hermitian and therefore it can be diagonalized
\begin{equation}
\hat \rho= \sum_n {\lambda_n |n\rangle\langle n|}.\label{eq 1.9}
\end{equation}
The hypothesis $\hat \rho^2=\hat \rho$ implies $\lambda_n=0, \pm1$, but for the normalization condition $Tr\{\hat \rho\}=1$ and $\hat \rho \ge 0$ the only possible eigenvalues are $\lambda_n=\delta_{n,k}$ and therefore $\hat \rho=|k\rangle\langle k|$.
\end{proof}

Another way to distinguish pure states from mixed states is provided by the trace of $\hat\rho^2$.

\begin{thm}
A state is pure iff $Tr\{\hat \rho^2\}=1$, while is mixed iff $Tr\{\hat \rho^2\}<1$.
\end{thm}
\begin{proof}
\begin{equation}
Tr\{\hat \rho^2\}=\sum_{a,b}P_aP_b|\langle \psi_a||\psi_b\rangle |^2 \le \sum_{a,b}P_aP_b=1 \label{eq 1.10},
\end{equation}
and the inequality becomes an equality only when $P_aP_b=\delta_{a,b}$, that is if $\hat \rho=|\psi\rangle\langle
\psi|$.
\end{proof}

Finally one could prove also the following theorem, showing that the purity of a state is connected with the amount of information that we have about it.
\begin{thm}
A state $\hat \rho$ is pure iff $H(\hat \rho)=0$, where $H(\hat \rho)$ is the Von Neumann entropy defined in (\ref{eq
1.8}).
\end{thm}

\section{Dynamics}
In the \textit{Schr\"odinger picture} the time evolution of a state is determined by the Hamiltonian $\hat H$ of the system through the Von Neumann equation (the density matrix equivalent of the Schr\"odinger equation),
\begin{equation}
\frac{d}{dt}\hat\rho(t)= \frac{1}{i\hbar}[\hat H,\hat \rho]
\end{equation}
which can be formally integrated giving
\begin{equation}
\hat\rho(t)=\exp(-\frac{i}{\hbar}\hat H t)\hat \rho(0) \exp(\frac{i}{\hbar} \hat H t),
\end{equation}
while all the observables do not depend on time.

A completely equivalent formulation of the system dynamics is given by the \textit{Heisenberg picture}, in which the state $\hat \rho$ remains constant in time while the observables evolve obeying the Heisenberg equation
\begin{equation}
\frac{d}{dt}\hat A= \frac{1}{i\hbar}[\hat A,\hat H], \label{HeisEq}
\end{equation}
which integrated gives
\begin{equation}
\hat A(t)=\exp(\frac{i}{\hbar}\hat H t)\hat A(0) \exp(-\frac{i}{\hbar} \hat H t). \label{Aevo}
\end{equation}

Finally we can give also another formulation which is intermediate between the two that we have seen: the \textit{interaction picture}.
Let us suppose that the Hamiltonian (in the Schr\"odinger picture) can be written as $\hat H=\hat H_0+\hat H_1$, where usually $\hat H_1$ is an interaction/perturbation term. We want to stay in a reference frame such that the effect of $\hat H_0$ \textit{cancelled}. This can be done by making the following transformation to the operators and to the interaction Hamiltonian $\hat H_1$,
\begin{eqnarray}
 \hat A_I(t)&=&\exp(-\frac{i}{\hbar}\hat H_0 t)\hat A(t) \exp(\frac{i}{\hbar} \hat H_0 t), \label{AIevo}\\
 \hat H_I(t)&=&\exp(-\frac{i}{\hbar}\hat H_0 t)\hat H_1 \exp(\frac{i}{\hbar} \hat H_0 t).
\end{eqnarray}
If we derive (\ref{AIevo}), using (\ref{HeisEq}), we get
\begin{equation}
\frac{d}{dt}\hat A_I(t)= \frac{1}{i\hbar}[\hat A_I(t),\hat H_I(t)], \label{IntEq}
\end{equation}
which is very similar to (\ref{HeisEq}), but in this case the evolution depends only on the interaction Hamiltonian $H_I$.

\end{chapter}

\begin{chapter}{CV systems in phase space}
A quantum system described by observables with a continuous spectrum is said to be a continuous variable (CV) system.
For example, the state of a particle in a generic potential is a CV system, since it is described in terms of its position and momentum which are continuous observables. 

All the quantum mechanical features of a system are very fragile if this system interacts with the environment, since the effect of decoherence destroys these quantum effects in a very short time. For this reason, thanks to its limited interaction with the environment, the electromagnetic radiation is a good candidate for observing quantum phenomena. It is well known that, in the canonical quantization of the electromagnetic field, a radiation mode is formally equivalent to a harmonic oscillator. Therefore when we talk about a CV system we usually refer to one or more harmonic oscillators described (if they are isolated) by the simple Hamiltonian

\begin{equation}
\hat{H} =\frac{1}{2}\sum_{k}^N \hbar \omega_k \Big( \hat p_k^2 +\hat q_k^2 \Big) \label{Hoscill},
\end{equation}
where $\hat q_k$ and $\hat p_k$ are dimensionless position and momentum operators satisfying the \textit{canonical commutation relations} (CCR): $[\hat q_k,\hat p_l]=i\delta_{il}$ and $[\hat q_k,\hat q_l]=[\hat p_k,\hat p_l]=0$. If we define the quadratures vector $\hat R=(\hat q_1,\hat p_1,\hat q_2,\hat p_2,...\hat q_N,\hat p_N)^T$ we can express the CCR in matrix form
\begin{equation}
 [\hat R_k,\hat R_l]=i\Omega_{kl},\label{CCR}
\end{equation}
where $\Omega$ is the $2N\times 2N$ symplectic matrix given by
\begin{equation}
\Omega=\bigoplus_{k=1}^N \left(
\begin{array}{cc}
  0&1 \\
 -1&0 \\
\end{array} 
\right).
\end{equation}
Since we deal with harmonic oscillators, it is useful to introduce also the \textit{annihilation} and \textit{creation} operators which are respectively
\begin{eqnarray}
 \hat a_k &=& \frac{\hat q_k+i \hat q_k}{\sqrt{2}},\\
 \hat a_k^\dag{}&=& \frac{\hat q_k-i \hat q_k}{\sqrt{2}},
\end{eqnarray}
and satisfy the bosonic commutation relations $[\hat a_k,\hat a_l^\dag]=\delta_{il}$ and $[\hat a_k,\hat a_l]=[\hat a_k^\dag,\hat a_l^\dag]=0$.

If we consider a single oscillator, it can be shown that the ground state of the Hamiltonian satisfies $\hat a |0\rangle=0$ and that all the other eigenstates (Fock states) are given by
\begin{equation}
 |n\rangle=\frac{\hat a^{\dag{}n}}{\sqrt{n!}}|0\rangle.
\end{equation}
Fock states have relatively large variances on the position and momentum operators. 
A quantum description which is close to the classical one is given by the \textit{coherent states} which are characterized by equal and minimum uncertainties in the quadratures $\Delta \hat q=\Delta \hat p=\frac{1}{2}$. These states are the eigenstates of the annihilation operator $\hat a |\alpha\rangle=\alpha|\alpha\rangle$ or equivalently $|\alpha \rangle=\exp(\alpha \hat a^\dag{}-\alpha^* \hat a)| 0 \rangle$, where $\alpha$ is the mean amplitude $\langle \hat a \rangle$ of the state. 

The states with minimum uncertainties but different variances in $\hat q$ and $\hat p$ are the \textit{squeezed states} $|\alpha,r\rangle=\exp(\alpha \hat a^\dag{}-\alpha^* \hat a)\exp(r^2 \hat a^{\dag{}2} -r^* \hat a^2)| 0 \rangle$. These  states (coherent and squeezed) are very important in quantum optics, however we do not give a deep treatment here since they are particular cases of the most general class of \textit{Gaussian states} studied in the next section.

If we want to give a quantum description of a system of oscillators we need to work with operators acting on infinite Hilbert spaces and this can create some problems. To overcome these problems it is better to use an equivalent representation in which we use c-number functions instead of operators and we work with a finite number of dimensions: the \textit{phase space}.

In classical mechanics, the state of a particle, or an ensemble of particles, is represented by a probability density $P(q,p)$ in phase space. With this distribution is possible to calculate the mean value of every observable,
\begin{equation}
\langle O(q,p)\rangle=\int P (q,p)O(q,p)dxdp \label{eq 4.1},
\end{equation}
in fact, if we power expand the observable $O(q,p)$, eq. (\ref{eq 4.1}) reduces to a sum of mean values of terms like $q^np^m$
\begin{equation}
\langle O(q,p)\rangle=\sum_{n,m}c_{nm}\int q^np^m P(q,p)dqdp \label{eq 4.2}.
\end{equation}

If we come back to the quantum world, it is natural to try to define a distribution $\mathcal W(q,p)$ in phase space similar to the classical one $P(q,p)$.
However if we deal with quantum states two problems arise. The first is that we have to consider the intrinsic quantum uncertainty of a state, so that for example a distribution which is too much localized is forbidden by the Heisenberg principle. Another problem is that in (\ref{eq 4.2}) the  variables $q$ and $p$ commute while the quantum operators $\hat{q}$ and $\hat{p}$ no not commute. For this reason we have different quantum distributions depending on the order of the variables products that we want to calculate.

\section{The Wigner function}
The fundamental operator that we need in the passage from the Hilbert space to the phase space is the \textit{Weyl operator} defined as 
\begin{equation}
\hat W_\xi= e^{i\xi^T \Omega \hat R},
\end{equation}
where $\Omega$ is the symplectic matrix, $\hat R$ is the quadrature vector and $\xi \in \mathbb R^{2N}$.
It can be shown that the CCR (\ref{CCR}) are equivalent to the \textit{Weyl relations}
\begin{equation}
\hat W_{\xi_1}\hat W_{\xi_2}= e^{i\xi_1^T \Omega \xi_2} \hat W_{\xi_1+\xi_2}.\label{Wrel}
\end{equation}
As we can describe a CV system using the quadrature operators satisfying the CCR, in the same way we can give an equivalent description in terms of Weyl operators satisfying (\ref{Wrel}).
Using these operators we can define something which is similar to a Fourier transform but it is applied to operators instead of functions: the \textit{Fourier-Weyl Transform}.
\begin{defn}
The Fourier-Weyl Transform is a map from $\mathcal H$ to $\mathcal L^2(\mathbb R^{2N})$ which associate to each operator $\hat A$ the function
\begin{equation}
A(\xi)=\textrm{Tr}(\hat A \hat W_\xi).\label{FWT}
\end{equation}
\end{defn}
The inversion of (\ref{FWT}) give us the possibility to express any operator as a linear combination of Weyl operators, in fact we have
\begin{equation}
\hat A=\frac{1}{(2\pi)^N}\int A(\xi) \hat W_{-\xi}d\xi.
\end{equation}
In analogy with the classical one, it can be shown also the \textit{quantum Parseval theorem}
\begin{equation}
\textrm{Tr}(\hat A \hat B)=\frac{1}{(2\pi)^N}\int A(\xi)^* B(\xi)d\xi.\label{parse}
\end{equation}

The most important application of the Fourier-Weyl Transform is when we use it to transform the density operator. In this case the result is called \textit{characteristic function}
\begin{equation}
\chi(\xi)=\textrm{Tr}(\hat \rho \hat W_\xi)\label{chi},
\end{equation}
because of the analogy with the classical characteristic function which is the Fourier transform of a density distribution.
For the Parseval theorem (\ref{parse}), the properties of the density operator that we have seen in the previous chapter imply
\begin{eqnarray}
\chi(0)&=&1,\\
\frac{1}{(2\pi)^N}\int |\chi(\xi)|^2 d\xi&\le& 1,
\end{eqnarray}
where the last integral is equal to $1$ iff the state $\hat \rho$ is pure.

Finally, an important property of the characteristic function is that, through its derivatives, we can express all the \textit{moments} of the density operator. In fact, if we define $\hat R'=\Omega \hat R$, we have
\begin{equation}
\textrm{Tr}[\hat \rho \{ (\hat R'_k)^n,(\hat R'_l)^m\}_{sym}] = \frac{(-i)^{n+m}}{n+m} \frac{\partial^{n+m}}{\partial\xi_k^m \partial\xi_l^n} \chi(\xi)\bigg\vert_{\xi=0} \label{chimom},
\end{equation}
where $\{ , \}_{sym}$ means the symmetrized product.

So, thanks to the characteristic function we can map a density operator to an element of $\mathcal L^2$, but if we want something which is analogous to a probability distribution in a phase space we need the inverse Fourier transform of $\chi(\eta)$
\begin{equation}
\mathcal W(\xi) = \frac{1}{(2\pi)^{2N}}\int \chi(\eta) e^{-i \xi^T \Omega \eta}d\xi, \label{Wigner}
\end{equation}
which is called the \textit{Wigner function}.
This function is very similar to a probability distribution but it has some different properties, for example it can take also negative values.
However it is true that, like a density probability, the Wigner function can be used to calculate the mean values of symmetrically ordered operators
\begin{equation}
\textrm{Tr}[\hat \rho \{ (\hat R'_k)^n,(\hat R'_l)^m\}_{sym}] = \int  \xi_k^n \xi_l^m \mathcal W(\xi) d\xi.
\end{equation}

\section{Gaussian states and Gaussian operations}

\begin{defn}
A quantum state $\hat \rho$ associated to $N$ harmonic oscillators is \textit{Gaussian} if its characteristic function is a Gaussian.
\end{defn}
If we define the displacement vector $d=\textrm{Tr}\{\hat \rho \hat R \}$, and the symmetrically ordered covariance matrix (CM) $V_{kl}=\textrm{Tr}\{\hat \rho [\hat R_k-d_k,\hat R_l-d_l]_+\}$, the characteristic function of a Gaussian state is
\begin{equation}
\chi(\eta)=e^{-\frac{1}{4}( \Omega \eta)^T V \Omega \eta + i d^T \Omega \eta},\label{chiga}
\end{equation}
where we can see that the first and the second moments completely define the state.
From the definition (\ref{Wigner}), the associated \textit{Wigner function} is
\begin{equation}
\mathcal W(\xi)=\pi^{-N}|V|^{-\frac{1}{2}} e^{-(\xi-d)^T\Omega^T V^{-1}\Omega(\xi -d)}, \label{Wgau}
\end{equation}
where $|V|$ is the determinant of the CM.

We must observe that not every Gaussian function can be a Wigner function of a physical state. In fact the Heisenberg principle imposes a constraint on the correlation matrix \cite{NoiseMatrix}.
\begin{thm}
A correlation matrix $V$ correspond to a physical Gaussian state iff
\begin{equation}
V+i\Omega\ge0. \label{bonafide}
\end{equation}
\end{thm}
We have seen that through the Wigner function we can have a rigorous phase space representation of the state $\hat \rho$ of a system. Now we may ask: if we make a physical operation on the state $\rho$ what is the effect on the Wigner function? We will answer to this question considering only the set of \textit{Gaussian operations}, which are the transformations $\hat \rho \longrightarrow \varepsilon (\hat \rho)$ which maps Gaussian states into Gaussian states. Since a Gaussian state is determined by its first and second moments we can completely define a Gaussian operation describing its effect on the correlation matrix $V$ and on the displacement vector $d$.

\subsection{Displacement transformations}
\begin{defn}
A displacement transformation is a translation in phase space parametrized by the displacement vector $\tau\in \mathbb R^{2N}$. It changes only the first moments of a Gaussian state
\begin{equation}
 d \longrightarrow d'=d+\tau.
\end{equation}
\end{defn}
The corresponding operator in the Hilbert space is the Weyl operator, we have indeed that
\begin{equation}
\hat W_\tau \hat R \hat W_\tau^\dag{}=\hat R +\tau.
\end{equation}
Any Gaussian operation can be seen as a composition of a transformation leaving the first moments unchanged followed by a displacement. For this reason, all the next operations we are going to study will act only on the correlation matrix $V$ of the Gaussian state.

\subsection{Symplectic transformations}
\begin{defn}
A symplectic transformation is parametrized by a $2N\times2N$ matrix $S$ which transforms the set of coordinates $\hat R$ to $\hat R'=S\hat R$, leaving the CCR unchanged
\begin{equation}
 S\Omega S^T =\Omega. \label{CCRsymp}
\end{equation}
 The effect on the CM of a Gaussian state is 
\begin{equation}
 V \longrightarrow V'=SVS^T.
\end{equation}
\end{defn}
The set of these transformations forms the \textit{real symplectic group} $Sp(2N,\mathbb R)$, which is a subgroup of the \textit{real special unitary group} $SU(2N,\mathbb R)$ since we always have $\det S=1$.

\begin{thm}[Metaplectic Group]
Given a symplectic transformation $S$, we can always find a unitary operation $\hat U_S$ in the underlying Hilbert space which realizes $S$ in the phase space. In other words, for every $S\in Sp(2N,\mathbb R)$ there exists an (up to a phase unique) unitary $U_S$ such that 
\begin{equation}
 \hat U_S \hat W_\xi \hat U_S^\dag{}=\hat W_{S\xi}, \qquad \forall \xi \in \mathbb R^{2N}.
\end{equation}
\end{thm}
The set of all these unitaries $\hat U_S$ forms the \textit{metaplectic group} $Mp(2N)$, which is a two-fold covering of $Sp(2N,\mathbb R)$. In fact, given a symplectic matrix $S$ we always have a $\pm$ ambiguity in the choice of $\hat U_S$.

\begin{thm}[Euler Decomposition]
Every symplectic transformation $S\in Sp(2N,\mathbb R)$ can be decomposed in the following way
\begin{equation}
S=R_2\left(\bigoplus_{k=1}^N D_k\right)  R_1,
\end{equation}
where $D_k$ are $2 \times 2$ symplectic diagonal matrices $D_k=\textrm{diag}[r_k,r_k^{-1}]$, while $R_{1,2} \in Sp(2N,\mathbb R) \cap SO(2N,\mathbb R)$ are symplectic rotation matrices.
\end{thm} 
The unitary transformations $\hat U_{R_{1,2}}\in Mp(2N)$ associated to $R_{1,2}$, can be physically realized with \textit{passive} optical elements (beam splitters and phase plates). The matrices $D_k$ instead correspond to single mode squeezing operations, which must be realized using \textit{active} optical elements (non linear crystals, homodine measurements, \textit{etc.}).

\begin{thm}[Williamson \cite{Williamson}]
Every positive definite $2N\times2N$ matrix $V$ can be diagonalized by a symplectic transformation $S \in Sp(2N,\mathbb R)$ in the following form
\begin{equation}
SVS^T=\textrm{diag}[s_1,s_1,s_2,s_2,....s_N,s_N].\label{sympdiag}
\end{equation}
The elements $s_k$ are called the symplectic eigenvalues and from (\ref{bonafide}) we can say that $V$ is a CM of a Gaussian state iff $s_k\ge1$ for every $k$.
\end{thm}
\subsection{Trace preserving Gaussian CP maps: Gaussian channels}
\begin{defn}
A \textit{trace preserving Gaussian completely positive map} (TGCP) is parametrized by a pair of $2N\times2N$ matrices $S$ and $G=G^T>0$, which satisfy the constraint $G+i \Omega-i S \Omega S^T\ge0$. The effect on the CM of a Gaussian state is
\begin{equation}
 V \longrightarrow V'=SVS^T+G. \label{TGCPformula}
\end{equation}
\end{defn}
Physical Gaussian operations like amplification/attenuation or any other one in which we make no measurements is a TGCP map. These map are also called \textit{Gaussian channels} since for example a dissipative medium in which there is an interaction with other external modes is modelled by a TGCP map.

Finally, another definition of a TGCP map comes from the \textit{dilation theorem}, which asserts that every Gaussian channel can be seen as the effect on a subsystem (partial trace) of a global symplectic map acting on a larger Hilbert space.

\subsection{Non-trace preserving Gaussian CP maps}
The most general physical operation acts on the density operator as a completely positive map $\hat \rho \rightarrow \varepsilon(\hat \rho)$. To characterize a CP map it is useful to apply the Jamiolkowski isomorphism which establishes a one to one correspondence between CP maps $\varepsilon \in \mathcal H$ and positive definite density operators in the doubled Hilbert space $\hat \rho_\varepsilon \in \mathcal H \otimes \mathcal H$ \cite{giedkemap}. 

Given the CP map we can construct a physical state $\hat \rho_\varepsilon$ starting from a two-mode infinitely squeezed state
\begin{equation}
|EPR\rangle=\lim_{r\rightarrow \infty}\cosh(r)^{-1}\sum_{n=0}^{\infty} \tanh(r)^n|n\rangle \otimes |n\rangle,
\end{equation}
and then applying $\varepsilon$ to the first mode
\begin{equation}
\hat \rho_{\varepsilon}=(\varepsilon \otimes I) |EPR\rangle.
\end{equation}
The inversion of such correspondence is also possible. In fact, given a state $\rho_\varepsilon$ it can be shown that, if we use it as a channel for the teleportation of a state $\hat \rho$, we can get (probabilistically) at the output $\varepsilon(\hat \rho)$.
If the CP map is Gaussian then also the state $\rho_\varepsilon$ is a Gaussian state and its CM completely characterizes the map.
\begin{defn}
A \textit{Gaussian completely positive map} (GCP) is parametrized by a $4N\times4N$ correlation matrix $\Gamma$ corresponding to a physical state
\begin{equation}
 \Gamma=
\left(
\begin{array}{cc}
\Gamma_1    &  \Gamma_{12} \\
\Gamma_{12} &  \Gamma_2    
\end{array}
\right)\ge -i \Omega \oplus \Omega.
\end{equation}
The effect of the map on the CM of a Gaussian state is
\begin{equation}
 V \longrightarrow V'= \Gamma_1 - \Gamma_{12} \Lambda \frac{1}{\Lambda \Gamma_2 \Lambda +V}\Lambda \Gamma_{12}^T,\label{GCPformula}
\end{equation}
where $\Lambda=diag[1,-1,1,-1,\dots]$.
\end{defn}
\end{chapter}
Any physical Gaussian operation, including also measurements and classical communication, is a GCP map and can be described by (\ref{GCPformula}).
For example given a bipartite state with CM
\begin{equation}
 V=
\left(
\begin{array}{cc}
A  &  C \\
C^T & B   \\
\end{array}
\right),
\end{equation}
then the projection of the second mode $B$ into the vacuum $|0\rangle\langle0|$ is associated by the Jamiolkowski isomorphism to the physical state of CM\footnote{We assume to loose the measured mode, therefore the dimension of $\Gamma$ is not $(4+4)\times(4+4)$, but $(2+4)\times(2+4)$.}
\begin{eqnarray}
\Gamma_1&=&\left( \begin{array}{cc}  \cosh r    & 0 \\
                                     0   &   \cosh r\\
\end{array}\right),\\
\Gamma_{1,2}&=&\left( \begin{array}{cccc}  \sinh r    & 0&0&0 \\
                                     0   &   -\sinh r&0&0\\
\end{array}\right),\\
 \Gamma_2&=&
\left( \begin{array}{cccc}  \cosh r    &      0        &0&0 \\
                               0        &    \cosh r   &0&0 \\
0&0&1&0\\ 
0&0&0&1\\
\end{array}\right),
\end{eqnarray}
where $r \rightarrow \infty$. From (\ref{GCPformula}), it can be shown \cite{giedkemap} that the action of the map on $V$ is
\begin{equation}
 V \longrightarrow V'= A - C \frac{1}{B +I}C^T.
\end{equation}

Finally, we remark that the formula given in (\ref{TGCPformula}), defining a TGCP map, can be obtained as a particular limit of (\ref{GCPformula}) \cite{giedkemap}.

\begin{chapter}{Entanglement}

Erwin Schr\"odinger, in his famous publication about the \textit{Schr\"odinger's cat paradox} of 1935, firstly introduced the word \textit{entanglement} to define a quantum correlation that can exists between two or more systems, so that the state of a system is affected by the state of the others also if these systems are spatially separated.
Entanglement is a phenomenon which is expected form the postulates of quantum mechanics and it is based on the quantum superposition principle. In classical mechanics, superposition is not expected and therefore \textit{entanglement} cannot exist, for this reason it seems so strange and paradoxical to our daily vision of the world.

\section{Historical and mathematical points of view}

The problem of entanglement, also if this word did not exist yet, was pointed out by Einstein, Podolsky and Rosen through the \textit{EPR paradox} \cite{EPR}, which seemed to bring into question the foundations of quantum mechanics.

Let us take a look to it. Einstein \textit{et al.} considered two particles such that $x_1-x_2=0$ and $p_1+p_2=0$.
This is not forbidden by the uncertainty principle, since $\hat x_1-\hat x_2 \equiv \hat X_-$ commutes with $\hat p_1+\hat p_2 \equiv \hat P_+$. So, it is possible to realize a physical state, the \textit{EPR pair}, such that the previous operators are equal to zero with infinite precision: $\Delta \hat X_-= \Delta \hat P_+=0$.

If we postulate the hypothesis of \textit{locality}, according to which a measure made on a particle cannot affect the state of another space-like separated particle, then the paradox arises. In fact, by measuring $x_1$ and $p_2$, it seems possible to know with arbitrary precision also $x_2$ and $p_1$, but this would violate the Heisenberg principle. The paradox can be avoided if we allow quantum mechanics to be a \textit{nonlocal} theory (as it has been actually demonstrated), so that a measurement on a particle changes the state of the other one and the previous argument cannot be used.

Even if the aim was to criticize the completeness of the quantum theory, Einstein \textit{at al.} actually centered the heart of quantum mechanics: entanglement. The principal feature of entanglement is indeed the possibility to have nonlocal effects and the EPR pair is still now the cornerstone of every CV information theory. 

We have briefly seen the historical origin of entanglement, now we give a rigorous mathematical definition of it. 

\begin{defn}
Given $N$ subsystems $1,2,\dots N$, with the associated Hilbert spaces $\mathcal H_1, \mathcal H_2, \dots \mathcal H_N$, according to the postulates of quantum mechanics, the whole system is an element of the tensor product of the $N$ Hilbert spaces. A state $\hat \rho \in \mathcal H_1 \otimes \mathcal H_2, \dots \otimes\mathcal H_N$ is \textit{separable} if it can be expressed as a convex sum of operators like $\hat \rho^{(1)}\otimes \hat\rho^{(2)}, \dots \otimes \hat\rho^{(N)}$, where $\hat\rho^{(i)} \in \mathcal H_i$ . That is
\begin{equation}
\hat \rho=\sum_k p_k \hat \rho_k^{(1)}\otimes \hat\rho_k^{(2)}, \dots \otimes \hat\rho_k^{(N)}, \qquad p_k \ge 0, \qquad \sum_kp_k=1 \label{entdef}.
\end{equation}
Otherwise the state is said to be \textit{entangled}.
\end{defn}
We observe that, given a composite state, different but equivalent decompositions in terms of $\rho^{(i)} \in \mathcal H_i$  exist. For this reason, if a state is not written explicitly in the form (\ref{entdef}), we cannot be sure that it is entangled, in fact there could be a different decomposition showing that the state is separable.

For Gaussian states we can give an analogous definition which depends on correlation matrices instead of density operators.

\begin{defn}
A Gaussian state $\hat \rho \in \mathcal H_1 \otimes \mathcal H_2, \dots \otimes\mathcal H_N$ whose CM is $V$ is \textit{separable} if a set of Gaussian states $\hat \rho^{(i)} \in \mathcal H_i$ with respective CM denoted by $V^{i}$ exists, such that
\begin{equation}
V \ge V^{(1)} \oplus V^{(2)}\dots \oplus V^{(N)}\label{entdefgau}.
\end{equation}
Otherwise the state is \textit{entangled}.
\end{defn}

\section{Separability criteria}
The previous formal definitions (\ref{entdef},\ref{entdefgau}) are very simple but they are often useless from a practical point of view. In fact, we usually need an operational criterion to distinguish separable from entangled states. In the following subsections we give some of the most relevant separability criteria known in literature. 

In the next chapters we will consider almost always $1 \times 1$ modes bipartite state, therefore the descriptions of the following criteria will be given for bipartite states. However, these descriptions could be easily extended to $N$ parties states without introducing any conceptual difference.

\subsection{Peres-Horodecki criterion}

If $\hat \rho$ is a density operator of a physical state, then given an arbitrary orthonormal basis we can write the state as $\hat \rho=\sum_{ij} \rho_i^j |i\rangle \langle j|$, where the Hermitian matrix $\rho_i^j$ is the representation of $\hat \rho$ in the given basis.

We define the \textit{transposition operation} $\hat\rho \longrightarrow \hat \rho^T$ as the map which associate each density operator represented by $\rho_i^j$ to the operator represented by $\rho_j^i$.

We observe that if $\hat \rho$ is a physical state than also $\hat \rho^T$ is a density operator of a physical state, in fact it satisfies the same three conditions given in the first chapter (\ref{eq 1.5}), (\ref{eq 1.6}) and (\ref{eq 1.7}).

If we have a bipartite state written as $\hat \rho= \sum_{nm\mu\nu}\rho_{m \mu}^{n\nu} |m\rangle \otimes |\mu\rangle \langle n|\otimes \langle \nu|$, where the Latin letters refer to the system $A$ while the Greek letters refer to the system $B$,   we can define the \textit{partial transposition operation} (PT) $\hat\rho \longrightarrow \hat \rho^{T_B}$ as the map which associate each density operator represented by $\rho_{m \mu}^{n\nu} $ to the operator represented by $\rho_{m \nu}^{n\mu} $. So a partial transposition can be seen as a transposition made only with respect to the system $B$.

Even if $\hat \rho$ is a physical state, the operator $\hat\rho^{T_B}$ could be unphysical (not positive semidefinite). This is a consequence of the fact that the transposition operation is a positive but not completely positive map.

If $\hat \rho$ is a separable state, we have that the partial transposed is given by
\begin{equation}
\hat \rho^{T_B}=\sum_k p_k \hat \rho_{Ak}\otimes \hat \rho_{Bk}^T,
\end{equation}
but since the operators $\rho_{Bk}^T$ correspond to physical states, than also the state of the whole composite system $\rho$ is physical.

So, observing that the trace and the Hermiticity of an operator are invariant under partial transposition, we can give the following criterion \cite{Peres}.

\begin{thm}[Peres-Horodecki] 
If a bipartite state $\hat \rho$ is separable, then the operator $\hat \rho^{T_B}$ is positive semidefinite. Conversely if $\hat \rho^{T_B} \not\ge 0$, then the state is entangled.
\end{thm}

The condition $\hat \rho^{T_B} \not\ge 0$ is sufficient but not necessary for entanglement. In fact entangled states exist such that $\hat \rho^{T_B} \ge 0$. These are the \textit{bound entangled states} and cannot be distilled, while it can be shown that every state such that $\hat \rho^{T_B} \not\ge 0$ is distillable.

Peres-Horodecki criterion is very good for discrete variables systems, like for example $N$ qubits, while for continuous variables systems the following criteria are more suitable. 

\subsection{Simon criterion}

Simon criterion is a CV adaptation of the Peres-Horodecki criterion. The basic idea is to consider the PT operation as a time reversal of a subsystem, so that in the phase space it is equivalent to a momentum inversion of system $B$. Let us give a proof of this.
It can be shown that the Wigner function of a bipartite state can be written also in this form
\begin{equation}
 \mathcal W(\xi)=\frac{1}{\pi^2} \int \langle q-x| \hat \rho |q+x\rangle e^{-i2 x^T p} d^2x \label{wigpos}
\end{equation}
where $\xi=(q_A,p_A,q_B,p_B)^T$, $q=(q_A,q_B)^T$ and $p=(p_A,p_B)^T$. If we consider the representation of the density operator in the position basis, for the definition of the partial transposition operation, we have that $\langle q_A,q_B| \hat \rho^{T_B} |q_A',q_B'\rangle=\langle q_A,q_B'| \hat \rho |q_A',q_B\rangle$. 
Let us write the Wigner function of $\hat \rho^{T_B}$:
\begin{equation}
\mathcal W(\xi) \longrightarrow \mathcal W'(\xi)=\frac{1}{\pi^2} \int \langle q-x| \hat \rho^{T_B} |q+x\rangle e^{-i2 x^T p} d^2x \label{wigpos2}.
\end{equation}
Now, if we make the change of variables $x\longrightarrow Z x$ in (\ref{wigpos2}), where $Z=diag[1,-1]$, we obtain
\begin{equation}
\mathcal W'(\xi)= \frac{1}{\pi^2} \int \langle q-x| \hat \rho |q+x\rangle e^{-i2 x^T Z p} d^2x \label{wigpos3}.
\end{equation}
So we can conclude that the partial transposition acts as an inversion of the momentum $p_B$ and, if we define the PT matrix $\Lambda=[1,1,1,-1]$, we can write
\begin{equation}
PT:\; \mathcal W(\xi) \longrightarrow \mathcal W(\Lambda \xi).
\end{equation}

The effect on the correlation matrix of a Gaussian state is then given by a matrix multiplication
\begin{equation}
PT:\; V \longrightarrow \Lambda V \Lambda.
\end{equation}

Using the result given in (\ref{bonafide}), we can finally apply the Peres-Horodecki criterion to a CV system.

\begin{thm}[Simon] 
If a bipartite state is separable, then 
\begin{equation}
\Lambda V\Lambda+ i \Omega \ge 0.
\end{equation}
Conversely if $\Lambda V\Lambda+ i \Omega$ has some negative eigenvalues, then the state is entangled.
\end{thm}
Moreover the non positive PT condition is also necessary for bipartite entanglement of Gaussian states.

\subsection{Duan \textit{et al.} criterion}

Another criterion, which is a good test for CV entanglement has been proposed by Duan, Giedke, Cirac and Zoller \cite{Duan}.
First of all we define two EPR like operators, which depend on the real parameter $a$,
\begin{equation}
\hat L_a \equiv |a|\hat x_A+\frac{1}{a}\hat x_B, \qquad \hat M_a \equiv |a|\hat p_A-\frac{1}{a}\hat p_B, \qquad a \in \mathbb R.\;
\end{equation}
\begin{thm}[Duan \textit{et al.}]
If a bipartite state is separable, then 
\begin{equation}
\forall a\in \mathbb R,\; a\ne0 \quad \Delta \hat L_a^2+\Delta \hat M_a^2 \ge a^2+\frac{1}{a^2}.\label{duannece}
\end{equation}
\end{thm}
\begin{proof}
The inequality (\ref{duannece}) can be derived supposing $\hat\rho$ in the separable form given in (\ref{entdef}) and then applying the Cauchy inequality.
\end{proof}

Another similar but weaker necessary condition for separability has been given by Mancini, Giovannetti, Vitali and Tombesi \cite{Camerino,Camerino2}, 
\begin{thm}[Mancini \textit{et al.}]
If a bipartite state is separable, then 
\begin{equation}
\forall a\in \mathbb R,\; a\ne0 \quad \Delta \hat L_a^2 \Delta \hat M_a^2 \ge \frac{1}{4} (a^2+\frac{1}{a^2})^2.
\end{equation}
\end{thm}

If we consider bipartite Gaussian states, the Duan \textit{et al.} criterion can rearranged in order to give a necessary but also sufficient condition for separability.
In fact they shown that the CM of a bipartite state can be transformed by local symplectic operations to what they called \textit{standard form II}, 
\begin{equation}
V_{II}=
\left(
\begin{array}{cccc}
n_1	 & 	&-d_1	  &	\\
 	 &n_2 	& 	  &d_2	\\
-d_1 	 & 	&m_1& 		\\
	 &d_2	& 	  &m_2   
\end{array}
\right),\label{VDuan}
\end{equation}
where all the coefficients are positive and satisfy the following constraint
\begin{equation}
 (m_1-1)/(n_1-1)=(m_2-1)/(n_2-1)=(d_1-d_2)^2/(n_1-n_2)^2=\lambda^2, \quad \lambda \in \mathbb R.\label{duanlambda}
\end{equation}

\begin{thm}[Duan \textit{et al.} - Gaussian states]
A bipartite Gaussian state is separable if, and only if, when it is transformed to the standard form II we have
\begin{equation}
\Delta \hat L_\lambda^2+\Delta \hat M_\lambda^2 \ge \lambda^2+\frac{1}{\lambda^2},
\end{equation}
where $\lambda$ is the parameter defined in (\ref{duanlambda}).
\end{thm}
\section{Entanglement measures}
In the previous section we have seen several separability criteria which answer to the question: ``Is this state entangled?''. However, in the Quantum Information field, entanglement is a precious physical resource which one needs to quantify, like energy or entropy. So another question arises: ``How much is this state entangled?''.
To quantify the entanglement of a state we need an \textit{entanglement monotone} or \textit{entanglement measure}, which is defined as a function $\hat\rho \longrightarrow E(\hat\rho) \in \mathbb R^+$ which associates to each density operator a real positive number and satisfies (possibly all) the following natural and reasonable properties.
\begin{enumerate}
 \item[1]\textit{(Monotonicity).} The average entanglement of all the outcomes of a local operation and classical communication (LOCC) performed on $\hat \rho$ must be less than the entanglement of $\rho$. If ${\hat \rho_i}$ are the results of a LOCC performed on $\hat \rho$ expected with probabilities ${p_i}$, that is $\hat \rho \xrightarrow{LOCC} \hat \rho'=\sum_i p_i \hat \rho_i$, then 
\begin{equation}
E(\hat \rho) \ge \sum_i p_i E(\hat \rho_i).
\end{equation}
\item [2]\textit{(Convexity).} Entanglement can not be increased in average by mixing other entangled states, i.e.
\begin{equation}
E(\sum_ip_i \hat \rho_i) \le \sum_i p_i E(\hat \rho_i).
\end{equation}
\item[3]\textit{(Unitary equivalence).} Entanglement must be invariant under local unitary operations,
\begin{equation}
E(\hat U_1 \otimes \hat U_2 \dots \otimes\hat U_N \hat \rho\hat U_1^\dag \otimes \hat U_2^\dag \dots \otimes\hat U_N^\dag) = E(\hat \rho).
\end{equation}
\item[4]\textit{(Faithfulness).} $E(\hat\rho)= 0$ iff $\hat\rho$ is separable.
\item[5]\textit{(Additivity).} $E(\hat \rho_1 \otimes \hat \rho_2)=E(\hat \rho_1)+E(\hat\rho_2)$.
\end{enumerate}

If we put together the quite abstract conditions 1 and 2, we get a very concrete physical requirement for an entanglement measure:
\begin{equation}
 E(\hat \rho) \ge E(\varepsilon (\hat \rho)),
\end{equation}
where $\hat \rho \longrightarrow \varepsilon(\hat\rho)$ is a LOCC. This means that an entanglement measure can not increase under local operations and classical communications.

We are going to introduce three of the most important entanglement monotones: the \textit{entropy of entanglement}, the \textit{(log-)negativity}, and the \textit{entanglement of formation}.

\subsection{Entropy of entanglement}
\begin{defn}
The entropy of entanglement is a measure defined for pure states. Given a bipartite pure state $\hat \rho=|\psi\rangle\langle\psi|\in \mathcal H_A \otimes \mathcal H_B$, then the entropy of entanglement is given by the Von Neumann entropy of one subsystem (tracing out the other),
\begin{equation}
 E_H(\hat \rho)=H[\textrm{tr}_A (\hat \rho)]=-\textrm{tr}[\textrm{tr}_A (\hat \rho)\ln(\textrm{tr}_A (\hat \rho))].
\end{equation}
\end{defn}
This is a consequence of the general property of quantum mechanics, according to which the more a quantum state is entangled the less we know about the single subsystems, so that the amount of disorder of a single subsystem quantifies the entanglement of the global system.

Every bipartite pure state admits a Schmidt decomposition like
\begin{equation}
 |\psi\rangle=\sum_i c_i |a_i\rangle \otimes |b_i\rangle,
\end{equation}
where $\{a_i\}$ and $\{b_i\}$ are two orthonormal basis of the respective Hilbert spaces $\mathcal H_A$ and $\mathcal H_B$, and $c_i\in \mathbb R^+$ are the Schmidt coefficients satisfying $\sum_i c_i^2=1$. The squared numbers $c_i^2$ are the eigenvalues of the reduced states $\textrm{tr}_A(\hat \rho)$ and $\textrm{tr}_B(\hat \rho)$, which have the same spectrum. This implies that the entropy of entanglement is equal if computed on the subsystem $A$ or $B$ and it is a function of the Schmidt coefficients only, in fact
\begin{equation}
 E_H(\hat \rho)=H[\textrm{tr}_A (\hat \rho)]=H[\textrm{tr}_B (\hat \rho)]=-\sum_i c_i^2\ln c_i^2.
\end{equation}
The entropy of entanglement satisfies all the previous properties of an \textit{entanglement monotone} only if the initial state is pure. Moreover, it can be shown that if we require also a \textit{(weak) continuity} property, then the entropy of entanglement is the unique entanglement monotone for pure states, in the sense that any other entanglement measure is a monotonic function of $E_H(\hat\rho)$.

An operational meaning of $E_H(\hat\rho)$ is connected with the distillable entanglement. Given $N$ copies of a state $\hat\rho=|\psi\rangle\langle\psi|$, then, in the asymptotic limit $N\rightarrow \infty$, one can distill $N E_H(\hat \rho)$ copies of maximally entangled states.
\subsection{Entanglement negativity and logarithmic negativity}
\begin{defn}
 Given a bipartite state $\hat \rho \in \mathcal H_A \otimes \mathcal H_B$, then the \textit{entanglement negativity} is given by
\begin{equation}
 \mathcal N(\hat \rho)= \frac{||\hat \rho^{T_B}||_{tr}-1}{2},\label{neg}
\end{equation}
where $||\dots||_{tr}$ is the \textit{trace norm} defined as $||\hat A||_{tr}=\textrm{tr}(\sqrt{\hat A^\dag \hat A})$.
\end{defn}
The trace norm is equal to the sum of the singular values of a matrix, but in our case the density operator is Hermitian and therefore the singular values are equal to the modulus of the eigenvalues of $\hat \rho$. Using the fact that $\textrm{tr}(\hat \rho^{T_B})=1$, eq. (\ref{neg}) can be written also as
\begin{equation}
 \mathcal N(\hat \rho)= \sum_{\lambda_i<0} |\lambda_i|,
\end{equation}
where $\lambda_i$ are the eigenvalues of $\hat \rho^{T_B}$, and the sum is made only over the negative ones.
\end{chapter}

Now it is clear why $\mathcal N(\hat \rho)$ is connected with the entanglement of $\hat \rho$, in fact we have already seen that the condition $\hat \rho^{T_B}\ge 0$ is necessary for separability and so the negativity measures how much this condition is violated.
 
Another entanglement monotone, which is not convex but it is additive, is the \textit{logarithmic negativity} defined as
\begin{equation}
 E_N(\hat \rho)=\ln ||\hat \rho^{T_B} ||_{tr}.
\end{equation}
It is an upper bound to the distillation rate, but a clear operational meaning of the log-negativity is not known. However $E_N(\hat \rho)$ is frequently used in Quantum Information, because it is easily computable if we deal with qubits or CV Gaussian states.

\begin{thm} 
If $\hat \rho$ is a bipartite Gaussian state of $1 \times 1$ modes, characterized by its correlation matrix $V$, then
\begin{equation} 
 E_N(\hat \rho)=\max\{-\ln(\nu),0\}, \label{entnegformula}
\end{equation}
where $\nu$ is the minimum symplectic eigenvalues of the PT matrix $\Lambda V \Lambda$.
\end{thm}
\begin{proof}
The matrix $\Lambda V \Lambda$ is positive definite and therefore, by the Williamson theorem, it can be transformed with unitary operations to the diagonal form $diag(s_1,s_1,s_2,s_2)$ given in (\ref{sympdiag}). The corresponding density operator is a tensor product of two thermal-like states $\hat \rho=\hat \rho_{n_1} \otimes \hat \rho_{n_2}$ with mean excitation numbers $n_i=(s_i-1)/2$. Since $\Lambda V \Lambda \ge 0$ then $n_i\ge-1/2$, moreover if $n_i\ge0$ the state is a physical thermal state, therefore the only possibility to have negative eigenvalues is when $-1/2\le n_i \le 0$. The Fock basis expansion of a thermal-like state is 
\begin{equation}
\hat \rho_{n_i}=\sum_{k} \frac{n_i^k}{(n_i+1)^{k+1}}|k\rangle \langle k |
\end{equation}
and if we suppose $-1/2\le n_i \le 0$, then the trace norm is by definition $||\hat \rho_{n_i} ||_{tr}=\sum_{k} |n_i|^k/(1-|n_i|)^{k+1}=(1+2 n_i)^{-1}=s_i^{-1}$. 
The conservation of the determinant $\det V=\det (\Lambda V \Lambda)$ imposes that only one PT symplectic eigenvalue $\nu$ can be less then one, so that $E_N(\hat \rho)=\max\{-\ln(\nu),0\}$.
\end{proof}

The log-negativity like any other entanglement monotone is invariant under local unitary operations. This means that the minimum PT symplectic eigenvalue $\nu$ of a bipartite Gaussian state depends only on the symplectic invariants of the matrix $V$. In fact, if we write $V$ in a block form like
\begin{equation}
 V=
\left( \begin{array}{ll} A    &  C   \\
                         C^T &   B       
\end{array}\right) ,
\end{equation}
it can be shown that 
\begin{equation}
\nu=\sqrt{\frac{\Sigma-\sqrt{\Sigma^2-4\det V}}{2}} \label{nuformula}
\end{equation}
where $\Sigma=\det A + \det B -2 \det C$.
\subsection{Entanglement of formation}
\begin{defn}
Given a quantum state $\hat \rho$, the \textit{ entanglement of formation} is defined as
\begin{equation}
 E_F(\hat \rho)=\min \sum_{k}p_k E_H(|\psi_k\rangle \langle \psi_k |)
\end{equation}
with the constraint 
\begin{equation}
 \hat \rho=\sum_{k}p_k |\psi_k\rangle \langle \psi_k |.\label{rhodec}
\end{equation}
\end{defn}
The operational meaning of $E_F(\hat \rho)$ is the minimum averaged entropy of entanglement of pure states, which are required to create the state $\hat \rho$.

Since the decomposition (\ref{rhodec}) is not unique, the entanglement of formation is difficult to calculate. A simple result is known for symmetric bipartite Gaussian states (with $\det A=\det B$), for which we have that
\begin{equation}
 E_F(\hat \rho)=E_H(|\psi_r \rangle \langle \psi_r |)
\end{equation}
where $|\psi_r \rangle$ is a two modes squeezed state with the same entanglement negativity of $\hat \rho$.
That is $|\psi_r\rangle=\sum_n \tanh^n (r)/\cosh (r) |n\rangle|n\rangle$, with $r=E_N(\hat \rho)/2$.

A general result valid also for non-symmetric Gaussian states has been recently proposed \cite{simon-eof}.

\begin{chapter}{CV Teleportation}\label{cvtele}
\textit{Quantum teleportation} is the transfer of an unknown quantum state form a sender (Alice) to a receiver (Bob) by means of a shared bipartite entangled state and appropriate classical communication.
If the shared entanglement is infinite than Bob recovers an exact copy of the state sent by Alice.

The first quantum teleportation protocol was given by Bennett \textit{et al.} in 1993 \cite{TeleBennett} for discrete variables systems. In this chapter we will consider only the continuous variable version, firstly proposed by Vaidman (1994) \cite{Vaidman} and Braunstein \textit{et al.} (1998) \cite{BraunsteinK}.

The word ``teleportation'' could be misleading, since there is not a transfer of matter, energy or information, but only the transfer of a quantum state.
For this reason, quantum teleportation is not in contradiction with the \textit{relativistic principle} according to which a signal can not travel with a velocity greater than the speed of light and it is consistent with the \textit{no-cloning theorem}  \cite{NoCloning}, which forbids to create a duplicate of a quantum state.

From the point of view of the foundations of quantum mechanics, teleportation is not surprising in itself since it is simply an application of the more general exceptional phenomenon which Einstein called \textit{spooky action at a distance}. In fact we have already seen (EPR/Schr\"odinger paradoxes) that, if we act on a single part (Alice) of an entangled state, then we can change the state of the other (Bob). Now, if Alice and Bob, using only LOCC, make a proper distant action which project Bob state into Alice input state, then this is called quantum teleportation.

In this chapter we first introduce the Braunstein-Kimble protocol using an infinitely squeezed EPR state and successively we give compact analytical results for the fidelity (success of the teleportation), in the situation when the shared bipartite state is a general Gaussian state which is not infinitely entangled.

\section{Braunstein-Kimble protocol with infinite entanglement}
We give a schematic description of the Braunstein-Kimble teleportation protocol in the Heisenberg picture.
\begin{enumerate}
\item \textit{Initial conditions}

Alice and Bob share the two modes of an infinitely squeezed EPR pair, such that
\begin{equation}
\hat x_b+\hat x_a=\hat p_b-\hat p_a=0 \label{teleEPR}.
\end{equation}
Alice has the input state $\hat\rho_{in}$ described by the quadratures $\hat x_{in}$ e $\hat p_{in}$, that she wants to teleport.

\item \textit{Beam splitter}

Alice mixes her mode of the EPR pair and the input state, through a symmetric beam splitter. 
The effect on the annihilation operators of the two modes is the following unitary transformation

\begin{equation}
 \left(\begin{array}{c} \hat a_+ \\ \hat a_-  \end{array}\right) = \frac{1}{\sqrt 2}
 \left(\begin{array}{rr} 1  &   1  \\
                    1  &  -1     \end{array}\right)
 \left(\begin{array}{c} \hat a_{in} \\ \hat a_a  \end{array}\right),
\end{equation}

so the new quadratures become
\begin{equation}
\hat x_{\pm}=\frac{\hat x_{in} \pm \hat x_a}{\sqrt 2},\qquad \hat p_{\pm}=\frac{\hat p_{in} \pm \hat p_{a}}{\sqrt 2}.
\end{equation}

\item \textit{Bell measurement}

Alice makes two homodyne measurements of the quadratures $\hat x_+$ and $\hat p_-$.
If we call the results of the measured quadratures with two real numbers $x_+$ and $p-$, then the system after the measurement is such that
\begin{equation}
\hat x_a=-\hat x_{in}+\sqrt 2x_+, \qquad \hat p_a=+\hat p_{in} - \sqrt 2p_-.
\end{equation}

Due to the entanglement correlations expressed in (\ref{teleEPR}), the measurements made by Alice cause the following instantaneous collapse of the Bob mode
\begin{equation}
\hat x_b=\hat x_{in}-\sqrt 2x_+, \qquad \hat p_b=\hat p_{in} - \sqrt 2p_-.
\end{equation}

\item \textit{Classical communication}
Alice sends to Bob, with a classical channel, the results of her measurements:  $x_+$ and $p_-$.
Without this classical communication Bob has a completely indefinite knowledge of the state of his mode.

\item \textit{Conditional displacement}

Bob, according to the values $x_+$ and $p_-$ given by Alice, makes on his mode a correction displacement:
\begin{equation}
\begin{array}{l} \hat x_b \longrightarrow  \hat x_{out}  = \hat x_b -  \sqrt 2x_+\;, \\
              \hat p_b \longrightarrow  \hat p_{out}\: = \hat p_b + \sqrt 2p_-\;.
\end{array}
\end{equation}

This completes the teleportation, since now the output mode is described by the same quadratures of the input state
\begin{equation}
 \begin{array}{l}   \hat x_{out}  = \hat x_{in}\;, \\
                 \hat p_{out} = \,\hat p_{in}\;. 
\end{array}
\end{equation}
\end{enumerate}
The protocol is very simple, in fact a great advantage of using CV variables rather than discrete variables is that the teleportation can be realized using only two simple optical elements: a beam splitter and a homodyne detector.
\section{Real teleportation with Gaussian states}

If the shared entanglement is not infinite, then the output state of the teleportation will be not completely equal to the input. 

\subsection{BK protocol in the Heisenberg picture}

Now we are going to see the relation between the input and the output state of the teleportation when the shared state is not an ideal EPR pair, but a Gaussian state $\hat \rho$ with zero mean displacement. We suppose the input state $\hat \rho_{in}$ to be Gaussian and so, since the operations of the protocol are GCP maps, also the output state $\hat \rho_{out}$ will be Gaussian.

With the same notation that we used for Gaussian Wigner functions (\ref{Wgau}), the degrees of freedom are then the first moments $d_{in}$, the correlations $V_{in}$ of the input state $\hat \rho_{in}$ and the CM $V$ of $\hat \rho$. While, what we want to calculate are the first and the second moments of $\hat \rho_{out}$, that we call $d_{out}$ and $V_{out}$.

A general treatment of the teleportation (good also for non Gaussian states) can be given using the Wigner function approach \cite{pirrew}. However, since we suppose to deal with Gaussian states, we will calculate $d_{out}$ and $V_{out}$ in the Heisenberg picture, where it is easier to see the physical meaning underlying the equations.

Let us define three column vectors containing the quadrature operators of the three modes (input, Alice and Bob) involved in the teleportation protocol,
\begin{equation}
 \hat u_{in}=
\left( 
\begin{array}{c}
\hat q_{in}\\
\hat p_{in}
\end{array} 
\right), \quad
\hat u_{a}=
\left( 
\begin{array}{c}
\hat q_{a}\\
\hat p_{a}
\end{array} 
\right), \quad
\hat u_{b}=
\left( 
\begin{array}{c}
\hat q_{b}\\
\hat p_{b}
\end{array} 
\right). \quad
\end{equation}
We define the EPR vector that in the ideal situation is zero but in general is a pair of operators given by
\begin{equation}
\hat u_E= \hat u_b + Z \hat u_a \label{EPRvect}
\end{equation}
where $Z=diag(1,-1)$. The correlations between the elements of this new vector can be expressed in a $2 \times 2$ matrix
$N_{ij}=\langle \Delta \hat u_{Ei} \Delta \hat u_{Ej}+\Delta \hat u_{Ej} \Delta \hat u_{Ei} \rangle$ and if we write $V$ in block form
\begin{equation}
 V=
\left( \begin{array}{ll} A    &  C   \\
                         C^T &   B       
\end{array}\right) ,
\end{equation}
then, from the definition (\ref{EPRvect}), it is straightforward to prove that
\begin{equation}
 N= ZAZ+CZ+C^TZ+B. \label{Nepr}
\end{equation}
We define also an analogous vector, which contains the quadratures measured by Alice,
\begin{equation}
\hat u_m= \hat u_{in} + Z \hat u_a \label{mvect}.
\end{equation}
If we subtract (\ref{mvect}) to (\ref{EPRvect}), we have
\begin{equation}
 \hat u_b=\hat u_{in}+\hat u_E - \hat u_m \label{uinit},
\end{equation}
which is actually an identity, but it gives us a useful way of expressing the quadratures of the mode possessed by Bob in terms of the fundamental operators which are involved in the teleportation.
Up to now, we have not started the protocol yet. When Alice makes the two homodyne measurements, she actually measures the vector $\hat u_E$ with infinite precision so that it will collapse into a pair of real numbers
\begin{equation}
 \hat u_E \xrightarrow{measurement} \bar u_E \in \mathbb R^2.
\end{equation}
Then Alice, using a classical channel, sends to Bob the vector $\bar u_E$, and Bob performs the conditional displacement $\tau=\bar u_E$ to his mode, so that at the end of the protocol, (\ref{uinit}) becomes
\begin{equation}
 \hat u_b \xrightarrow{protocol} \hat u_{out}=\hat u_{in}+\hat u_E.\label{ufin}
\end{equation}
Thanks to the simplicity of the Heisenberg picture, we can see that the effect of the Braunstein-Kimble protocol is to cancel the vector $\hat u_m$ in (\ref{uinit}), while the vector $\hat u_E$ depends on the EPR correlations shared by Alice and Bob. 
In the ideal situation (\ref{teleEPR}) $\hat u_E$ is exactly 0, while in general it will introduce some additional noise that will affect the fidelity of the teleportation.

Since we supposed the bipartite state $\hat \rho$ to be centered in phase space, then $\langle \hat u_E \rangle=0$ and so, from (\ref{ufin}) we always have $d_{in}=d_{out}$. While, for what concerns the second moments, from (\ref{ufin}) we have that
\begin{equation}
 V_{out}=V_{in}+N. \label{Vout}
\end{equation}
Therefore, the difference between the input and the output state is the positive definite matrix $N$, which should be as much as possible near to zero in order to have a good teleportation.
\subsection{Teleportation fidelity}
A measure of how much a quantum state is similar to another one is the \textit{fidelity}.
\begin{defn}
Given two states $\hat \rho_1$ and $\hat \rho_2$, the \textit{fidelity} is defined as
\begin{equation}
\mathcal F(\hat\rho_1,\hat \rho_2)= \left[\textrm{tr} \sqrt{\sqrt{\hat \rho_1}\hat \rho_2 \sqrt{\hat \rho_1}}\,\right]^2, \label{fidgen}
\end{equation}
and if one operator is pure, e.g. $\rho_1=|\psi\rangle \langle \psi|$, then (\ref{fidgen}) reduces to
\begin{equation}
\mathcal F(|\psi\rangle \langle \psi|,\hat \rho_2)= \langle \psi |\hat \rho_2 |\psi \rangle \label{fidpure}.
\end{equation}
\end{defn}
The fidelity $\mathcal F(\hat\rho_1,\hat \rho_2)$ is a number between $0$ and $1$ and can be seen as the probability that measuring the system $\hat \rho_2$ it collapses to the state $\hat \rho_1$. In particular, if both states are pure, $\hat \rho_1=|\psi\rangle \langle \psi|$, $\hat \rho_2=|\phi\rangle \langle \phi|$, then the fidelity is equal to the transition probability between the two states $|\langle\psi|\phi\rangle|^2$.

A quantitative measure of the success of the teleportation is then given by the fidelity between the input and the output state and if we suppose the input state to be pure then we have
\begin{equation}
 \mathcal F=\langle \psi_{in}| \rho_{out} |\psi_{in} \rangle=\textrm{tr}(\hat \rho_{in} \hat \rho_{out}).\label{fidteltr}
\end{equation}
Using the Parseval theorem (\ref{parse}), the trace in the last part of (\ref{fidteltr}) can be written as a scalar product of characteristic functions
\begin{equation}
\mathcal F=\frac{1}{2 \pi} \int \chi_{in}^*(\eta)\chi_{out}(\eta) d\eta. \label{fidparse}
\end{equation}
The characteristic functions have the Gaussian form given in (\ref{chiga}), so that (\ref{fidparse}) is a Gaussian integral which can be solved analytically
\begin{equation}
\mathcal F=\frac{1}{2 \pi} \int e^{-\frac{1}{4} \eta^T\Omega^T (V_{in}+V_{out}) \Omega\eta} d\eta=\frac{2}{\sqrt{\det [V_{in}+V_{out}]}}.
\end{equation}
If we substitute the expression of $V_{out}$ (\ref{Vout}), we obtain a formula which is true for the general\footnote{If the state shared by Alice and Bob has zero displacement and Bob always performs the ``correct''
conditional displacement, the teleportation protocol is invariant under displacement transformation, i.e., all states with the same covariance
matrix but different coherent components are teleported with the same fidelity \cite{fiu02,pirrew}.} 
teleportation of an input pure state with CM $V_{in}$, by using a bipartite state characterized by the EPR noise matrix $N$ given in (\ref{Nepr}),
\begin{equation}
 \mathcal F=\frac{2}{\sqrt{\det[2V_{in}+N]}}.\label{fidmatrix}
\end{equation}

 If we want a measure of the quality of the teleportation in itself which is independent from the input state, we can use the \textit{fidelity of entanglement swapping} \cite{fiu02}, which is given by
\begin{equation}
\mathcal F_{swap}=\frac{2}{\sqrt{\det N}}.\label{fidswap}
\end{equation}

Finally we give an important theorem, which underlines the importance of quantum entanglement to realize a good teleportation.
\begin{thm} \label{Fmaxclass}
The maximum fidelity achievable in the teleportation of a coherent state, based only on a classical strategy (without entanglement), is $\mathcal F=1/2$, \cite{Fclass}.
\end{thm}
\subsection{A simple example}

As a simple example of quantum teleportation we consider a coherent input state ($V_{in}=I$) and a shared two-mode squeezed state characterized by the CM
\begin{equation}
 V=
\left( \begin{array}{cc} I \cosh r    &  -Z \sinh r  \\
                        -Z \sinh r    &   I \cosh r   
\end{array}\right).
\end{equation}
From (\ref{Nepr}), the noise matrix is equal to $N=2 e^{-r}I$, and using (\ref{fidmatrix}) we get the simple result
\begin{equation}
 \mathcal F=\frac{1}{1+e^{-r}}.
\end{equation}
It is interesting to note that if we use (\ref{nuformula}) to calculate the minimum symplectic eigenvalue of the PT state, we get $\nu=e^{-r}$ and so, we can write
\begin{equation}
 \mathcal F=\frac{1}{1+\nu}. \label{twomodefid}
\end{equation}
This is a particular example of the relation between the fidelity and the entanglement negativity which will be studied in detail in the next chapter.

Finally we observe that the expression (\ref{twomodefid}) is consistent with (Thm. \ref{Fmaxclass}), in fact if $\mathcal F>1/2$, then $\nu<1$ and this is a sufficient condition for the entanglement of the two-mode squeezed state.

\end{chapter}
\begin{chapter}[Optimal fidelity and entanglement]{Optimal fidelity and entanglement\normalsize{ \cite{OptFid}}} \label{optimization}
\section{Introduction}

The problem afforded in this chapter is: ``Given the entangled state $\hat \rho$ shared by Alice and Bob, if we are allowed to perform local operations and classical communications, which is the maximum of the teleportation fidelity we can get?'' and ``What is its relation with the entanglement of $\hat \rho$?''.
The problem is non-trivial because the set of operations that Alice and Bob can adopt is very large. In
fact, apart from local TGCP maps, they can adopt two further options: i) use \textit{non-trace preserving} Gaussian operations in which some
ancillary mode is subject to Gaussian measurement, i.e., projected onto a Gaussian state, rather than discarded \cite{giedkemap}; ii) use local
non-Gaussian operations (either with measurement on ancillas or not), i.e., those involving interactions which are non-quadratic in the
canonical coordinates. The first class of maps, together with TGCP maps, forms the most general class of Gaussian completely positive (GCP)
operations, capable of preserving the Gaussian nature of the state shared by Alice and Bob. Non-Gaussian operations instead will transform the
initial Gaussian bipartite state of Alice and Bob into a non-Gaussian one, and they can also increase the fidelity of teleportation in some
cases \cite{nonGtelep}.
Here we restrict to input coherent states, which represent the basic resource for many quantum communication schemes and because the conventional teleportation protocol of Ref.~\cite{BraunsteinK}, while working excellently for coherent states, is less suited for teleporting nonclassical states \cite{ade-chi08}.

The improvement of the teleportation of coherent states by means of local operations and its relation with the entanglement of the shared bipartite state has been already discussed in a number of papers \cite{Bow,kim,fiu02,ade-illu05}. Ref.~\cite{Bow} showed that in some cases the
fidelity of teleportation may be improved by local squeezing transformations, while Ref.~\cite{kim} showed that in the case of a shared
asymmetric mixed entangled resource, teleportation fidelity can be improved even by a local \textit{noisy} operation. These results were
generalized in Ref.~\cite{fiu02} which showed how fidelity can be maximized over all local trace-preserving Gaussian completely positive (TGCP)
maps, i.e., those that can be performed by first adding ancillary systems in Gaussian states, then performing unitary Gaussian transformations
on the whole system, and finally discarding the ancillas. Ref.~\cite{fiu02} confirmed that the optimal local TGCP map maybe a noisy one, i.e.,
that teleportation fidelity can be increased even by decreasing the entanglement and increasing the noise of the shared entangled state.
Ref.~\cite{fiu02}, however, did not discuss the relationship between entanglement and the optimal fidelity $\mathcal{F}_{opt}$.
Ref.~\cite{ade-illu05} instead found this relationship, but only for a subclass of symmetric Gaussian entangled state shared by Alice and Bob:
for this class it is $\mathcal{F}_{opt}=\left(1+\nu\right)^{-1}$, where $\nu$ is the lowest symplectic eigenvalue of the partial transposed (PT) state (\ref{nuformula}),which is connected with the entanglement log-negativity by $E_{\mathcal{N}}=\max [0,-\ln \nu]$.

In this chapter we generalize in various directions the results of Refs.~\cite{fiu02,ade-illu05}. We show that if Alice and Bob share a bipartite
Gaussian state with a given $\nu$ and one restricts to local GCP maps which preserve such a Gaussian nature, the optimized fidelity always
satisfies
\begin{equation}\label{lowupp}
    \frac{1+\nu}{1+3\nu} \leq \mathcal{F}_{opt}\leq \frac{1}{1+\nu}.
\end{equation}
We also show that the upper bound is reached iff Alice and Bob share a symmetric entangled state. Moreover we determine the optimal local
transformations at Alice and Bob sites and the corresponding value of $\mathcal{F}_{opt}$ as a function of the symplectic invariants of the
shared CV entangled state when one restricts to local TGCP maps.

The starting point of our investigation is the formula for the fidelity that we derived in the previous chapter
\begin{equation}
\mathcal F=\frac{2}{\sqrt{\det{(2 V_{in}+N)}}},\label{fid}
\end{equation}
where
\begin{equation} \label{noise}
N=ZAZ+ZC+C^T Z+B,
\end{equation}
with $Z=\mathrm{diag}(1,-1)$. 
As we have already said, we shall restrict to the case of input coherent states, $V_{in}=I$, so that Eq.~(\ref{fid}) reduces to
\begin{equation}
\mathcal F=\frac{2}{\sqrt{4+2 \textrm{Tr} N+\det N}}\label{fidco}.
\end{equation}
The maximization of the teleportation fidelity over all possible Gaussian LOCC strategies therefore means to determine the optimal local transformation of matrices $A$, $B$ and $C$ which makes $N$ as small as
possible.

As showed in \cite{marian}, using an unbalanced beam splitter is equivalent, for the teleportation protocol, to a squeezing operation by Alice.
Therefore the optimization over all Alice and Bob local operations includes also any eventual modification of the beam splitter used for the
joint homodyne measurement.

\section{Upper and lower bounds for the fidelity of teleportation}

In this section we prove Eq.~(\ref{lowupp}), i.e., the upper and lower bounds for the fidelity of teleportation for input coherent states. An
important preliminary result enabling us to derive the two bounds is the fact that the optimal noise matrix $N$ is very simple: in fact, the
maximum teleportation fidelity is obtained when $N$ is proportional to the $2\times 2$ identity matrix $I$. More precisely, we have the
following
\begin{lem}[\textbf{Optimal noise matrix}]
If $\omega_{opt}$ is an optimal local GCP map which gives the maximum of the fidelity $\mathcal F_{max}$ for the teleportation of a coherent
state, then the resulting noise matrix is a multiple of the identity, that is $N_{opt}=2n_{opt}I$. \label{lemN}
\end{lem}
\begin{proof}
First of all we observe that for a $2 \times 2$, symmetric and positive semidefinite matrix like $N$, the condition $N=2n_{opt}I$ is equivalent
to $\textrm{Tr} N=2\sqrt{\det N}$. Therefore we have to show that $\omega_{opt}$ is such that $\textrm{Tr}N_{opt}=2\sqrt{\det N_{opt}}$. We do
this by reductio ad absurdum supposing that $\omega_{opt}$ gives a noise matrix with $\textrm{Tr}N_{opt}>2\sqrt{\det N_{opt}}$. However, within
the class of local GCP maps, there exists a subclass of local symplectic (i.e., unitary Gaussian) maps realized by a generic symplectic $S_{b}$
on Bob mode, and the associated symplectic map $ S_{a}=Z S_{b}Z$ on Alice mode, which act as an \textit{effective} symplectic transformation on
$N_{opt}$, $N_{opt}'=S_b N_{opt} S_b^T$ (see Eq.~(\ref{noise})). We can always choose $S_b$ such that $N_{opt}'=\sqrt{\det N_{opt}}I$, for which
$\textrm{Tr}{N_{opt}'}=2 \sqrt{\det N_{opt}}<\textrm{Tr}{N_{opt}}$ while $\det N_{opt}'=\det N_{opt}$. However, we see from Eq.~(\ref{fidco}),
that this local symplectic operation \textit{increases} the teleportation fidelity, but this is absurd because we assumed from the beginning
that $N_{opt}$ is optimal.
\end{proof}
From this lemma and Eq.~(\ref{fidco}) we can therefore rewrite the optimal fidelity of teleportation in terms of the single positive parameter
$n_{opt}=\sqrt{\det N_{opt}}/2$ as
\begin{equation}
 \mathcal F_{opt} = \frac{1}{1+n_{opt}}.\label{fidn}
\end{equation}
We can now derive an upper bound for $\mathcal F_{opt}$ for a given entanglement of the state shared by Alice and Bob. We quantify such
entanglement in terms of the lowest partially transposed (PT) symplectic eigenvalue $\nu$ (\ref{nuformula}).
\begin{thm}[\textbf{Upper bound}] For a given Gaussian bipartite state shared by Alice and Bob, with lowest PT symplectic eigenvalue
$\nu$, the fidelity of the teleportation of a coherent state is limited from above by
\begin{equation}
\mathcal F_{opt}\le \frac{1}{1+ \nu}. \label{ubound}
\end{equation}
\end{thm}
\begin{proof}
Let us suppose that we can achieve a larger fidelity $\mathcal F=1/(1+n_{opt})$ with $0<n_{opt}<\nu$. Alice can in principle have at her
disposal a two-mode squeezed state, with the usual correlation matrix
\begin{equation}
W=
\left(
\begin{array}{cc}
I\cosh r & -Z\sinh r \\
-Z\sinh r & I\cosh r \\
\end{array}
\right),
\end{equation}
($r$ is the squeezing parameter) and use this two-mode squeezed state, together with the bipartite state shared with Bob already optimized over
all local GCP maps, to implement a CV entanglement swapping protocol \cite{swap}. In fact, by mixing at a balanced beam splitter her mode of the
bipartite state shared with Bob and one part of the two-mode squeezed state, and performing homodyne measurements at the output, Bob mode gets
entangled with the remaining part of the two-mode squeezed state in Alice hands. Since the noise added to the teleported state is
$N_{opt}=2n_{opt}I$, it is straightforward to see that the two remaining modes are then described by the following CM
\begin{equation}
W_{swap}=
\left(
\begin{array}{cc}
I\cosh r & -Z\sinh r \\
-Z\sinh r & I[2n_{opt}+\cosh r] \\
\end{array}
\right).
\end{equation}
In other words, before entanglement swapping, Alice and Bob shared an entangled state with CM $V$ and entanglement characterized by $\nu$; after
entanglement swapping, they share a state with CM $W_{swap}$. In the limit of infinite squeezing the lowest PT symplectic eigenvalue of
$W_{swap}$ tends to $n_{opt}$, i.e., $\lim_{r\rightarrow \infty} \nu_{swap}=n_{opt}$. Since we supposed $n_{opt}< \nu$, this means that for a
sufficiently large squeezing parameter $r$, $\nu_{swap}<\nu$, i.e., Alice and Bob have increased their entanglement. However this is impossible
because we have employed only local operations. Therefore it must be $n_{opt} \geq \nu$.
\end{proof}
We complete the characterization of the optimal fidelity of teleportation in terms of the entanglement shared by the two distant parties by
providing also a lower bound for $\mathcal F_{opt}$, proving in this way the result of Eq.~(\ref{lowupp}).
\begin{thm}[\textbf{Lower bound}] For a given Gaussian bipartite state shared by Alice and Bob, with lowest PT symplectic eigenvalue
$\nu$, the fidelity of the teleportation of a coherent state is limited from below by
\begin{equation}
\mathcal F_{opt} \ge \frac{1+\nu}{1+3\nu}.\label{lbound}
\end{equation}
\end{thm}
\begin{proof}
From the definition of symplectic eigenvalue, one has that a $4 \times 4$ symplectic matrix $S$ exists which diagonalizes $\Lambda V \Lambda$
($\Lambda=\textrm{diag}(Z,I)$), i.e., the PT matrix of the CM $V$. This means $S \Lambda V \Lambda S^T=\textrm{diag}(\nu,\nu,\mu,\mu)$, where
$\mu$ is the largest PT symplectic eigenvalue. By writing $S$ in $2\times 2$ block form
\begin{equation}
S=
\left(
\begin{array}{cc}
 W_a & W_b \\
 W_c & W_d \\
\end{array}
\right),
\end{equation}
and rewriting the diagonalization condition for the upper $2 \times 2$ block only, one gets the following condition
\begin{equation}\label{relation}
W_a Z A Z W_a^T + W_a Z C W_b^T +W_b C^T Z W_a^T + W_b B W_b^T=\nu I.
\end{equation}
The symplectic transformation $S$ transforms the vector of quadratures $\hat{\xi}$ into $\hat{\xi}'=(x_a',y_a',x_b',y_b')^T=S \hat{\xi}$ and the
PT vector into $\hat{\xi}''=(x_a'',y_a'',x_b'',y_b'')^T=S\Lambda \hat{\xi}$. One has $[x_a',y_a']=i$, because commutation relation are preserved
by $S$, implying
\begin{equation} \label{det1}
\det W_a+\det W_b=1.
\end{equation}
The commutation relation is instead not preserved for the PT transformed quadratures, and introducing a real parameter $\epsilon$ such that
$[x_a'',y_a'']=i\epsilon$, we get another condition for the two upper blocks of $S$,
\begin{equation}
-\det W_a+\det W_b= \epsilon,
\end{equation}
which together with Eq.~(\ref{det1}), gives the parametrization
\begin{equation} \label{detsol}
 \det W_a=(1 -\epsilon)/2,\qquad \det W_b=(1 +\epsilon)/2.
\end{equation}
Now, since $\Delta x_a''^2= \Delta y_a''^2=\nu/2$, the Heisenberg uncertainty principle imposes that $|\epsilon|\le\nu$ and in particular for
every entangled state we have $|\epsilon|\le\nu <1$. This latter condition, together with Eq.~(\ref{detsol}), suggests an alternative
parametrization in terms of the angle $\theta=\arctan\sqrt{(1-\epsilon)/(1+\epsilon)}$ ($0 < \theta < \pi/2$),
\begin{equation} \label{det theta}
 \sqrt{\det W_a}=\sin\theta,\quad \sqrt{\det W_b}=\cos\theta .
\end{equation}
The $2 \times 2$ matrices $W_a$ and $W_b$ and the parameter $\theta$ allow to construct an appropriate local map which will lead us to derive a
lower bound for the fidelity. This local map is a TGCP map which, at the level of CM, acts as \cite{TGCP2}
\begin{equation}
V \rightarrow V'=SVS^T+G, \label{genTGCP}
\end{equation}
with $S$ and $G$ satisfying \begin{equation}\label{condiTGCP}
 G+i\mathcal J-iS\mathcal J S^T \geq 0.
\end{equation}
If the TGCP map is local, then $S=S_a
\oplus S_b$ and $G= G_a \oplus G_b$, with $G_k+i J-iS_k J S_k^T \geq 0$ ($k=a,b$).

The desired local TGCP map $\omega_\theta$ is defined in terms of $S_a$, $S_b$, $G_a$ and $G_b$ in the following way
\begin{subequations}
\label{tgcpmap}
\begin{eqnarray}
S_a & = & \left\{\begin{array}{cc}
 ZW_aZ \left[\cos\theta\right]^{-1} & 0 < \theta \leq \pi/4 \\
 ZW_aZ\left[\sin\theta\right]^{-1} & \pi/4 \leq \theta < \pi/2 , \\
\end{array}
\right.\\
S_b & = & \left\{\begin{array}{cc}
 W_b\left[\cos\theta\right]^{-1} & 0 < \theta \leq \pi/4 \\
 W_b\left[\sin\theta\right]^{-1} & \pi/4 \leq \theta < \pi/2,  \\
\end{array}
\right.\\
G_a & = & \left\{\begin{array}{cc}
 \left[1-\tan^2\theta\right]I & 0 < \theta \leq \pi/4 \\
 0 & \pi/4 \leq \theta < \pi/2 , \\
\end{array}
\right.\\
G_b & = & \left\{\begin{array}{cc}
 0 & 0 < \theta \leq \pi/4 \\
 \left[1-\cot^2\theta\right]I & \pi/4 \leq \theta < \pi/2 . \\
\end{array}
\right. 
\end{eqnarray}
\end{subequations}
By applying Eqs.~(\ref{noise}), (\ref{relation}) and (\ref{genTGCP}), one can see that this local TGCP map transforms the noise matrix $N$ into
a final matrix proportional to the identity, given by
\begin{eqnarray}
 N&=& [\nu/\cos^2\theta+ 1-\tan^2\theta]I,\quad 0 < \theta \leq \pi/4,\label{noiseb} \\
 N&=&[\nu/\sin^2\theta+1-\cot^2\theta]I,\quad \pi/4 \leq \theta < \pi/2.\label{noisea}
\end{eqnarray}
It is however convenient to come back to the parametrization in terms of $\epsilon$, which allows to express the final $N$ in a unique way, for
$0 < \theta < \pi/2$. In fact, from Eqs.~(\ref{noiseb})-(\ref{noisea}), one gets
\begin{equation}
N=2 \frac{\nu+|\epsilon|}{1+|\epsilon|}I,
\end{equation}
which, inserted into Eq.~(\ref{fidn}), yields
\begin{equation}
 \mathcal F=\frac{1+|\epsilon|}{1+\nu+2|\epsilon|}.\label{fideps}
\end{equation}
From the condition imposed by the Heisenberg uncertainty principle $0 \leq |\epsilon|\leq \nu$, we see that the fidelity is minimum when
$|\epsilon|=\nu$, so that we get the following lower bound
\begin{equation}
 \mathcal F_{opt}\ge \frac{1+\nu}{1+3 \nu}.
\end{equation}
\end{proof}
Theorems 1 and 2 provide a very useful characterization of the optimal fidelity which can be achieved with Gaussian local operations at Alice
and Bob site. In fact, the bounds are quite tight because the region between the upper and the lower bound is quite small (see
Fig.~\ref{bounds}). Therefore, by simply computing the lowest PT symplectic eigenvalue of the CM of the shared state and using the bounds, one
gets a good estimate of the maximum fidelity that can be obtained with appropriate local operations. In fact, the error provided by the bounds
is never larger than $0.086$ (see Fig.~\ref{error}).

\begin{figure}[tbh]
\centerline{\includegraphics[width=0.7\textwidth]{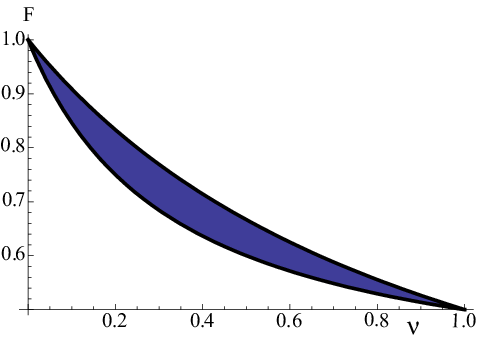}} \caption{Plot of the upper and lower bounds (Eq.~(\ref{ubound}) and
(\ref{lbound}) respectively) for the fidelity of teleportation of coherent states. The blue region is the allowed region in the $(\mathcal
F,\nu)$ plane.} \label{bounds}
\end{figure}

\begin{figure}[tbh]
\centerline{\includegraphics[width=0.7\textwidth]{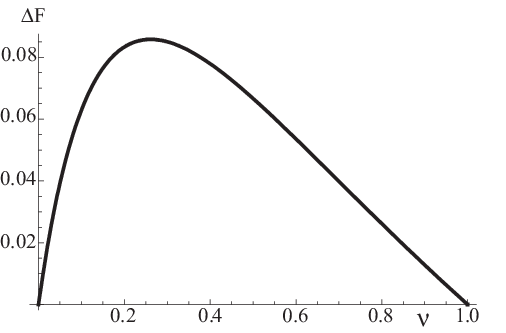}} \caption{Plot of the distance between the upper and lower bounds for the
teleportation fidelity versus the allowed values of $\nu$ for an entangled state between Alice and Bob. We see that the error $\Delta \mathcal
F$ under which we can estimate the maximum fidelity is less than $0.086$.} \label{error}
\end{figure}

\begin{cor}[\textbf{Upper bound achieved in the symmetric case}]
The upper bound $\mathcal F_{opt}=1/(1+\nu)$ is achieved iff the bipartite Gaussian state shared by Alice and Bob is symmetric. The optimal
local transformation in the symmetric case is a local symplectic map.
\end{cor}
\begin{proof}
The ``if'' part of the theorem directly follows as a special case of the preceding proof. If the Gaussian state shared by Alice and Bob is
symmetric, it is $\det W_a=\det W_b$, implying $\epsilon=0$. Then, Eq.~(\ref{fideps}) shows that in this case the fidelity reaches the upper
bound, $\mathcal F_{opt}=1/(1+\nu)$. Moreover in this case $\theta=\pi/4$ and the local TGCP map of Eqs.~(\ref{tgcpmap}) is optimal and it is a
symplectic one, with $S_a=ZW_aZ \sqrt{2}$, $S_b=W_b\sqrt{2}$, $G_a=G_b=0$. The ``only if'' part instead can be easily proved by using the result
of Theorem 3 about the CM of the optimized bipartite state shown in the following section. The proof is given in the Appendix.
\end{proof}
This latter corollary provides the generalization of the result of Ref.~\cite{ade-illu05}, which obtained the same relation between optimal
fidelity and $\nu$ but by considering only a special class of symmetric Gaussian bipartite state for Alice and Bob, obtained by mixing at a beam
splitter two single-mode thermal squeezed states.

\section{Determination of the optimal local map}

We have derived a lower bound for the optimal fidelity of teleportation of coherent states, by explicitly constructing the family of local TGCP
maps $\omega_{\theta}$ of Eq.~(\ref{tgcpmap}), which transform Alice and Bob shared state so that the corresponding fidelity of teleportation is
given by Eq.~(\ref{fideps}), interpolating between the lower and upper bound of Theorems 1 and 2 by varying $\epsilon = \cos 2\theta$. The map
$\omega_{\theta}$ is symplectic only for $\theta = \pi/4$ and in this case it is the local optimal map for a symmetric shared state, since it
reaches the upper bound of Theorem 1 (see Corollary 1). When $\theta \neq \pi/4$, $\omega_{\theta}$ is a noisy (i.e., non unitary) map, and it
is not optimal in general, because one cannot exclude that different LOCC strategies by Alice and Bob may yield a better noise matrix $N$ and
therefore a larger value of the teleportation fidelity. As discussed in the introduction, here we shall study the optimization of the
teleportation by restricting to GCP maps, which preserves the Gaussian nature of the bipartite state initially shared by Alice and Bob.

In this section we shall derive two results: i) the general form of the final CM of the bipartite Gaussian state after the optimization over all
local GCP maps; ii) the optimal local TGCP map, i.e., the local TGCP map which maximizes the teleportation fidelity when one restricts to TGCP
maps only, excluding in this way measurements of ancillary modes.

Ref.~\cite{fiu02} already provided the analytical procedure for the determination of all the parameters of the optimal TGCP map. Here, by
further elaborating the approach of Ref.~\cite{fiu02}, we will show that the optimal TGCP map can always be written in a simple form, as a local
symplectic operation eventually followed by a single mode \textit{attenuation} \cite{Cav82}, either at Alice or Bob site.

\subsection{Standard form of the optimal correlation matrix}

In this subsection we show that, even if we do not know the specific form of the optimal GCP map, one can always characterize it indirectly by
determining the general form of its outcome, i.e., the general form of the CM of the final Gaussian state shared by Alice and Bob after the
maximization. We begin with the following lemma.
\begin{lem}[\textbf{Standard form III}] The correlation matrix $V$ of every bipartite Gaussian state can be transformed by
local symplectic operations into the following normal form
\begin{equation}
V_{1}= \left(
\begin{array}{cccc}
n_1  &  &-d_1     & \\
     &n_2   &     &d_2  \\
-d_1     &  &m_1&       \\
     &d_2   &     &m_2
\end{array}
\right),\label{Vsigma}
\end{equation}
where all the coefficients are positive and satisfy the following constraints
\begin{equation}
 n_1-n_2=m_1-m_2=d_1-d_2=\lambda, \qquad \lambda \in \mathbb R \label{constr}.
\end{equation}
That is:
\begin{equation}
V_{1}= \left(
\begin{array}{cccc}
n+\lambda&  &-d-\lambda&    \\
     &n     &     &d    \\
-d-\lambda&     &m+\lambda&     \\
     &d &     &m
\end{array}
\right).\label{Vsigma2}
\end{equation}
\end{lem}
\begin{proof}It is well known that it is possible to transform every $V$ of an entangled state in the usual normal form
(standard form I) \cite{Duan,Simon}
\begin{equation}
V_{N}= \left(
\begin{array}{cccc}
a    &  &-c_1     & \\
     &a     &     &c_2  \\
-c_1     &  &b&         \\
     &c_2   &     &b
\end{array}
\right),\label{Vstandard}
\end{equation}
where all the coefficients are positive. Now we perform a local symplectic operation composed by two local squeezing operations,
$S_a=\textrm{diag}(\sqrt{r_a},1/\sqrt{r_a})$ and $S_b=\textrm{diag}(\sqrt{r_b},1/\sqrt{r_b})$. We impose the first two conditions
$n_1-n_2=m_1-m_2=\lambda$,
\begin{equation}
a(r_a-r_a^{-1})=b(r_b-r_b^{-1})=\lambda,\label{nm}
\end{equation}
which solved for positive $r_a$ and $r_b$ give
\begin{eqnarray}
r_a(\lambda)&=&\lambda/(2a)+\sqrt{1+(\lambda/2a)^2}, \label{ra}\\
r_b(\lambda)&=&\lambda/(2b)+\sqrt{1+(\lambda/2b)^2}.\label{rb}
\end{eqnarray}
Now we impose the last constraint $d_1-d_2=\lambda$, that is
\begin{equation}
c_1 \sqrt{r_a(\lambda)r_b(\lambda)}-c_2/ \sqrt{r_a(\lambda)r_b(\lambda)}=\lambda.\label{d}
\end{equation}
Our lemma is proved if there is at least one solution $\lambda$ of Eq.~(\ref{d}). Since $\lambda=0$ is the trivial solution when $c_1=c_2$ and
$V_N=V_1$, we can exclude this particular case and divide Eq.~(\ref{d}) by $\lambda$. Therefore we have to show that the equation
\begin{equation}
f(\lambda)=\frac{1}{\lambda}[c_1 \sqrt{r_a(\lambda)r_b(\lambda)}-c_2/ \sqrt{r_a(\lambda)r_b(\lambda)}]=1\label{equation}
\end{equation}
admits at least one real solution. If $|\lambda|\gg a$ and $|\lambda|\gg b$, we can power expand the square roots in Eqs.~(\ref{ra})-(\ref{rb})
so that we easily find the following limits for $f(\lambda)$:
\begin{eqnarray}
\lim_{\lambda\rightarrow\infty}f(\lambda)&=&c_1/\sqrt{ab}\le1,\label{ine1}\\
\lim_{\lambda\rightarrow-\infty}f(\lambda)&=& c_2/{\sqrt{ab}}\le1,\label{ine2} \\
\lim_{\lambda\rightarrow 0^\pm}f(\lambda)&=& \pm \textrm{sign}(c_1-c_2)\infty.\label{lim0}
\end{eqnarray}
The inequalities in Eqs.~(\ref{ine1})-(\ref{ine2}) follow from the Cauchy-Schwartz inequality $\langle x_1^2\rangle\langle
x_2^2\rangle\ge\langle x_1 x_2\rangle^2$ applied to the quadrature operators of the two modes. Given the three limits (\ref{ine1}), (\ref{ine2})
and (\ref{lim0}), since $f(\lambda)$ is continuous everywhere except at the origin, at least one solution of Eq.~(\ref{equation}) exists.
Moreover this solution has the same sign of $c_1-c_2$.\end{proof} We have defined the standard form of lemma 2 as standard form III because it
is very similar to the standard form II defined in Ref.~\cite{Duan} for the determination of a necessary and sufficient entanglement criterion
for bipartite Gaussian states. In particular the two standard forms coincide in the special case of a symmetric bipartite state ($n_1=m_1$ and
$n_2=m_2$ or equivalently $n=m$).

\begin{thm}[\textbf{Form of the CM of the optimized bipartite state}]
The optimal GCP map $\omega_{opt}$ maximizing the teleportation fidelity is such that the CM of the transformed bipartite state is in the
standard form III $V_1 $ defined by Eqs.~(\ref{Vsigma})-(\ref{Vsigma2}).
\end{thm}
\begin{proof}
By means of local symplectic operations, we can always put the CM of the bipartite state of Alice and Bob in the form of Eq.~(\ref{Vsigma}), but
without the constraints of Eq.~(\ref{constr}). We first restrict to local symplectic operations and show that the optimal local symplectic
operation always transforms to a state with a CM satisfying the constraints of Eq.~(\ref{constr}). Since the CM is tridiagonal, then any
possible optimal map must be a squeezing transformations of the two modes given by $S_a=\textrm{diag}(r_a,r_a^{-1})$ and
$S_b=\textrm{diag}(r_b,r_b^{-1})$ \cite{fiu02}. Let us define
\begin{eqnarray}
\alpha(r_a,r_b)&=& r_a^2 n_1 -2d_1 r_a r_b+r_b^2 m_1, \\
\beta(r_a,r_b)&=&r_a^{-2}n_2-2d_2 (r_a r_b)^{-1}+r_b^{-2} m_2,
\end{eqnarray}
so that the noise matrix of Eq.~(\ref{noise}) is equal to $N=\textrm{diag}(\alpha,\beta)$. The optimal map must minimize $\det N=\alpha\beta$,
and therefore we impose that
\begin{equation}
 \nabla \alpha(r_a,r_b) \beta(r_a,r_b) =0,\label{ab}
\end{equation}
where $\nabla=(\partial_{r_a},\partial_{r_b})$. Due to Lemma 1, we must also have that $\alpha(r_a,r_b)=\beta(r_a,r_b)\ne 0$, and therefore
Eq.~(\ref{ab}) reduces to
\begin{equation}
 \nabla [\alpha(r_a,r_b)+ \beta(r_a,r_b)] =0\label{a+b}.
\end{equation}
The CM of the transformed state is the optimized one iff the optimal local symplectic operation is the identity map, that is, if
$\alpha(1,1)=\beta(1,1)$ and
\begin{equation}
 \nabla [\alpha(r_a,r_b)+ \beta(r_a,r_b)]\Big|_{r_a=r_b=1} =0.
\end{equation}
It is easy to check that these conditions are satisfied iff
\begin{equation}
 n_1-n_2=m_1-m_2=d_1-d_2,
\end{equation}
which are exactly the constraints of Eq.~(\ref{constr}). The theorem is proved if we show that the normal form $V_1$ of Eq.~(\ref{Vsigma2}) is
actually kept also if one maximizes over the broader class of GCP maps. In fact, if by reductio ad absurdum, we assume that an optimal
(non-symplectic) GCP map exists leading to a CM not satisfying the constraints of Eq.~(\ref{constr}), we could always apply a further symplectic
map which, by repeating the maximization above, would transform to a state with a CM satisfying the constraints (\ref{constr}) and yielding a
larger teleportation fidelity. But this is impossible, because it contradicts the initial assumption of starting from the optimal bipartite
state.
\end{proof}
In other words, the CM of the state shared by Alice and Bob after the maximization of the teleportation fidelity is the one with the standard
form III of Eq.~(\ref{Vsigma2}) because it is the unique CM for which the optimal map is the identity operation on both Alice and Bob site.

\subsection{Optimal trace-preserving Gaussian CP map}

In the former subsection we have determined the form of the CM of the optimized state of Alice and Bob, without determining however which is the
local GCP map which maximizes the teleportation fidelity. Here we find this optimal map, restricting however to the smaller class of
\textit{trace-preserving} GCP maps. The case of Gaussian maps including Gaussian measurements on ancillas will be afforded elsewhere.

Ref.~\cite{lin00} has introduced the notion of \textit{minimal noise} TCGP maps, as the extremal solution of the condition of
Eq.~(\ref{condiTGCP}). These maps are the ones that, for a given matrix $S$, possess the ``smallest'' positive matrix $G$ realizing a CP map. It
is easy to check that a minimal noise TGCP map satisfies the relation $\det G = \left(1-\det S\right)^2$. An example of minimal noise TGCP map
is an \textit{attenuation} \cite{Cav82}, i.e., the transmission of a single boson mode through a beam splitter with transmissivity $\tau$ ($0\leq
\tau\leq 1$), such that
\begin{equation} a \rightarrow \tau a + \sqrt{1-\tau^2} a_V, \label{atte}
\end{equation} where $a=\left(\hat{x}+i\hat{p}\right)/\sqrt{2}$ is the annihilation operator of the mode, and $a_V$ that of the vacuum mode
entering the unused port of the beam splitter.

It is evident that the TGCP map maximizing the teleportation fidelity has to be a minimal noise TGCP map \cite{fiu02}. We prove now a useful
decomposition theorem.
\begin{thm}[\textbf{Decomposition of TGCP maps}]
A minimal noise TGCP map $\omega $ on a single mode system with $\det G\le 1$ can always be decomposed into a symplectic transformation
$\sigma_1$, followed by an attenuation $\tau$ and by a second symplectic transformation $\sigma_2$, that is,
\begin{equation}\label{teo}
\omega =\sigma_2 \circ \tau \circ \sigma_1.
\end{equation}
Therefore a local minimal noise TGCP map on a bipartite CV system can always be decomposed into a local symplectic map, followed by the tensor
product of two local attenuations and by a second local symplectic map. \end{thm}
\begin{proof}
We consider a generic minimal noise TGCP map such that $V \rightarrow V'=SVS^T+G$ for a generic CM $V$, with $\det G = \left(1-\det S\right)^2$.
$G$ is a positive symmetric matrix, and therefore a symplectic matrix $T_{2}$ exists such that $G=T_{2}T_{2}^T (1-s)$, where $s=\det S$. We then
define the symplectic matrix $T_{1}=\frac{1}{\sqrt{s}} T_{2}^{-1} S$, and we also consider an attenuation map with transmissivity $\sqrt{s}$. If
we now first apply the symplectic map defined by $T_1$, then the attenuation map and finally the second symplectic map defined by $T_2$, by
using the relations $T_{2}\sqrt{s}T_{1}=S$ and $G=T_{2}T_{2}^T (1-s)$, one can check that the composition of the three maps reproduces the given
TGCP map.
\end{proof}
\begin{cor} \label{cordec}
A minimal noise TGCP map $\omega $ on a single mode system with $G$ proportional to the identity matrix, i.e., $G=(1-s) I$ $(0 \leq s \leq 1$)
can always be decomposed into a symplectic transformation $\sigma_1$, followed by an attenuation $\tau$,
\begin{equation}\label{teo2}
\omega =\tau \circ \sigma_1.
\end{equation}
Therefore a local minimal noise TGCP map on a bipartite CV system with $G_i=(1-s_i) I$ $(0 \leq s_i \leq 1$, $i=a,b$) can always be decomposed
into a local symplectic map, followed by the tensor product of two local attenuations.
\end{cor}
\begin{proof}It is sufficient to repeat the former proof and consider that, since $G=(1-s) I$, $T_2=I$ and therefore the second symplectic map is
the identity operation.
\end{proof}
This latter case is of interest because the optimal TGCP map must have in fact the property $G_i=(1-s_i) I$, $i=a,b$. To show this we first
simplify the scenario by exploiting the results of Ref.~\cite{fiu02}, which provides the general analytical procedure to derive the optimal
local TGCP map. In fact, Ref.~\cite{fiu02} shows that when the optimal local TGCP map is a minimal noise, non-symplectic one, it can be
performed on one site only, i.e., either on Alice or on Bob alone. Suppose that the non-symplectic map is performed on Bob site; it is
straightforward to see that, under a generic local TGCP map $V \rightarrow V'=SVS^T+G$, with $S=I \oplus S_b$ and $G=I \oplus G_b$, the noise
matrix $N$ transforms according to $N \rightarrow N'=\Gamma+G_b$ where
\begin{equation}\label{transfoN}
\Gamma=Z A  Z  + Z C S_b^T +S_b C^T Z+S_b B S_b^T.
\end{equation}
Ref.~\cite{fiu02} shows that the optimal map is such that $\Gamma \propto G_b$, and since Lemma 1 shows that the optimal $N$ is proportional to
the identity, this implies that $G_b$ must be proportional to the identity. Therefore Corollary 2 leads us to conclude that \textit{the optimal
local TGCP map is either a local symplectic map, or a local symplectic map followed by an attenuation by a beam splitter, placed either on Alice
or on Bob mode}.

Theorem 4 and Corollary 2 therefore provide a very simple and clear description of the TGCP map which maximizes the teleportation fidelity,
which is not evident in the treatment of Ref.~\cite{fiu02}. We can further characterize the optimal local TGCP map by determining: i) the form
of the first local symplectic map; ii) the conditions under which the optimal local operation is noisy, i.e., when one has also to add a beam
splitter with appropriate transmissivity on Alice or Bob mode in order to maximize the teleportation fidelity. In order to do that we first need
a further lemma, similar to lemma 2.

\begin{lem} For any given positive real parameter $\eta$, the correlation matrix $V$ of every bipartite
Gaussian state can be transformed by local symplectic operations into the following normal form
\begin{equation}
V_{\eta}= \left(
\begin{array}{cccc}
n_1  &  &-d_1     & \\
     &n_2   &     &d_2  \\
-d_1     &  &m_1&       \\
     &d_2   &     &m_2
\end{array}
\right),\label{Veta}
\end{equation}
where all the coefficients are positive and satisfy the following constraint
\begin{equation}
 n_1-n_2=\eta(d_1-d_2)=\eta^2(m_1-m_2)=\lambda, \qquad \lambda \in \mathbb R \label{consteta}.
\end{equation}
That is, we have a family of normal forms depending on the parameter $\eta$,
\begin{equation}
V_{\eta}= \left(
\begin{array}{cccc}
n+\lambda   &   &-d-\lambda/\eta    &   \\
        &n  &          &d   \\
-d-\lambda/\eta  &  &m+\lambda/\eta^2  &    \\
        &d  &          &m
\end{array}
\right).\label{Veta2}
\end{equation}
\end{lem}
\begin{proof}
With the same procedure used in the proof of Lemma 2, we arrive to an equation similar to (\ref{equation}), while the corresponding three limits
are exactly the same of (\ref{ine1}), (\ref{ine2}) and (\ref{lim0}), since the factor $\eta$ cancels out. As a consequence the continuity
argument is valid also in this case, and therefore, for every fixed parameter $\eta$, one can find a transformation which puts the CM in the
normal form $V_\eta$.
\end{proof}
We notice two facts that will be useful for the next theorem: i) the optimal CM standard form of Theorem 3, $V_1$,  belongs to the class of
normal forms $V_{\eta}$, since it is obtained for $\eta =1$; ii) when $0 < \eta < 1$, the state with CM $V_{\eta}$ is transformed into the
Gaussian state with CM equal to $V_1$ when a beam splitter with transmissivity $\eta$ is put on Bob mode. We then arrive at the theorem about
the optimal TGCP map.

\begin{thm}
The optimal local TGCP map maximizing the fidelity of teleportation of coherent states can always be decomposed into a local symplectic map,
eventually followed by an attenuation either on Alice or on Bob mode. The first local symplectic map is the one transforming the CM of the
Gaussian state shared by Alice and Bob into one particular normal form of the family $V_\eta$ defined in Eqs.~(\ref{Veta})-(\ref{consteta}),
with $0 < \eta \leq 1$. One has to add the attenuation on one of the two modes for realizing the optimal TGCP map if there is a value of $\eta$,
let us say $\eta=\tau$, such that the coefficients of $V_\tau$ satisfy the relations
\begin{equation}\label{tau}
\tau=\frac{d}{m-1},\qquad \tau<1.
\end{equation}
If instead the condition of Eq.~(\ref{tau}) is never satisfied in the interval $\eta \in [0,1)$, the optimal TGCP map is formed only by the
local symplectic map transforming the CM into the normal form $V_1$ (i.e., $\eta=1$ and no attenuation is required).
\end{thm}
\begin{proof}
Using Lemma 3, we can always apply a local symplectic map which transform the CM into one of the form of the family of Eq.~(\ref{Veta2}), with
$\eta=\tau$. We then apply an attenuation on Bob mode with transmissivity $\tau$ and then try to find the maximum fidelity as in the proof of
Theorem 3. We have that the two diagonal elements of the noise matrix $N$ now read
\begin{equation}
\alpha(\tau)= \beta(\tau)= n-2 \tau d +\tau^2 m+1-\tau^2,
\end{equation}
The optimal map must minimize $\det N=\alpha^2$, and therefore we impose
\begin{equation}
 \frac{d \alpha(\tau)}{d \tau}=0,\label{abeta}
\end{equation}
which is satisfied iff $\tau=d/(m-1)$. If $0 < \tau < 1$, this map composed by the local symplectic map and the attenuation is the optimal map.
If instead for any $\tau \in [0,1]$, the condition of Eq.~(\ref{tau}) is not satisfied, there is no critical point in this interval and
therefore the optimal map is just the symplectic transformation to the normal form $V_1$. One has also to check the behavior at the lower
boundary value $\tau=0$, but this is trivial because this means that Bob uses the vacuum to implement the teleportation which is never the
optimal solution if we have an entangled channel. In fact, if Bob uses the vacuum, the channel looses its quantum nature and the maximum of the
fidelity is the classical one $\mathcal F=1/2$, which is below the lower bound for any entangled state given by Eq.~(\ref{lbound}).
\end{proof}
Theorem 5 therefore characterizes in detail the optimal TGCP map, giving in particular the conditions under which this map is noisy, i.e.,
non-symplectic and therefore when teleportation is improved by \textit{increasing the noise and decreasing the entanglement} of the shared state.

Using this latter theorem we can also determine how, from an operational point of view, one can compute the value of the teleportation fidelity
maximized over all TGCP maps, starting from the symplectic invariants of the bipartite Gaussian state initially shared by Alice and Bob. From
the CM of this latter state one can:
\begin{enumerate}
 \item compute the four symplectic invariants $a=\sqrt{\det A}$, $b=\sqrt{\det B}$, $c=\sqrt{|\det C|}$ and $v=\det V$.
\item Knowing the first three invariants of the channel, the elements $n$, $m$, and $d$ of the normal form $V_{\eta}$ can be expressed as
functions of only the two unknown parameters $\lambda$ and $\eta$ as
\begin{eqnarray}
 n(\lambda) &=& -\lambda/2       + \sqrt{a^2+(\lambda/2)^2},  \\
 m(\lambda, \eta) &=& -\lambda/2\eta^2 + \sqrt{b^2+(\lambda/2\eta^2)^2},  \\
 d(\lambda, \eta) &=& -\lambda/2\eta   + \sqrt{c^2+(\lambda/2\eta)^2}.
\end{eqnarray}
\item The two parameters $\lambda$ and $\eta$ can be found solving the following system in the region $0<\eta<1$,
\begin{equation}\label{sys}
\Bigg \{
 \begin{array}{rcl}
  \det V_{\eta}(\lambda,\eta)&=&v \\
  \eta\; [m(\lambda,\eta)-1]&=&d(\lambda,\eta)\\
 \end{array}
\end{equation}
and solving also the first equation in the boundary $\eta=1$,
\begin{equation}
 \det V_1(\lambda)=v,\label{sys1}
\end{equation}
(the two conditions of (\ref{sys}) come from the invariance property of the determinant of the channel and form the maximization condition of
Eq.~(\ref{tau})).

\item We call $(\lambda_i,\eta_i)$ with $i=1,2\dots k$, the union of the solutions of (\ref{sys}) and
(\ref{sys1}) (we have at least one solution because (\ref{sys1}) admits at least a solution). We then compute the candidate fidelities
\begin{eqnarray}
&& F_i=2\left[2+\sqrt{a^2+\lambda_i^2/4}+\sqrt{b^2\eta_i^4+\lambda_i^2/4}\right. \nonumber \\
&& \left.-2\sqrt{c^2\eta_i^2+\lambda_i^2/4}+1-\eta_i^2 \right ]^{-1}, \label{candifide}
\end{eqnarray}
so that the maximum fidelity will be $\mathcal F_{opt}=\max\{F_i\}$.
\end{enumerate}

\section{Appendix}

We now prove the ``only if'' part of Corollary 1, i.e., that if $\mathcal F_{opt}=1/(1+\nu)$ (the upper bound of the optimal fidelity), then the
bipartite Gaussian state shared by Alice and Bob is symmetric.
\begin{proof}
Theorem 3 shows that the final CM after any optimization map must be $V_1$ of Eq.~(\ref{Vsigma2}). Using Lemma 1 and the explicit form of $V_1$,
the hypothesis is equivalent to $n+m-2d=2\nu$, so we can make the substitution $d=(n+m)/2-\nu$, which is just a different parametrization:
$V_1(n,m,d,\lambda)\rightarrow V_1(n,m,\nu,\lambda)$. Now, the condition that $\nu$ is equal to the PT minimum symplectic eigenvalue gives us a
constraint on the parameters of the matrix $V_1$
\begin{equation}
\nu\{V_1(n,m,\nu,\lambda)\}=\nu. \label{nuconstr}
\end{equation}
If the state is symmetric, which means $n=m$, then the condition of Eq.~($\ref{nuconstr}$) is identically satisfied. If the state is non
symmetric, Eq.~(\ref{nuconstr}) is a non-trivial equation that solved for $\lambda$ gives $\bar \lambda=(m-n)^2/8\nu-n-m$. However the
corresponding matrix $V_1=(n,m,\nu,\bar \lambda)$ is not the CM of a physical state; in fact, the characteristic polynomial of $V_1$ can be
written as $P(x)=(c_0+c_1 x + x^2)(g_0+g_1x+x^2)$ where $c_0=-\nu(n+m+\nu)$, but this means that $V_1$ has at least one negative eigenvalue and
therefore it is not positive definite. Therefore $\mathcal F_{opt}=1/(1+\nu)$ is realized only if Alice and Bob state is a Gaussian symmetric
state.
\end{proof}
\end{chapter}

\part[\small{OPTOMECHANICAL DEVICES}]{OPTOMECHANICAL DEVICES}

\chapter{Quantum Optomechanics}\label{optomech}
In this chapter we give some fundamental theoretical tools, useful to study the dynamics of optical and mechanical systems and their mutual interaction. 

There is one inconvenience that is usually neglected in abstract Quantum Information theories, but that we need to take into account when considering real systems; this is \textit{noise}, thermal and quantum noise. Thermal noise, is the same which is studied in classical mechanics, is due to the thermal excitation of the system and vanishes when $T\rightarrow 0$, while quantum noise is due to intrinsic quantum uncertainties of non-commuting operators and cannot be completely eliminated even at zero temperature. Both thermal and quantum noises can be described by different approaches: the \textit{master equation} (Schr\"odinger picture), the \textit{Fokker-Planck equation} (phase space), the \textit{Langevin equations} (Heisenberg picture). 
Here we are going to consider only the latter approach, which is suitable for linearized systems (the ones we are interested in) and moreover it gives a consistent treatment of quantum Brownian motion.

The last section is about the coupling between a mechanical resonator and an optical mode, due to the effect of radiation pressure.

\section{Input-output theory}
The dynamics of an optical cavity mode can be formulated using an input-output theory, in which the noise is seen as an input filed entering in the cavity or, equivalently, as an output field exiting from the cavity. 

\begin{figure}[H]
\center{ \includegraphics[width=0.7\textwidth]{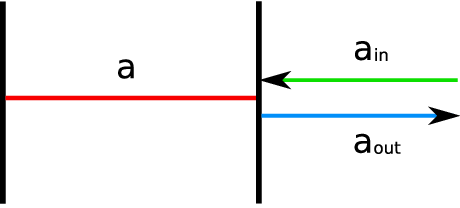}
\caption{Optical cavity mode interacting with input and output fields.}}
\end{figure} 

We derive the Langevin equations following the treatment of Gardiner and Collett \cite{GardColl}. We start form the full Hamiltonian
\begin{equation}
 \hat H= \hat H_c + \hat H_{bath}+ \hat H_{int},
\end{equation}
where $\hat H_c$ is the Hamiltonian of the cavity mode (we neglect constant vacuum terms)
\begin{equation}
 \hat H_c=\hbar \omega_c \hat a^\dag \hat a,
\end{equation}
$H_{bath}$ is the Hamiltonian of the external fields modeled as a bath of harmonic oscillators
\begin{equation}
 \hat H_{bath}= \int_0^\infty d\omega\, \hbar \omega \hat b^\dag (\omega) \hat b(\omega),
\end{equation}
while $H_{int}$ is the linear interaction Hamiltonian in the rotating wave approximation
\begin{equation}\label{HintRWA}
 \hat H_{int}= i \hbar \int_0^\infty d\omega \kappa (\omega) \left[\hat a \hat b^\dag (\omega)-\hat b(\omega) \hat a^\dag\right] ,
\end{equation}
where $\kappa(\omega)$ is the coupling constant between the modes. Consistently with the rotating wave approximation, we can assume that the integral in (\ref{HintRWA}) gives a time average non-zero contribution only for frequencies near the cavity resonance $\omega_c\simeq 10^{15}$, so that we can extend the lower integration limit to $-\infty$,
\begin{equation}\label{HintRWA2}
 \hat H_{int}= i \hbar \int_{-\infty}^\infty d\omega \kappa (\omega)  \left[\hat a \hat b^\dag (\omega)-\hat b(\omega) \hat a^\dag \right].
\end{equation}
The commutation relation are the usual bosonic CCR,
\begin{equation}
 [\hat a, \hat a^{\dag}]=1\qquad [\hat b(\omega), \hat b^\dag(\omega')]=\delta(\omega-\omega'), \label{abCCR}
\end{equation}
from which it is evident that $\hat a$ and $\hat a^\dag$ are dimensionless operators, while $\hat b(\omega)$ and  $\hat b^\dag(\omega)$ have the dimensions of a square root of time.
Let us apply the Heisenberg equation of motion, given in the first chapter (\ref{HeisEq}), to the bath annihilation operator,
\begin{equation}
 \dot {\hat b}(\omega)=\frac{1}{i\hbar}[\hat b(\omega),\hat H]=-i\omega \hat b(\omega)+\kappa(\omega)\hat a,
\end{equation}
which integrated for $t>t_0$ gives
\begin{equation}\label{binitial}
 \hat b(\omega)= e^{-i\omega (t-t_0)}\hat b_0(\omega)+\kappa(\omega)\int_{t_0}^{t}dt' e^{-i\omega(t-t')} \hat a (t'),
\end{equation}
where $\hat b_0(\omega)$ is the bath operator at the initial time $t_0$. 
We can write the same solution for $t<t_1$ 
\begin{equation}\label{bfinal}
 \hat b(\omega)= e^{-i\omega (t-t_1)}\hat b_1(\omega)-\kappa(\omega)\int_{t}^{t_1}dt' e^{-i\omega(t-t')} \hat a (t'),
\end{equation}
where $\hat b_1(\omega)$ is the bath operator at the final time $t_1$.

Now we consider the Heisenberg equation for the cavity annihilation operator,
\begin{equation}
 \dot {\hat a}=\frac{1}{i\hbar}[\hat a,\hat H]=-i\omega_c \hat a -\int_{-\infty}^{\infty}d\omega \kappa(\omega)\hat b(\omega),
\end{equation}
if we substitute the solution for $\hat b(\omega)$ depending on the initial condition (\ref{binitial}), we get
\begin{equation}\label{dotab0pre}
 \dot {\hat a}=-i\omega_c \hat a- \int_{-\infty}^{\infty}d\omega \kappa(\omega) e^{-i\omega(t-t_0)}\hat b_0(\omega)-\int_{-\infty}^{\infty}d\omega \kappa^2(\omega) \int_{t_0}^{t}dt'e^{i\omega(t-t')} \hat a(t').
\end{equation}
Now we make the Markovian assumption of a frequency-independent coupling constant $\kappa(\omega)=\gamma/2\pi$,
so that (\ref{dotab0pre}) becomes
\begin{equation}\label{dotab0delta}
 \dot {\hat a}=-i\omega_c \hat a-  \sqrt{\frac{\gamma}{2\pi}} \int_{-\infty}^{\infty}d\omega e^{-i\omega(t-t_0)}\hat b_0(\omega)-\gamma \int_{t_0}^{t}dt' \delta(t-t')\hat a(t').
\end{equation}
We define the input-output field operators
\begin{eqnarray}
 \hat a_{in}(t)&=&-\frac{1}{2\pi} \int_{-\infty}^{\infty}d\omega e^{-i\omega(t-t_0)}\hat b_0(\omega),\\
 \hat a_{out}(t)&=&\frac{1}{2\pi} \int_{-\infty}^{\infty}d\omega e^{-i\omega(t-t_0)}\hat b_1(\omega), 
\end{eqnarray}
satisfying the CCR
\begin{equation}
 [\hat a_{in}(t),\hat a_{in}^\dag(t')]=[\hat a_{out}(t),\hat a_{out}^\dag(t')]=\delta(t-t').
\end{equation}
If we apply the property of the delta function according to which, if it is centered on one of the integration limits, it gives only half of its contribution
\begin{equation}
\int_{a}^{b}dx \delta(x-b)f(x)=\frac{1}{2}f(b),
\end{equation}
then (\ref{dotab0delta}) becomes
\begin{equation}\label{dotab0}
 \dot {\hat a}=-i\omega_c \hat a- \frac{\gamma}{2}\hat a + \sqrt{\gamma} \hat a_{in},
\end{equation}
while if we repeat the same procedure using the final time solution for $\hat b(\omega)$ (\ref{bfinal}), we get
\begin{equation}\label{dotab1}
 \dot {\hat a}=-i\omega_c \hat a+ \frac{\gamma}{2}\hat a - \sqrt{\gamma} \hat a_{out}.
\end{equation}
The difference between (\ref{dotab1}) and (\ref{dotab0}) give us a useful input-output relation
\begin{equation}\label{InOut}
 \hat a_{in}+\hat a_{out}=\sqrt{\gamma} \hat a.
\end{equation}
Equations (\ref{dotab0}) and (\ref{dotab1}) are two equivalent quantum Langevin equations describing the cavity mode dynamics with a damping rate $\gamma$ and where the noise field is seen as an input or output optical field. 
However, since $\hat a_{out}$ depends on future unknown boundary conditions, Eq. (\ref{dotab1}) is useless form a practical point of view, while (\ref{dotab0}) gives us the possibility to easily calculate the time evolution of the first and the second moments of $\hat a$, thanks to the statistical properties of $\hat a_{in}$.

In fact, if at the initial time $t_0$ the state of the system is factorized like $\hat \rho_0=\hat \rho_c \otimes \hat \rho_n$, where $\hat \rho_c$ is the cavity mode density operator and $\hat \rho_n$ is the thermal state of the bath with 
\begin{equation}
\textrm{tr}\{\hat \rho_n \hat b^\dag(\omega)\hat b(\omega')\}=\bar n(\omega)\delta(\omega-\omega')=\frac{1}{\exp(\frac{\hbar \omega}{K_B T})-1}\delta(\omega-\omega'),
\end{equation}
now, since the cavity mode has a spectrum centered in $\omega=\omega_c$, within the rotating wave approximation we have
\begin{eqnarray}
 \langle \hat a_{in}(t) \rangle&=&0,\label{stoc1}\\ 
\langle \hat a_{in}^\dag(t) \hat a_{in}(t')\rangle&\simeq&\bar n(\omega_c)\delta(t-t'),  \label{stoc2}\\
\langle \hat a_{in}(t) \hat a_{in}^\dag(t')\rangle&\simeq&[\bar n(\omega_c)+1]\delta(t-t'), \label{stoc3}
\end{eqnarray}
moreover, if we deal with optical frequencies, then $\bar n (\omega_c)<<1$ and the mean excitation number can be usually neglected.

We observe that the Langevin equation (\ref{dotab0}) is very general, since it is valid not only if the input state is a thermal filed, but e.g. also if it is a coherent or squeezed state. Obviously, in the these cases, Eq. (\ref{stoc1}-\ref{stoc2}-\ref{stoc3}) will be different.

\section{Quantum Brownian motion}

The quantum Brownian motion of a massive particle in a potential is more intricate than the optical counterpart described in the previous section. The problem is that the mechanical coupling is different form the optical one and moreover we cannot apply the rotating wave approximation; this implies different properties of the noise operator, which has no more the physical interpretation of an optical input field but it represents a driving stochastic force.

Let us start with the Hamiltonian of a particle (lower case operators) in a bath of mechanical oscillators (upper case operators),
\begin{equation}
 \hat H= \frac{\hat p^2}{2m}+\hat V(\hat q) + \int_{0}^{\infty}d\omega \frac{\hat P_\omega'^2}{2 m_\omega}+\frac{k_\omega}{2} (\hat Q_\omega' - \hat q)^2  \label{Hmechanic}
\end{equation}
where the interaction is assumed to depend only on the relative distance between the particle and the oscillators: $\hat Q_\omega' - \hat q$.

The following calculations are very similar but not exactly equal to that given by Giovannetti and Vitali \cite{GIOV01}, because here we start from an equivalent but different Hamiltonian (\ref{Hmechanic}), proposed by Gardiner and Zoller \cite{gard}.

Calculations can be simplified using dimensionless quadratures for the bath operators,
\begin{eqnarray}
 \hat Q_\omega &=& \hat Q_\omega' \sqrt{\frac{k_\omega}{\hbar \omega}}\\
\hat P_\omega&=&\hat P_\omega' \sqrt{\frac{1}{m_\omega \hbar \omega}}
\end{eqnarray}
so that the Hamiltonian becomes,
\begin{equation}
 \hat H= \frac{\hat p^2}{2m}+\hat V(\hat q) + \frac{1}{2}\int_{0}^{\infty}d\omega  \left[\hbar \omega \hat P_\omega^2+ (\sqrt{\hbar \omega}\hat Q_\omega - \kappa_\omega \hat q)^2\right]. \label{Hadim}
\end{equation}
where the only non-vanishing commutation relations are
\begin{eqnarray}
[\hat q,\hat p]&=&i \hbar \;, \\
\left[\hat Q_\omega,\hat P_{\omega'}\right]&=&i\delta(\omega-\omega') \;,
\end{eqnarray}
and $\kappa_\omega=\sqrt{k_\omega}$ is the coupling constant.

The Heisenberg equations for the bath operators are
\begin{eqnarray}
\dot{\hat Q}_\omega&=&\omega \hat P_\omega \;, \\
\dot{\hat P}_\omega&=&-\omega \hat Q_\omega +\kappa_\omega\sqrt{\frac{\omega}{\hbar}} \hat q\;,
\end{eqnarray}
and if we define the annihilation operators $\hat b_\omega=(\hat Q_\omega+i \hat P_\omega)/\sqrt{2}$, we have
\begin{equation}
\dot{\hat b}_\omega=-i\omega \hat b_\omega+i\kappa_\omega\sqrt{\frac{\omega}{2\hbar}}\hat q,
\end{equation}
which integrated gives
\begin{equation}
{\hat b}_\omega(t)=e^{-i\omega(t-t_0)} \hat b_\omega(t_0)+i\kappa_\omega \sqrt{\frac{\omega}{2\hbar}}\int_{t_0}^t dt'e^{-i\omega(t-t')}\hat q(t'), \label{bdotBrow}
\end{equation}
where $t_0$ is an arbitrary initial time.

Now we write the Heisenberg equation for the particle variables
\begin{eqnarray}
\dot{\hat q}&=&\frac{\hat p}{m}, \\
\dot{\hat p}&=&\frac{i}{\hbar}[\hat V,\hat p]-\int_{0}^\infty d\omega k_\omega^2 \hat q+ \int_{0}^\infty d\omega k_\omega \sqrt{\frac{\hbar \omega}{2}}(\hat b_\omega+\hat b_\omega^\dag) \;, \label{pdotBrow1}
\end{eqnarray}
if we substitute (\ref{bdotBrow}) in (\ref{pdotBrow1}) we get
\begin{equation}
 \dot{\hat p}(t)=\frac{i}{\hbar}[\hat V,\hat p(t)]+\hat \xi(t)+\int_{0}^\infty d\omega k_\omega^2\left[- \hat q(t)+ \omega \int_{t_0}^t dt' \sin[\omega(t-t')]q(t')\right] \;, \label{pdotBrow2}
\end{equation}
where we have defined the Brownian force operator
\begin{equation}
\hat \xi(t)=\int_{0}^\infty d\omega k_\omega \sqrt{\frac{\hbar \omega}{2}}\left[e^{-i\omega(t-t_0)}\hat b_\omega(t_o)+e^{i\omega(t-t_0)}\hat b_\omega^\dag(t_o)\right] \;. \label{BrownNoise}
\end{equation}
The time integral in (\ref{pdotBrow2}) can be done by parts, giving
\begin{equation}
 \dot{\hat p}(t)=\frac{i}{\hbar}[\hat V, \hat p(t)]+\hat \xi(t)-\int_{0}^\infty d\omega k_\omega^2\left[\cos[\omega(t-t_0)]\hat q(t_0)+ \int_{t_0}^t dt' \cos[\omega(t-t')]\frac{\hat p(t')}{m}\right] \;. \label{pdotBrow3}
\end{equation}
If we suppose a frequency independent coupling $\kappa_\omega^2=2\gamma/\pi$, using that $\int_0^{\infty}d\omega\cos(\omega t)=\pi \delta(t)$, we get
\begin{equation}
 \dot{\hat p}(t)=\frac{i}{\hbar}[\hat V, \hat p(t)]-\frac{\gamma}{m} \hat p(t)+\hat \xi(t)-2 \gamma \delta(t-t_0) \hat q(t_0) \; \label{pdotBrow4}.
\end{equation}
The last term in (\ref{pdotBrow4}) can be neglected for $t>t_0$, therefore we arrive to a quantum Langevin equation which is very similar to the classical one
\begin{equation}
 \dot{\hat p}(t)=\frac{i}{\hbar}[\hat V, \hat p(t)]-\frac{\gamma}{m} \hat p(t)+\hat \xi(t) \; \label{pdotBrow5}.
\end{equation}
The difference between the QLE (\ref{pdotBrow5}) and the classical countrepart is on the commutation rules and on the second moments of the stochastic force $\hat \xi(t)$ (\ref{BrownNoise}).
In fact it can be shown \cite{GIOV01} that
\begin{equation}
 [\hat \xi(t),\hat \xi(t')]=2i\hbar \gamma \frac{d}{dt}\delta(t-t')
\end{equation}
and, if the particle is in a thermal bath, then
\begin{eqnarray}
 \langle \hat\xi(t) \rangle&=&0 \\
\langle [\hat\xi(t),\hat \xi(t')]_+\rangle&=& \frac{2\hbar\gamma}{\pi} \int_0^\infty d\omega\, \omega \cos[\omega (t-t')]\coth(\frac{\hbar \omega}{2 K_B T}),\label{symBrown} 
\end{eqnarray}
It is notable form ((\ref{symBrown}), that a truly quantum Brownian motion is not a Markov process, however in almost every situation, we can apply the high temperature limit 
\begin{equation}
\hbar \omega \coth(\frac{\hbar \omega}{2 K_B T}) \xrightarrow{T>>\hbar \omega/K_B} 2K_B T \; 
\end{equation}
and we recover a Markov process
\begin{equation}
 \langle [\hat\xi(t),\hat \xi(t')]_+\rangle= 4\gamma K_B T \delta(t-t').\label{symBrownMark}
\end{equation}

\section{Radiation pressure}

J. C. Maxwell, in his \textit{Treatise on Electricity and Magnetism} of 1873, first predicted that light possesses a momentum and therefore it exerts a pressure when it hits upon a surface,
\begin{quote}``... in a medium in which waves are propagated,
there is a pressure in the direction normal to the
waves, and numerically equal to the energy in
unit of volume.'' \cite{maxwell}
\end{quote}

This pressure is very small, for example the Sun light pressure on Earth is of the order of some micropascals.
However it becomes extremely important in the dynamics of micro/nano optomechanical systems, where the typical interactions are very weak and the pressure of a laser field on a mechanical device, e.g. on a mirror, is no more negligible. 

We are going to derive the quantum mechanical Hamiltonian of a cavity mode coupled with a movable mirror (usually one of the cavity extremes).

\subsection{Semiclassical theory of radiation pressure}
We model the system as $n$ photons hitting a movable mirror of an optical cavity.

\begin{figure}[H]
\centerline{\includegraphics[width=0.7\textwidth]{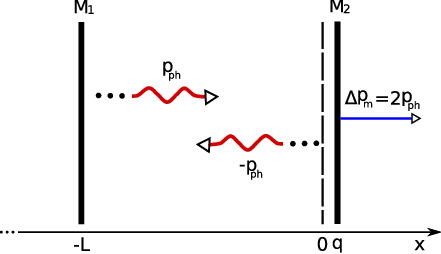}} 
\caption{Elastic collisions of photons upon the movable mirror $M_2$. $M_1$ is fixed.} \label{semipressure}
\end{figure}

When a photon with momentum $p_{ph}=E_{ph}/c$ makes an elastic collision with the mirror (see Fig. \ref{semipressure}), the conservation of total momentum imposes a change in the mirror momentum of 
\begin{equation}
\Delta p_{m}=2p_{ph}=\frac{2E_{ph}}{c}.
\end{equation}
The time delay between different collisions of a photon with the mirror is given by the round trip time 
\begin{equation}
\Delta t=\frac{2L}{c}.
\end{equation}
The total force exerted by all the $n$ photons on the mirror is given by the Newton law
\begin{equation}
F=n\frac{\Delta p_m}{\Delta t}=n\frac{E_{ph}}{L},
\end{equation}
corresponding to the classical Hamiltonian
\begin{equation}
 H=H_{light}+H_{mirror}+H_{int},
\end{equation}
with
\begin{equation}
 H_{int}=-n\frac{E_{ph}}{L}q,
\end{equation}
where $q$ is the displacement of the mirror form its equilibrium position.

The Hamiltonian quantization can be done associating to $E_{ph}$ the photon energy $E_{ph}=\hbar \omega$ and to the variables $n$ and $q$ the corresponding quantum mechanical operators,
\begin{eqnarray}
 n&\longrightarrow& \hat a^\dag \hat a ,\\
q&\longrightarrow& \hat q ,
\end{eqnarray}
so that the quantum Hamiltonian becomes
\begin{equation}
 \hat H=\hat H_{light}+\hat H_{mirror}+\hat H_{int},
\end{equation}
with
\begin{equation}
 \hat H_{int}=-\frac{\hbar \omega}{L}\hat a^\dag \hat a\,\hat q. \label{radpress}
\end{equation}

\subsection{Eigenmodes theory of radiation pressure}
Another derivation of the radiation pressure potential (\ref{radpress}) can be obtained starting from the eigenmodes equation connecting the supported frequencies $\omega_n$ and the length $L$ of the cavity.

\begin{figure}[H]
\centerline{\includegraphics[width=0.7\textwidth]{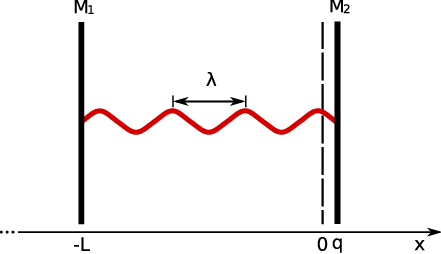}} 
\caption{Variations of the cavity length cause variations of the cavity eigenfrequencies.}\label{modespressure}
\end{figure}

In the simple configuration of Fig.(\ref{modespressure}), if we suppose the mirrors $M_1$ and $M_2$ fixed at $x=-L$ and $x=0$ respectively, then the standing waves condition requires that the eigenmodes are such that $n\lambda/2=L$ 
with $n\in \mathbb N$. This means that the allowed wave vectors are 
\begin{equation}
 k_n=\frac{2\pi}{\lambda}=\frac{n\pi}{L}
\end{equation}
and using the plane wave relation $\omega_n=k_nc$, we find the eigenmodes equation
\begin{equation}
 \omega_n=\frac{n\pi c}{L}. \label{eigenf}
\end{equation}
Let us suppose that we pump a particular mode $\omega_{\bar n}$ and that the mirror of the cavity are fixed, then the Hamiltonian will be
\begin{equation}
 \hat H=\hbar \omega_{\bar n} \hat a^\dag \hat a.
\end{equation}
Now we relax the condition of a fixed cavity length and we let the right mirror $M_2$ free to move around its equilibrium position $x=0$ (see Fig. \ref{modespressure}), if we call with $\hat q<<L$ the displacement of the mirror, then the cavity eigenmodes will depend on $\hat q$ through
\begin{equation}
 \omega_{\bar n}(\hat q)=\frac{n\pi c}{L+\hat q}\simeq \frac{n\pi c}{L}\left[1-\frac{\hat q}{L}\right]=\omega_{\bar n}(0)\left[1-\frac{\hat q}{L}\right]. \label{omegaq}
\end{equation}

The Hamiltonian of the optomechanical system will be
\begin{equation}
 \hat H=\hbar \omega(\hat q) \hat a^\dag \hat a + \hat H_m, \label{Homegaq}
\end{equation}
where $\hat H_m$ is the Hamiltonian of the non interacting mirror, which can be treated for example like a free mass or like a harmonic oscillator, depending on the particular system.

Substituting the linear expansion of the frequency (\ref{omegaq}) in (\ref{Homegaq}), we can separate a free field energy and an interaction potential which is exactly equal to that given in (\ref{radpress}),
\begin{equation}\label{hamiltmirr}
 \hat H=\hbar \omega_{\bar n}(0) \hat a^\dag \hat a +\hat H_m- \frac{\hbar \omega_{\bar n}(0)}{L} \hat a^\dag \hat a\,\hat q\,.
\end{equation}

\subsubsection{Cavity with a membrane inside}
Differently from the semiclassical derivation which is good for perfectly reflective mirrors, the expansion of the eigen-frequencies as functions of $\hat q$ can be easily generalized to more complex systems, like partially reflective mirrors, or cavities with a thick membrane inside \cite{harris,membrane}.

\begin{figure}[H]
\centerline{\includegraphics[width=0.7\textwidth]{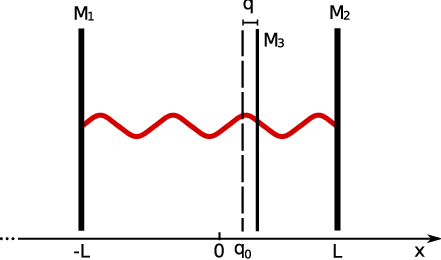}} 
\caption{A thick membrane $M_3$ is placed inside an optical cavity. Radiation pressure and vibrations cause small displacements $q$ from the rest position $q_0$.}\label{Figmemb}
\end{figure}

Let us consider a cavity of length $2L$ and a vibrating membrane positioned at $x=q_0$ (see Fig. \ref{Figmemb}).

It can be shown that the eigenmodes of the system satisfy the following equations
\begin{equation}
\cot k (L+q_0) + \cot k (L-q_0) = 2 \sqrt{R/T}, \label{eigenmemb}
\end{equation}
where $k$ is the wave vector, while $R$ and $T$ are the membrane reflectivity and transmissivity respectively.

We first study the special limit $q_0=T=0$. In this case we have two equal cavities and the supported frequencies are the same of (\ref{eigenf}). For each integer $n$ we have two modes of frequency $\omega_n$, one on the left and one on the right of the membrane.

Now we relax the condition of perfect reflectivity by considering a membrane with $T \ne 0$ and we assume also a de-centered position of the membrane $q_0\ne 0$. The symmetry has been broken, and for each integer $n$ the previous eigenmode $\omega_n$ splits in
\begin{eqnarray}
\omega_n^-(q_0) &=& \omega_n + \frac{c}{2 L} \big[ \sin^{-1} (\sqrt{R} \cos 2 k_n q_0)-\sin^{-1}\sqrt{R}\big], \label{split1}\\
\omega_n^+(q_0)&=& \omega_n + \frac{c}{2 L} \big[\pi- \sin^{-1} (\sqrt{R} \cos 2 k_n q_0)-\sin^{-1}\sqrt{R}\big]. \label{split2}
\end{eqnarray}
These expressions are just a different rewriting of to the implicit relation (\ref{eigenmemb}).
Now we consider small displacements $\hat q<<q_0$ of the membrane around its equilibrium position $q_0$.
Equations (\ref{split1}-\ref{split2}) can be linearized, giving
\begin{eqnarray}
\omega_n^-(\hat q) &=& \omega_n -\delta^- - f \frac{\omega_n}{L} \hat q ,\label{omega-} \\
\omega_n^+(\hat q)&=& \omega_n + \delta^+ + f \frac{\omega_n}{L} \hat q ,\label{omega+}
\end{eqnarray}
where
\begin{eqnarray}
\delta^- &=&   \frac{c}{2 L} \big[\sin^{-1}\sqrt{R}- \sin^{-1} (\sqrt{R} \cos 2 k_n q_0)\big] , \\
\delta^+ &=& \frac{c}{2 L} \big[\pi- \sin^{-1} (\sqrt{R} \cos 2 k_n q_0)-\sin^{-1}\sqrt{R}\big], \\
f &=& \frac{\sin 2 k_n q_0}{\sqrt{R^{-1}-\cos^2 2 k_n q_0}}.\label{f}
\end{eqnarray}

If the equilibrium position $q_0$ is such that the sine function in (\ref{f}) is different from zero, then the position dependace of the frequency is very similar to (\ref{omegaq}).

The Hamiltonian of the system, corresponding to the splitted modes (\ref{omega-}-\ref{omega+}), is
\begin{eqnarray}\label{Hmembrane}
 \hat H&=&\hbar (\omega_{\bar n}-\delta^-)\hat a^\dag \hat a+\hbar (\omega_{\bar n}+\delta^+)\hat b^\dag \hat b +\hat H_m \\
       && - \hbar f\frac{\hbar \omega_{\bar n}}{L} \hat a^\dag \hat a\,\hat q+ \hbar f\frac{\hbar \omega_{\bar n}}{L} \hat b^\dag \hat b\,\hat q\,,
\end{eqnarray}
where, in the first line there are the free field Hamiltonians of the two shifted optical modes (with lowering operators $\hat a$ and $\hat b$) and of the moving membrane, while in the second line there are the radiation pressure contributions of the two modes. We observe that, in this configuration, the radiation pressure terms have opposite signs corresponding to opposite force directions.

\begin{chapter}[One driven mode and a movable mirror]{One driven mode and a movable mirror\normalsize{ \cite{1mfilter}}}\label{1mchapter}
\section{Introduction}

Experimental demonstration of genuine quantum states of macroscopic mechanical resonators with a mass in the nanogram-milligram range represents an important step not only for the high-sensitive detection of displacements and forces, but also for the foundations of physics. 
It would represent, in fact, a remarkable signature of the quantum behavior of a macroscopic object, allowing to shed further light onto the quantum-classical boundary. Significant experimental
\cite{sidebcooling,harris}
and theoretical \cite{genes07} efforts are currently devoted to cooling such microresonators to their quantum ground state.

However, the generation of other examples of quantum states of a micro-mechanical resonator has been also considered recently. The most relevant
examples are given by squeezed and entangled states. Squeezed states of nano-mechanical resonators are potentially useful
for surpassing the standard quantum limit for position and force detection. 

The conditions under which entanglement between macroscopic objects can arise is also currently investigated. After the proposal of Ref.~\cite{Camerino}, in which two mirrors of a ring cavity are entangled by the radiation pressure of the cavity mode, other optomechanical systems have been proposed for entangling optical and/or mechanical modes by means of the radiation pressure interaction. 
Refs.~\cite{prl07} considered the simplest scheme capable of generating stationary optomechanical entanglement, i.e., a single Fabry-Perot cavity with a movable mirror.

Here we shall reconsider the Fabry-Perot model of Ref.~\cite{prl07}, which is remarkable for its simplicity and robustness against temperature of the resulting entanglement, and extend its study in various directions. 
In particular we shall develop a general theory showing how the entanglement between the mechanical resonator and optical output modes can be properly defined and calculated.
This is important since any quantum communication application involves \emph{traveling output} modes rather than intracavity ones, furthermore, by considering the output field, one can adopt a multiplexing approach that is, using appropriate filters, one can always select many different traveling output modes originating from a single intracavity mode (see Fig.~1). 

\begin{figure}[H]
\centerline{\includegraphics[width=0.7\textwidth]{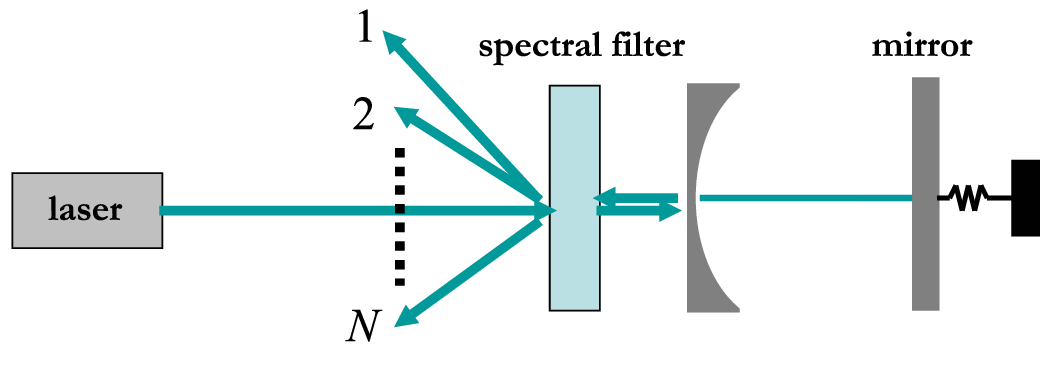}} \caption{Scheme of the cavity, which is driven by a laser and has a vibrating
mirror. With appropriate filters one can select $N$ independent modes from the cavity output field.}\label{fig1}
\end{figure}
 
We shall see that the relevant dynamics induced by radiation pressure interaction is carried by the two output modes corresponding to the first Stokes and anti-Stokes sidebands of the driving laser. 
In particular, the optomechanical entanglement with the intracavity mode is optimally transferred to the output Stokes sideband mode, which is however robustly entangled also with the anti-Stokes output mode. We shall see that the present Fabry-Perot cavity system is preferable with respect to the free space model of Refs.~\cite{prltelep}, because entanglement is achievable in a much more accessible experimental parameter region.

\section{System dynamics}

We consider a driven optical cavity coupled by radiation pressure to a micromechanical oscillator. The typical experimental configuration is a
Fabry-Perot cavity with one mirror much lighter than the other (see e.g. \cite{harris}), but our
treatment applies to other configurations such as the silica toroidal microcavity of Refs.~\cite{sidebcooling}. Radiation
pressure couples each cavity mode with many vibrational normal modes of the movable mirror. However, by choosing the detection bandwidth so that
only an isolated mechanical resonance significantly contributes to the detected signal, one can restrict to a single mechanical oscillator,
since inter-mode coupling due to mechanical nonlinearities are typically negligible. The
Hamiltonian of the system reads \cite{GIOV01}
\begin{eqnarray}
&& H=\hbar\omega_{c}a^{\dagger}a+\frac{1}{2}\hbar\omega_{m}(p^{2}+q^{2})-
\hbar G_{0}a^{\dagger}a q  \\
&& +i\hbar E(a^{\dagger}e^{-i\omega_{0}t}-ae^{i\omega_{0}t}).  \label{ham0}
\end{eqnarray}
The first term describes the energy of the cavity mode, with lowering operator $a$ ($[a,a^{\dag}]=1$), cavity frequency $\omega_c$ and decay
rate $ \kappa$. The second term gives the energy of the mechanical mode, modeled as harmonic oscillator at frequency $\omega_m$ and described by
dimensionless position and momentum operators $q$ and $p$ ($ [q,p]=i$). The third term is the radiation-pressure coupling of rate $
G_0=(\omega_c/L)\sqrt{\hbar/m \omega_m}$, where $m$ is the effective mass of the mechanical mode \cite{Pinard}, and $L$ is an effective length
that depends upon the cavity geometry: it coincides with the cavity length in the Fabry-Perot case, and with the toroid radius in the case of
Refs.~\cite{sidebcooling}. The last term describes the input driving laser with frequency $\omega_0$, where $E$ is related to the
input laser power $P$ by $|E|=\sqrt{2P \kappa/\hbar \omega_0}$. One can adopt the single cavity mode description of Eq.~(\ref{ham0}) as long as
one drives only one cavity mode and the mechanical frequency $\omega_m$ is much smaller than the cavity free spectral range $FSR \sim c/L$. In
this case, scattering of photons from the driven mode into other cavity modes is negligible \cite{law}.

The dynamic is also determined by the fluctuation-dissipation processes affecting both the optical and the mechanical mode. As we have seen in the previous sections, these fluctuations can be taken into account using the Langevin equations approach (\ref{dotab1},\ref{pdotBrow5}). In this case, we can write a set of nonlinear QLE in the interaction picture with
respect to $ \hbar \omega_0 a^{\dag}a$,
\begin{subequations}
\label{nonlinlang}
\begin{eqnarray}
\dot{q}&=&\omega_m p, \\
\dot{p}&=&-\omega_m q - \gamma_m p + G_0 a^{\dag}a + \xi, \\
\dot{a}&=&-(\kappa+i\Delta_0)a +i G_0 a q +E+\sqrt{2\kappa} a^{in},
\end{eqnarray}
where $\Delta_0=\omega_c-\omega_0$. The mechanical mode is affected by a viscous force with damping rate $\gamma_m$ and by a Brownian stochastic
force with zero mean value $\xi$, that obeys the correlation function (\ref{BrownNoise}), which rewritten for dimensionless mirror operators becomes
\end{subequations}
\begin{equation}  \label{browncorre}
\left \langle \xi(t) \xi(t^{\prime})\right \rangle = \frac{\gamma_m}{\omega_m%
} \int \frac{d\omega}{2\pi} e^{-i\omega(t-t^{\prime})} \omega \left[%
\coth\left(\frac{\hbar \omega}{2k_BT}\right)+1\right],
\end{equation}
where $k_B$ is the Boltzmann constant and $T$ is the temperature of the reservoir of the micromechanical oscillator. The cavity mode amplitude instead decays at the rate $\kappa$ and is affected by the vacuum radiation input noise
$a^{in}(t)$, whose correlation functions are given by (see \ref{stoc1}-\ref{stoc3})
\begin{equation}
\langle a^{in}(t)a^{in,\dag}(t^{\prime})\rangle =\left[N(\omega_c)+1\right]%
\delta (t-t^{\prime}),  \label{input1}
\end{equation}
and
\begin{equation} \langle a^{in,\dag}(t)a^{in}(t^{\prime})\rangle =N(\omega_c)\delta (t-t^{\prime}), \label{input2}
\end{equation}
where $N(\omega_c)=\left(\exp\{\hbar \omega_c/k_BT\}-1\right)^{-1}$ is the equilibrium mean thermal photon number. At optical frequencies $\hbar
\omega_c/k_BT \gg 1$ and therefore $N(\omega_c)\simeq 0$, so that only the correlation function of Eq.~(\ref{input1}) is relevant. We shall
neglect here technical noise sources, such as the amplitude and phase fluctuations of the driving laser. They can hinder the achievement of
genuine quantum effects (see e.g. \cite{sidebcooling}), but they could be easily accounted for by introducing fluctuations of the modulus and of
the phase of the driving parameter $E$ of Eq.~(\ref{ham0}).

\subsection{Linearization and stability analysis}

As shown in \cite{prl07}, significant optomechanical entanglement is achieved when radiation pressure coupling is strong, which is realized when
the intracavity field is very intense, i.e., for high-finesse cavities and enough driving power. In this limit (and if the system is stable) the
system is characterized by a semiclassical steady state with the cavity mode in a coherent state with amplitude $ \alpha_s = E/(\kappa+i
\Delta)$, and the micromechanical mirror displaced by $q_s = G_0 |\alpha_s|^2/\omega_m$ (see Refs.~\cite{prl07,manc-tomb} for details).
The expression giving the intracavity amplitude $\alpha_s$ is actually an implicit nonlinear equation for $\alpha_s$ because
\begin{equation}
\Delta = \Delta_0- \frac{G_0^2 |\alpha_s|^2}{\omega_m}.
\end{equation}
is the effective cavity detuning including the effect of the stationary radiation pressure. As shown in Refs.~\cite{prl07}, when
$|\alpha_s| \gg 1$ the quantum dynamics of the fluctuations around the steady state is well described by linearizing the nonlinear QLE of
Eqs.~(\ref{nonlinlang}). Defining the cavity field fluctuation quadratures $\delta X\equiv(\delta a+\delta a^{\dag})/\sqrt{2}$ and $\delta
Y\equiv(\delta a-\delta a^{\dag})/i\sqrt{2}$, and the corresponding Hermitian
input noise operators $X^{in}\equiv(a^{in}+a^{in,\dag})/\sqrt{2}$ and $%
Y^{in}\equiv(a^{in}-a^{in,\dag})/i\sqrt{2}$, the linearized QLE can be written in the following compact matrix form \cite{prl07}
\begin{equation}\label{compaeq}
\dot{u}(t)=A u(t)+n(t),
\end{equation}
where $u^{T}(t) =(\delta q(t), \delta p(t),\delta X(t), \delta Y(t))^T$ ($^{T}$ denotes the transposition) is the vector of CV fluctuation
operators , $n^{T}(t) =(0, \xi(t),\sqrt{2\kappa}X^{in}(t), \sqrt{2\kappa}Y^{in}(t))^T$ the corresponding vector of noises and $A$ the matrix
\begin{equation}\label{dynmat}
  A=\left(\begin{array}{cccc}
    0 & \omega_m & 0 & 0 \\
     -\omega_m & -\gamma_m & G & 0 \\
    0 & 0 & -\kappa & \Delta \\
    G & 0 & -\Delta & -\kappa
  \end{array}\right),
\end{equation}
where
\begin{equation} G=G_0 \alpha_s\sqrt{2}=\frac{2\omega_c}{L}\sqrt{\frac{P \kappa}{m \omega_m \omega_0 \left(\kappa^2+\Delta^2\right)}},
\label{optoc}
\end{equation}
is the \emph{effective} optomechanical coupling (we have chosen the phase reference so that $\alpha_s$ is real and positive). When $\alpha_s \gg
1$, one has $G \gg G_0$, and therefore the generation of significant optomechanical entanglement is facilitated in this linearized regime.

The formal solution of Eq.~(\ref{compaeq}) is $ u(t)=M(t) u(0)+\int_0^t ds M(s) n(t-s)$, where $M(t)=\exp\{A t\}$. The system is stable and
reaches its steady state for $t \to \infty$ when all the eigenvalues of $A$ have negative real parts so that $M(\infty)=0$. The stability
conditions can be derived by applying the Routh-Hurwitz criterion \cite{grad}, yielding the following two nontrivial conditions on the system parameters,
\begin{subequations}
\label{stabpre}
\begin{eqnarray}
&& s_1=2\gamma_m\kappa\left\{ \left[ \kappa^{2}+\left( \omega_m-\Delta\right) ^{2}\right] \left[ \kappa^{2}+\left(
\omega_m+\Delta\right) ^{2}\right] \right.  \notag \\
&&\left.+\gamma_m\left[ \left( \gamma_m+2\kappa\right) \left( \kappa^{2}+\Delta ^{2}\right) +2\kappa\omega_{m}^{2}\right] \right\}  \notag
\\
&&+\Delta\omega_{m} G^{2}\left( \gamma_m+2\kappa\right) ^{2}>0, \\
&&s_2=\omega_{m}\left( \kappa^{2}+\Delta^{2}\right) -G^{2}\Delta>0.\label{stabbis}
\end{eqnarray}
\end{subequations}
which will be considered to be satisfied from now on. Notice that when $\Delta > 0$ (laser red-detuned with respect to the cavity) the first
condition is always satisfied and only $s_2$ is relevant, while when $\Delta < 0$ (blue-detuned laser), the second condition is always satisfied
and only $s_1$ matters.

\subsection{Correlation matrix of the quantum fluctuations}

The steady state of the bipartite quantum system formed by the vibrational mode of interest and the fluctuations of the intracavity mode can be
fully characterized. In fact, the quantum noises $\xi$ and $a^{in}$ are zero-mean quantum Gaussian noises and the dynamics is linearized, and as
a consequence, the steady state of the system is a zero-mean bipartite Gaussian state, fully characterized by its $4 \times 4 $ correlation
matrix (CM) $ V_{ij}=\left(\langle u_i(\infty)u_j(\infty)+ u_j(\infty)u_i(\infty)\rangle\right)/2$. Starting from Eq.~(\ref{compaeq}), this
steady state CM can be determined in two equivalent ways. Using the Fourier transforms $\tilde{u}_i(\omega)$ of $u_i(t)$, one has
\begin{equation}
V_{ij}(t)=\int \int\frac{d\omega d\omega'}{4\pi}e^{-it(\omega+\omega')}\left\langle \tilde{u}_i(\omega) \tilde{u}_j(\omega')+
\tilde{u}_j(\omega')\tilde{u}_i(\omega)\right\rangle. \label{defV}
\end{equation}
Then, by Fourier transforming Eq.~(\ref{compaeq}) and the correlation functions of the noises, Eqs.~(\ref{browncorre}) and (\ref{input1}), one
gets
\begin{eqnarray}\label{lemma}
&&\frac{\left\langle \tilde{u}_i(\omega) \tilde{u}_j(\omega')+ \tilde{u}_j(\omega')\tilde{u}_i(\omega)\right\rangle}{2} \\
&&= \left[\tilde{M}(\omega)D(\omega)\tilde{M}(\omega')^{T}\right]_{ij}\delta(\omega+\omega'),
\end{eqnarray}
where we have defined the $4 \times 4 $ matrices
\begin{equation} \tilde{M}(\omega)=\left(i\omega+A\right)^{-1} \label{invea}
\end{equation}
and
\begin{equation}\label{diffuom}
  D(\omega)=\left(\begin{array}{cccc}
    0 & 0 & 0 & 0 \\
     0 & \frac{\gamma_m \omega}{\omega_m} \coth\left(\frac{\hbar \omega}{2k_BT}\right) & 0 & 0 \\
    0 & 0 & \kappa & 0 \\
    0 & 0 & 0 & \kappa
  \end{array}\right).
\end{equation}
The $\delta(\omega+\omega')$ factor is a consequence of the stationarity of the noises, which implies the stationarity of the CM $V$: in fact,
inserting Eq.~(\ref{lemma}) into Eq.~(\ref{defV}), one gets that $V$ is time-independent and can be written as
\begin{equation}
V=\int d\omega \tilde{M}(\omega)D(\omega)\tilde{M}(\omega)^{\dagger}. \label{Vfin}
\end{equation}
It is however reasonable to simplify this exact expression for the steady state CM, by appropriately approximating the thermal noise
contribution $D_{22}(\omega)$ in Eq.~(\ref{diffuom}). In fact $k_B T/\hbar \simeq 10^{11}$ s$^{-1}$ even at cryogenic temperatures and it is
therefore much larger than all the other typical frequency scales, which are at most of the order of $10^9$ Hz. The integrand in
Eq.~(\ref{Vfin}) goes rapidly to zero at $\omega \sim 10^{11}$ Hz, and therefore one can safely neglect the frequency dependence of
$D_{22}(\omega)$ by approximating it with its zero-frequency value
\begin{equation}  \label{thermappr}
\frac{\gamma_m \omega}{\omega_m} \coth\left(\frac{\hbar \omega}{2k_BT}%
\right) \simeq \gamma_m \frac{2k_B T}{\hbar \omega_m} \simeq \gamma_m\left(2%
\bar{n}+1\right),
\end{equation}
where $\bar{n}=\left(\exp\{\hbar \omega_m/k_BT\}-1\right)^{-1}$ is the mean thermal excitation number of the resonator.

It is easy to verify that assuming a frequency-independent diffusion matrix $D$ is equivalent to make the following Markovian approximation on
the quantum Brownian noise $\xi(t)$,
\begin{equation}
\label{browncorreMa} \left\langle \xi(t)\xi(t^{\prime})+\xi(t^{\prime})\xi(t)\right\rangle/2 \simeq \gamma_m (2n+1) \delta(t-t^{\prime}).
\end{equation}
Within this Markovian approximation, the above frequency domain treatment is equivalent to the time domain derivation considered in \cite{prl07}
which, starting from the formal solution of Eq.~(\ref{compaeq}), arrives at
\begin{equation} \label{cm2}
V_{ij}(\infty)=\sum_{k,l}\int_0^{\infty} ds \int_0^{\infty}ds' M_{ik}(s) M_{jl}(s')D_{kl}(s-s'),
\end{equation}
where $D_{kl}(s-s')=\left(\langle n_k(s)n_l(s')+ n_l(s')n_k(s)\rangle\right)/2$ is the matrix of the stationary noise correlation functions. The
Markovian approximation of the thermal noise on the mechanical resonator yields $D_{kl}(s-s')= D_{kl} \delta(s-s')$, with $D = \mathrm{Diag }
[0,\gamma_m (2 \bar{n}+1),\kappa,\kappa]$, so that Eq.~(\ref{cm2}) becomes
\begin{equation} \label{cm3}
V =\int_0^{\infty} ds  M(s)D M(s)^{T},
\end{equation}
which is equivalent to Eq.~(\ref{Vfin}) whenever $D$ does not depend upon $\omega$. When the stability conditions are satisfied ($M(\infty)=0$),
Eq.~(\ref{cm3}) is equivalent to the following Lyapunov equation for the steady-state CM,
\begin{equation} \label{lyap}
AV+VA^{T}=-D,
\end{equation}
which is a linear equation for $V$ and can be straightforwardly solved, but the general exact expression is too cumbersome and will not be
reported here.

\section{Optomechanical entanglement with the intracavity mode}

\label{Sec:intra}

In order to establish the conditions under which the optical mode and the
mirror vibrational mode are entangled we consider the logarithmic negativity
$E_{\mathcal{N}}$, which we have defined in (\ref{entnegformula}) as
\begin{equation}
E_{\mathcal{N}}=\max [0,-\ln 2\eta ^{-}],  \label{logneg}
\end{equation}%
where $\eta ^{-}\equiv 2^{-1/2}\left[ \Sigma (V)-\left[ \Sigma (V)^{2}-4\det
V\right] ^{1/2}\right] ^{1/2}$, 
with $\Sigma (V)\equiv \det V_{m}+\det V_{c}-2\det V_{mc}$, and we have used
the $2\times 2$ block form of the CM
\begin{equation}
V\equiv \left(
\begin{array}{cc}
V_{m} & V_{mc} \\
V_{mc}^{T} & V_{c}%
\end{array}%
\right) .  \label{blocks}
\end{equation}%

\subsection{Correspondence with the down-conversion process}

As already shown in \cite{sidebcooling,genes07} many features of the radiation pressure interaction in the cavity can be understood by
considering that the driving laser light is scattered by the vibrating cavity boundary mostly at the first Stokes ($\omega _{0}-\omega _{m}$)
and anti-Stokes ($\omega _{0}+\omega _{m}$) sidebands. Therefore we expect that the optomechanical interaction and eventually entanglement will
be enhanced
when the cavity is resonant with one of the two sidebands, i.e., when $%
\Delta =\pm \omega _{m}.$

It is useful to introduce the mechanical annihilation operator $\delta
b=(\delta q+i\delta p)/\sqrt{2}$, obeying the following QLE
\begin{equation}
\delta \dot{b}=-i\omega _{m}\delta b-\frac{\gamma _{m}}{2}\left( \delta
b-\delta b^{\dagger }\right) +i\frac{G}{2}\left( \delta a^{\dag }+\delta
a\right) +\frac{\xi }{\sqrt{2}}.  \label{bequat}
\end{equation}
Moving to another interaction picture by introducing the slowly-moving tilded operators $\delta b(t)=\delta \tilde{b}(t)e^{-i\omega _{m}t}$ and
$ \delta a(t)=\delta \tilde{a}(t)e^{-i\Delta t}$, we obtain from the linearized version of Eq.~(\ref{nonlinlang}c) and Eq.~(\ref{bequat}) the
following QLEs
\begin{eqnarray}
\delta \dot{\tilde{b}} &=&-\frac{\gamma _{m}}{2}\left( \delta \tilde{b}%
-\delta \tilde{b}^{\dagger }e^{2i\omega _{m}t}\right) +\sqrt{\gamma _{m}}%
b^{in}  \notag \\
&&\ \ +i\frac{G}{2}\left( \delta \tilde{a}^{\dag }e^{i(\Delta +\omega
_{m})t}+\delta \tilde{a}e^{i(\omega _{m}-\Delta )t}\right)   \label{mode2} \\
\delta \dot{\tilde{a}} &=&-\kappa \delta \tilde{a}+i\frac{G}{2}\left( \delta
\tilde{b}^{\dag }e^{i(\Delta +\omega _{m})t}+\delta \tilde{b}e^{i(\Delta
-\omega _{m})t}\right)   \notag \\
&&\ \ +\sqrt{2\kappa }\tilde{a}^{in}.  \label{modemech2}
\end{eqnarray}%
Note that we have introduced two noise operators: i) $\tilde{a}%
^{in}(t)=a^{in}(t)e^{i\Delta t}$, possessing the same correlation function
as $a^{in}(t)$; ii) $b^{in}(t)=\xi (t)e^{i\omega _{m}t}/\sqrt{2}$ which, in
the limit of large $\omega _{m}$, acquires the correlation functions \cite%
{gard}
\begin{eqnarray}
\langle b^{in,\dag }(t)b^{in}(t^{\prime })\rangle  &=&\bar{n}\delta
(t-t^{\prime }), \\
\langle b^{in}(t)b^{in,\dag }(t^{\prime })\rangle  &=&\left[ \bar{n}+1\right]
\delta (t-t^{\prime }).
\end{eqnarray}%
Eqs.~(\ref{mode2})-(\ref{modemech2}) are still equivalent to the linearized QLEs of Eq.~(\ref{compaeq}), but now we particularize them by
choosing $\Delta =\pm \omega _{m}.$
If the cavity is resonant with the Stokes sideband of the driving laser, $%
\Delta =-\omega _{m}$, one gets
\begin{eqnarray}
\delta \dot{\tilde{b}} &=&-\frac{\gamma _{m}}{2}\delta \tilde{b}+\frac{%
\gamma _{m}}{2}\delta \tilde{b}^{\dagger }e^{2i\omega _{m}t}+i\frac{G}{2}%
\delta \tilde{a}^{\dag }\nonumber \\
&+& i\frac{G}{2}\delta \tilde{a}e^{2i\omega _{m}t}+%
\sqrt{\gamma _{m}}b^{in},  \label{mode3} \\
\delta \dot{\tilde{a}} &=&-\kappa \delta \tilde{a}+i\frac{G}{2}\delta \tilde{%
b}^{\dag }+i\frac{G}{2}\delta \tilde{b}e^{2i\omega _{m}t}+\sqrt{2\kappa }%
\tilde{a}^{in},  \label{mode4}
\end{eqnarray}%
while when the cavity is resonant with the anti-Stokes sideband of the driving laser, $\Delta =\omega _{m}$, one gets
\begin{eqnarray}
\delta \dot{\tilde{b}} &=&-\frac{\gamma _{m}}{2}\delta \tilde{b}+\frac{%
\gamma _{m}}{2}\delta \tilde{b}^{\dagger }e^{2i\omega _{m}t}+i\frac{G}{2}%
\delta \tilde{a}\nonumber \\
&+& i\frac{G}{2}\delta \tilde{a}^{\dag }e^{-2i\omega _{m}t}+%
\sqrt{\gamma _{m}}b^{in},  \label{mode3prime} \\
\delta \dot{\tilde{a}} &=&-\kappa \delta \tilde{a}+i\frac{G}{2}\delta \tilde{%
b}+i\frac{G}{2}\delta \tilde{b}^{\dag }e^{-2i\omega _{m}t}+\sqrt{2\kappa }%
\tilde{a}^{in}.  \label{mode4prime}
\end{eqnarray}%
From Eqs.~(\ref{mode3})-(\ref{mode4}) we see that, for a blue-detuned driving laser, $\Delta =-\omega _{m}$, the cavity mode and mechanical
resonator are coupled via two kinds of interactions: i) a down-conversion process characterized by $\delta \tilde{b}^{\dag }\delta
\tilde{a}^{\dag }+\delta \tilde{a}\delta \tilde{b}$, which is resonant and ii) a beam-splitter-like process characterized by $\delta
\tilde{b}^{\dagger }\delta \tilde{a}+\delta \tilde{a}^{\dagger }\delta \tilde{b}$, which is off resonant. Since the beam splitter interaction is
not able to entangle modes starting from classical input states, and it is also off-resonant in this case, one can invoke the
rotating wave approximation (RWA) (which is justified in the limit of $\omega _{m} \gg G,\kappa $) and simplify the interaction to a down
conversion process, which is known to generate bipartite entanglement. In the red-detuned driving laser case,
Eqs.~(\ref{mode3prime})-(\ref{mode4prime}) show that the two modes are strongly coupled by a beam-splitter-like interaction, while the
down-conversion process is off-resonant. If one chose to make the RWA in this case, one would be left with an effective beam splitter
interaction which cannot entangle. Therefore, in the RWA limit $\omega _{m} \gg G,\kappa $, the best regime for strong optomechanical
entanglement is when the laser is blue-detuned from the cavity resonance and down-conversion is enhanced.  However, as it will be seen in the
following section, this is hindered by instability and one is rather forced to work in the opposite regime of a red-detuned laser where
instability takes place only at large values of $G$.

\subsection{Entanglement in the blue-detuned regime}

The CM of the Gaussian steady state of the bipartite system, can be obtained
from Eqs.~(\ref{mode3})-(\ref{mode4}) and Eqs.~Eqs.~(\ref{mode3prime})-(\ref%
{mode4prime}) in the RWA limit, with the techniques of the former section
(see also \cite{Vitali07})
\begin{equation}
V\equiv V^{\pm }=\left(
\begin{array}{cccc}
V_{11}^{\pm } & 0 & 0 & V_{14}^{\pm } \\
0 & V_{11}^{\pm } & \pm V_{14}^{\pm } & 0 \\
0 & \pm V_{14}^{\pm } & V_{33}^{\pm } & 0 \\
V_{14}^{\pm } & 0 & 0 & V_{33}^{\pm }%
\end{array}%
\right) ,  \label{corremat2}
\end{equation}%
where the upper (lower) sign corresponds to the blue- (red-)detuned case,
and
\begin{subequations}
\label{corrematelemrwa}
\begin{eqnarray}
V_{11}^{\pm } &=&\bar{n}+\frac{1}{2}+\frac{2G^{2}\kappa \left[ 1/2\pm \left(
\bar{n}+1/2\right) \right] }{\left( \gamma _{m}+2\kappa \right) \left(
2\gamma _{m}\kappa \mp G^{2}\right) }, \\
V_{33}^{\pm } &=&\frac{1}{2}+\frac{G^{2}\gamma _{m}\left[ \bar{n}+1/2\pm 1/2%
\right] }{\left( \gamma _{m}+2\kappa \right) \left( 2\gamma _{m}\kappa \mp
G^{2}\right) }, \\
V_{14}^{\pm } &=&\frac{2G\gamma _{m}\kappa \left[ \bar{n}+1/2\pm 1/2\right]
}{\left( \gamma _{m}+2\kappa \right) \left( 2\gamma _{m}\kappa \mp
G^{2}\right) }.
\end{eqnarray}%
For clarity we have included the red-detuned case in the RWA approximation and we see that $\det V_{mc}^{\pm }=\mp (V_{14}^{\pm })^{2}$, i.e.,
is non-negative in this latter case, which is a sufficient condition for the separability of bipartite states \cite{Simon}. Of course, this is
expected, since it is just the beam-splitter interaction that generates this CM. Thus, in the weak optomechanical coupling regime of the RWA
limit, entanglement is obtained only for a blue-detuned laser, $\Delta =-\omega _{m} $. However, the amount of achievable optomechanical
entanglement at the steady state is seriously limited by the stability condition of Eq.~(\ref%
{stabpre}a), which in the RWA limit $\Delta =-\omega _{m}\gg \kappa ,\gamma _{m} $, simplifies to $G<\sqrt{2\kappa \gamma _{m}}$. Since one needs
small mechanical dissipation rate $\gamma _{m}$ in order to see quantum effects, this means a very low maximum value for $G$. The logarithmic
negativity $E_{\mathcal{N}}$ is an increasing function of the effective optomechanical coupling $G$ (as expected) and therefore the stability
condition puts a strong upper bound also on $E_{\mathcal{N}}$. In fact, it is possible to prove that the following bound on $E_{\mathcal{N}}$
exists
\end{subequations}
\begin{equation}
E_{\mathcal{N}}\leq \ln \left[ \frac{1+G/\sqrt{2\kappa \gamma _{m}}}{1+\bar{n%
}}\right] ,  \label{logneg2}
\end{equation}%
showing that $E_{\mathcal{N}}\leq \ln 2$ and above all that entanglement is extremely fragile with respect to temperature in the RWA limit
because, due to the stability condition, $E_{\mathcal{N}}$ vanishes as soon as $\bar{n} \geq 1$.

\subsection{Entanglement in the red-detuned regime}

We conclude that, due to instability, one can find significant
optomechanical entanglement, which is also robust against temperature, only
far from the RWA regime, in the strong coupling regime in the region with
positive $\Delta $, because Eq.~(\ref{stabbis}) allows for higher values of $G$%
. This is confirmed by Fig.~\ref{intracav-ent}, where $E_{\mathcal{N}}$ is plotted versus the normalized detuning $\Delta /\omega _{m}$ and the
normalized input power $P/P_{0}$, ($P_{0}=50$ mW) at a fixed value of the cavity finesse $F=F_{0}=1.67\times 10^{4}$ in (a), and versus the
normalized finesse $F/F_{0}$ and normalized input power $P/P_{0}$ at a fixed cavity detuning $\Delta =\omega _{m}$ in (b). We have assumed an
experimentally achievable situation, i.e., a mechanical mode with $\omega _{m}/2\pi =10$ MHz, $\mathcal{Q}=10^{5}$, mass $m=10$ ng, and a cavity
of length $L=1$ mm, driven by a laser with wavelength $810$ nm, yielding $G_{0}=0.95$ kHz and a cavity bandwidth $\kappa =0.9\omega _{m}$ when
$F=F_{0}$. We have assumed a
reservoir temperature for the mirror $T=0.4$ K, corresponding to $\bar{n}%
\simeq 833$. Fig.~\ref{intracav-ent}a shows that $E_{\mathcal{N}}$ is peaked around $\Delta \simeq \omega _{m}$, even though the peak shifts to
larger values of $\Delta $ at larger input powers $P$. For increasing $P$ at fixed $%
\Delta $, $E_{\mathcal{N}}$ increases, even though at the same time the instability region (where the plot suddenly interrupts) widens. In
Fig.~\ref{intracav-ent}b we have fixed the detuning at $\Delta =\omega _{m}$ (i.e., the cavity is resonant with the anti-Stokes sideband of the
laser) and varied
both the input power and the cavity finesse. We see again that $E_{\mathcal{N%
}}$ is maximum just at the instability threshold and also that, once that
the finesse has reached a sufficiently large value, $F\simeq F_{0}$, $E_{%
\mathcal{N}}$ roughly saturates at larger values of $F$. That is, one gets
an optimal optomechanical entanglement when $\kappa \simeq \omega _{m}$ and
moving into the well-resolved sideband limit $\kappa \ll \omega _{m}$ does
not improve the value of $E_{\mathcal{N}}$. The parameter region analyzed is
analogous to that considered in \cite{prl07}, where it has been shown that
this optomechanical entanglement is extremely robust with respect to the
temperature of the reservoir of the mirror, since it persists to more than $%
20$ K.

\begin{figure}[H]
\centerline{\includegraphics[width=0.7\textwidth]{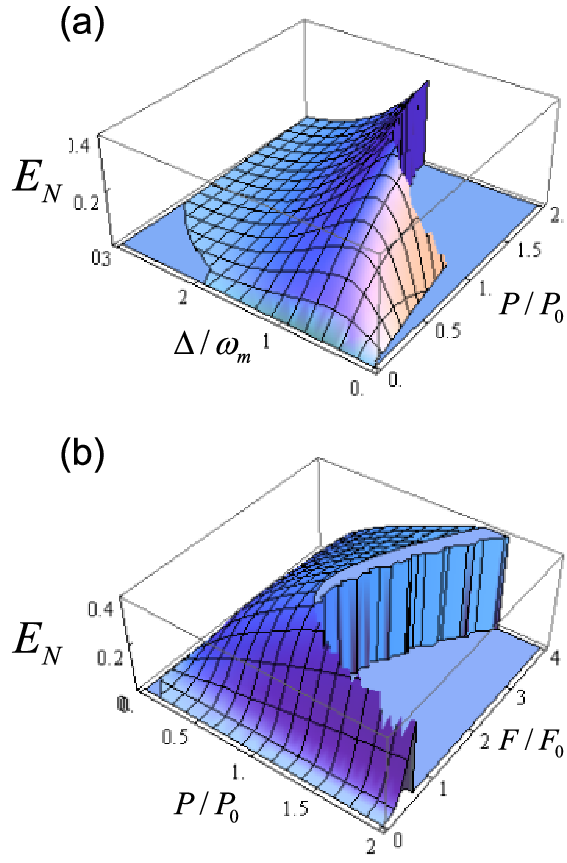}} \caption{(Color online) (a) Logarithmic negativity $E_{\mathcal{N}}$
versus the normalized detuning $\Delta/\protect\omega_m$ and normalized input power $P/P_0$, ($P_0=50$ mW) at a fixed value of the cavity
finesse $F=F_0=1.67 \times 10^4$; (b) $E_{\mathcal{N}}$ versus the normalized finesse $F/F_0$
and normalized input power $P/P_0$ at a fixed detuning $\Delta =\protect%
\omega_m$. Parameter values are $\protect\omega _{m}/2\protect\pi =10$ MHz, $%
\mathcal{Q}=10^5$, mass $m=10$ ng, a cavity of length $L=1$ mm driven by a
laser with wavelength $810$ nm, yielding $G_0=0.95$ KHz and a cavity
bandwidth $\protect\kappa=0.9 \protect\omega_m$ when $F=F_0$. We have
assumed a reservoir temperature for the mirror $T=0.4$ K, corresponding to $%
\bar{n}\simeq 833$. The sudden drop to zero of $E_{\mathcal{N}}$ corresponds
to entering the instability region.}
\label{intracav-ent}
\end{figure}

\subsection{Relationship between entanglement and cooling}

As discussed in detail in \cite{genes07} the same cavity-mechanical resonator system can be used for realizing
cavity-mediated optical cooling of the mechanical resonator via the back-action of the cavity mode. In particular, back-action
cooling is optimized just in the same regime where $\Delta \simeq \omega _{m} $. This fact is easily explained by taking into account the
scattering of the laser light by the oscillating mirror into the Stokes and anti-Stokes sidebands. The generation of an anti-Stokes photon takes
away a vibrational phonon and is responsible for cooling, while the generation of a Stokes photon heats the mirror by producing an extra phonon.
If the cavity is resonant with the anti-Stokes sideband, cooling prevails and one has a positive net laser cooling rate given by the difference
of the scattering rates.

It is therefore interesting to discuss the relation between optimal
optomechanical entanglement and optimal cooling of the mechanical resonator.
This can easily performed because the steady state CM $V$ determines also
the resonator energy, since the effective stationary excitation number of
the resonator is given by $n_{eff}=\left( V_{11}+V_{22}-1\right) /2$ (see
Ref.~\cite{genes07} for the exact expression of these matrix elements giving
the steady state position and momentum resonator variances). In Fig.~\ref%
{intracav-cool} we have plotted $n_{eff}$ under \emph{exactly the same} parameter conditions of Fig.~\ref{intracav-ent}. We see that
\emph{ground state cooling is approached ($n_{eff}<1$) simultaneously with a significant entanglement}. This shows that a significant
back-action cooling of the resonator by the cavity mode is an important condition for achieving an entangled steady state which is robust
against the effects of the resonator thermal bath.

Nonetheless, entanglement and cooling are different phenomena and optimizing one does not generally optimize also the other. This can be
seen by comparing Figs.~\ref{intracav-ent} and \ref{intracav-cool}: $E_{%
\mathcal{N}}$ is maximized always just at the instability threshold, i.e., for the maximum possible optomechanical coupling, while this is not
true for $n_{eff}$, which is instead minimized quite far from the instability threshold. For a more clear understanding we make use of some of
the results obtained for ground state cooling in Refs.~\cite{genes07}. In the perturbative limit where $G \ll \omega
_{m},\kappa $, one can define scattering rates into the Stokes ($A_{+}$) and anti-Stokes ($A_{-}$) sidebands as
\begin{equation}
A_{\pm }=\frac{G^{2}\kappa /2}{\kappa ^{2}+(\Delta \pm \omega _{m})^{2}},
\end{equation}
so that the net laser cooling rate is given by
\begin{equation}\label{netlaser}
\Gamma =A_{-}-A_{+}>0.
\end{equation}
The final occupancy of the mirror mode is consequently given by \cite{genes07}
\begin{equation}
n_{eff}=\frac{\gamma _{m}\bar{n}}{\gamma _{m}+\Gamma }+\frac{A_{+}}{\gamma
_{m}+\Gamma },
\end{equation}
where the first term in the right hand side of the above equation is the minimized thermal noise, that can be made vanishingly small provided
that $\gamma_m \ll \Gamma$, while the second term shows residual heating produced by Stokes scattering off the vibrational ground state. When
$\Gamma \gg \gamma _{m}\bar{n}$, the lower bound for $n_{eff}$ is practically set by the ratio $A_{+}/\Gamma $. However, as soon as $G$ is
increased for improving the entanglement generation, scattering into higher order sidebands takes place, with rates proportional to higher
powers of $G$. As a consequence, even though the effective thermal noise is still close to zero, residual scattering off the ground state takes
place at a rate that can be much higher than $A_{+}$. This can be seen more clearly in the exact expression of $\langle \delta q^{2}\rangle
=V_{11}$ given in \cite{genes07}, which is shown to diverge at the threshold given by Eq.~(\ref{stabbis}).

\begin{figure}[H]
\centerline{\includegraphics[width=0.7\textwidth]{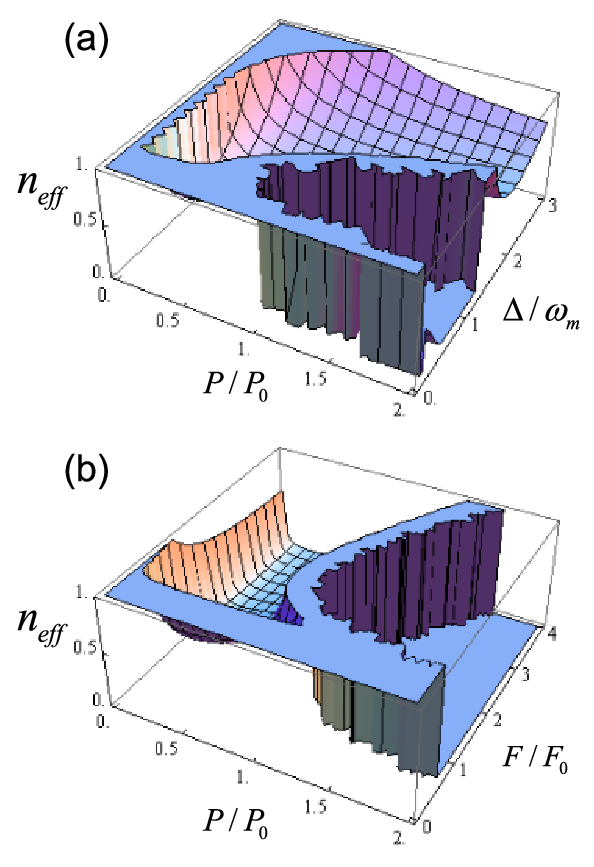}} \caption{(Color online) (a) Effective stationary excitation number of
the
resonator $n_{eff}$ versus the normalized detuning $\Delta /\protect\omega %
_{m}$ and normalized input power $P/P_{0}$, ($P_{0}=50$ mW) at a fixed value
of the cavity finesse $F=F_{0}=1.67\times 10^{4}$; (b) $n_{eff}$ versus the
normalized finesse $F/F_{0}$ and normalized input power $P/P_{0}$ at a fixed
detuning $\Delta =\protect\omega _{m}$. Parameter values are the same as in
Fig.~\protect\ref{intracav-ent}. Again, the sudden drop to zero corresponds
to entering the instability region.}
\label{intracav-cool}
\end{figure}

\section{Optomechanical entanglement with cavity output modes}

The above analysis of the entanglement between the mechanical mode of interest and the intracavity mode provides a detailed description of the
internal dynamics of the system, but it is not of direct use for practical applications. In fact, one typically does not have direct access to
the intracavity field, but one detects and manipulates only the cavity output field. For example, for any quantum communication application, it
is much more important to analyze the entanglement of the mechanical mode with the \emph{optical cavity output}, i.e., how the intracavity
entanglement is transferred to the output field. Moreover, considering the output field provides further options. In fact, by means of spectral
filters, one can always select many different traveling output modes originating from a single intracavity mode and this gives the opportunity
to easily produce and manipulate a multipartite system, eventually possessing multipartite entanglement.

\subsection{General definition of cavity output modes}

The intracavity field $\delta a(t)$ and its output are related by the input-output relation given in (\ref{InOut}), which rewritten with the notation of this chapter becomes
\begin{equation}  \label{inout}
a^{out}(t)= \sqrt{2\kappa}\delta a(t)-  a^{in}(t),
\end{equation}
where the output field possesses the same correlation functions of the optical input field $a^{in}(t)$ and the same commutation relation, i.e.,
the only nonzero commutator is $\left[a^{out}(t),a^{out}(t ^{\prime})^{\dagger}\right]=\delta(t-t^{\prime})$. From the continuous output field
$a^{out}(t)$ one can extract many independent optical modes, by selecting different time intervals or equivalently, different frequency
intervals (see e.g. \cite{fuchs}). One can define a generic set of $N$ output modes by means of the corresponding annihilation operators
\begin{equation}
a^{out}_k(t)=\int_{-\infty}^{t}ds g_k(t-s)a^{out}(s),\;\;\;k=1,\ldots N,\label{filter1}
\end{equation}
where $g_k(s)$ is the causal filter function defining the $k$-th output mode. These annihilation operators describe $N$ independent optical
modes when $\left[a^{out}_j(t),a^{out}_k(t)^{\dagger}\right]=\delta_{jk}$, which is verified when \begin{equation} \int_{0}^{\infty}ds g_j(s)^*
g_k(s)=\delta_{jk}\label{filter2},
\end{equation}
i.e., the $N$ filter functions $g_k(t)$ form an orthonormal set of square-integrable functions in $[0,\infty)$. The situation can be
equivalently described in the frequency domain: taking the Fourier transform of Eq.~(\ref{filter1}), one has
\begin{equation}
\tilde{a}^{out}_k(\omega)=\int_{-\infty}^{\infty}\frac{dt}{\sqrt{2\pi}} a^{out}_k(t)e^{i\omega
t}=\sqrt{2\pi}\tilde{g}_k(\omega)a^{out}(\omega),\label{filterFT}
\end{equation}
where $\tilde{g}_k(\omega)$ is the Fourier transform of the filter function.
An explicit example of an orthonormal set of filter functions is given by
\begin{equation}
g_k(t)=\frac{\theta(t)-\theta(t-\tau)}{\sqrt{\tau}}e^{-i\Omega_k t} , \label{filterex}
\end{equation}
($\theta$ denotes the Heavyside step function) provided that $\Omega_k$ and $\tau$ satisfy the condition
\begin{equation}\label{interfer}
\Omega_j-\Omega_k=\frac{2\pi}{\tau}p, \;\;\;{\rm integer} \;\;p.
\end{equation} These functions describe a set of independent optical modes, each centered
around the frequency $\Omega_k$ and with time duration $\tau$, i.e., frequency bandwidth $\sim 1/\tau$, since
\begin{equation}
\tilde{g}_k(\omega)=\sqrt{\frac{\tau}{2\pi}}e^{i(\omega-\Omega_k)\tau/2}\frac{\sin\left[(\omega-\Omega_k)\tau/2\right]}{(\omega-\Omega_k)\tau/2}
. \label{filterex2}
\end{equation}
When the central frequencies differ by an integer multiple of $2\pi/\tau$, the corresponding modes are independent due to the destructive
interference of the oscillating parts of the spectrum.

\subsection{Stationary correlation matrix of output modes}

The entanglement between the output modes defined above and the mechanical mode is fully determined by the corresponding $(2N+2)\times (2N+2)$
CM, which is defined by
\begin{equation}
V^{out}_{ij}(t)=\frac{1}{2}\left\langle u^{out}_i(t) u^{out}_j(t)+u^{out}_j(t) u^{out}_i(t)\right\rangle,\label{defVout}
\end{equation}
where
\begin{eqnarray}
&&u^{out}(t)\label{defuout} \\
&&=\left(\delta q(t),\delta p(t),X_1^{out}(t),Y_1^{out}(t),\ldots,X_N^{out}(t),Y_N^{out}(t)\right)^T \nonumber
\end{eqnarray}
is the vector formed by the mechanical position and momentum fluctuations and by the amplitude
($X_k^{out}(t)=\left[a^{out}_k(t)+a^{out}_k(t)^{\dagger}\right]/\sqrt{2}$), and phase
($Y_k^{out}(t)=\left[a^{out}_k(t)-a^{out}_k(t)^{\dagger}\right]/i\sqrt{2})$ quadratures of the $N$ output modes. The vector $u^{out}(t)$
properly describes $N+1$ independent CV bosonic modes, and in particular the mechanical resonator is independent of (i.e., it commutes with) the
$N$ optical output modes because the latter depend upon the output field at former times only ($s<t$). From the definition of $u^{out}(t)$, of
the output modes of Eq.~(\ref{filter1}), and the input-output relation of Eq.~(\ref{inout}) one can write
\begin{eqnarray}
&&u^{out}_i(t)=\int_{-\infty}^t ds T_{ik}(t-s)u^{ext}_k(s)\nonumber \\
&&-\int_{-\infty}^t ds T_{ik}(t-s)n^{ext}_k(s), \label{inoutgen}
\end{eqnarray}
where
\begin{equation}
u^{ext}(t)=\left(\delta q(t),\delta p(t),X(t),Y(t),\ldots,X(t),Y(t)\right)^T\label{defintra}
\end{equation}
is the $2N+2$-dimensional vector obtained by extending the four-dimensional vector $u(t)$ of the preceding section by repeating $N$ times the
components related to the optical cavity mode, and
\begin{equation}\label{defintranois}
n^{ext}(t) =\frac{1}{\sqrt{2\kappa}}\left(0,0,X_{in}(t),Y_{in}(t),\ldots,X_{in}(t),Y_{in}(t)\right)^T
\end{equation}
is the analogous extension of the noise vector $n(t)$ of the former section without however the noise acting on the mechanical mode. In
Eq.~(\ref{inoutgen}) we have also introduced the $(2N+2)\times (2N+2)$ block-matrix consisting of $N+1$ two-dimensional blocks
\begin{small}
\begin{equation}\label{transf}
  T(t)=\left(\begin{array}{ccccccc}
    \delta(t) & 0 & 0 & 0 & 0 & 0 & \ldots\\
     0 & \delta(t) & 0 & 0 & 0 & 0 & \ldots\\
    0 & 0 & \sqrt{2\kappa}{\rm Re}g_1(t) & -\sqrt{2\kappa}{\rm Im}g_1(t) & 0 & 0 & \ldots\\
    0 & 0 & \sqrt{2\kappa}{\rm Im}g_1(t) & \sqrt{2\kappa}{\rm Re}g_1(t)& 0 & 0 & \ldots\\
    0 & 0 &  0 & 0 & \sqrt{2\kappa}{\rm Re}g_2(t) & -\sqrt{2\kappa}{\rm Im}g_2(t) &\ldots\\
    0 & 0 &  0 & 0 & \sqrt{2\kappa}{\rm Im}g_2(t) & \sqrt{2\kappa}{\rm Re}g_2(t)&\ldots\\
  \vdots & \vdots & \vdots & \vdots & \vdots & \vdots & \ldots \end{array}\right).
\end{equation}
\end{small}
Using Fourier transforms, and the correlation function of the noises, one can derive the following general expression for the stationary output
correlation matrix, which is the counterpart of the $4\times 4$ intracavity relation of Eq.~(\ref{Vfin})
\begin{equation}
V^{out}=\int d\omega
\tilde{T}(\omega)\left[\tilde{M}^{ext}(\omega)+\frac{P_{out}}{2\kappa}\right]D^{ext}(\omega)\left[\tilde{M}^{ext}(\omega)^{\dagger}+\frac{P_{out}}{2\kappa}
\right]\tilde{T}(\omega)^{\dagger}, \label{Vfinout}
\end{equation}
where $P_{out}={\rm Diag}[0,0,1,1,\ldots]$ is the projector onto the $2N$-dimensional space associated with the output quadratures, and we have
introduced the extensions corresponding to the matrices $\tilde{M}(\omega)$ and $D(\omega)$ of the former section,
\begin{equation}
\tilde{M}^{ext}(\omega)=\left(i\omega+A^{ext}\right)^{-1},
\end{equation}
with
\begin{equation}\label{aexten}
  A^{ext}=\left(\begin{array}{ccccccc}
    0 & \omega_m & 0 & 0 & 0 & 0 & \ldots\\
     -\omega_m & -\gamma_m & G & 0 & G & 0 & \ldots\\
    0 & 0 & -\kappa & \Delta & 0 & 0 & \ldots\\
    G & 0 & -\Delta & -\kappa & 0 & 0 & \ldots\\
    0 & 0 &  0 & 0 & -\kappa & \Delta &\ldots\\
    G & 0 &  0 & 0 & -\Delta & -\kappa &\ldots\\
  \vdots & \vdots & \vdots & \vdots & \vdots & \vdots & \ldots \end{array}\right).
\end{equation}
and
\begin{equation}\label{dexten}
  D^{ext}(\omega)=\left(\begin{array}{ccccccc}
    0 & 0 & 0 & 0 & 0 & 0 & \ldots\\
     0 & \frac{\gamma_m \omega}{\omega_m} \coth\left(\frac{\hbar \omega}{2k_BT}\right) & 0 & 0 & 0 & 0 & \ldots\\
    0 & 0 & \kappa & 0 & \kappa & 0 & \ldots\\
    0 & 0 & 0 & \kappa & 0 & \kappa & \ldots\\
    0 & 0 & \kappa & 0 & \kappa & 0 & \ldots\\
    0 & 0 & 0 & \kappa & 0 & \kappa & \ldots\\
  \vdots & \vdots & \vdots & \vdots & \vdots & \vdots & \ldots \end{array}\right).
\end{equation}
A deeper understanding of the general expression for $V^{out}$ of Eq.~(\ref{Vfinout}) is obtained by multiplying the terms in the integral: one
gets
\begin{eqnarray}
V^{out}&=&\int d\omega
\tilde{T}(\omega)\tilde{M}^{ext}(\omega)D^{ext}(\omega)\tilde{M}^{ext}(\omega)^{\dagger}\tilde{T}(\omega)^{\dagger}+\frac{P_{out}}{2}+  \nonumber\\
&&\frac{1}{2}
\int d\omega \tilde{T}(\omega)\left[\tilde{M}^{ext}(\omega)R_{out}+R_{out}\tilde{M}^{ext}(\omega)^{\dagger}\right]\tilde{T}(\omega)^{\dagger},
\label{Vfinout2}
\end{eqnarray}
where $R_{out}=P_{out} D^{ext}(\omega)/\kappa=D^{ext}(\omega)P_{out}/\kappa$ and we have used the fact that
\begin{equation}
\int \frac{d\omega}{4\kappa^2} \tilde{T}(\omega)P_{out}D^{ext}(\omega)P_{out}\tilde{T}(\omega)^{\dagger}=\frac{P_{out}}{2}. \label{Vfinoutlem}
\end{equation}
The first integral term in Eq.~(\ref{Vfinout2}) is the contribution coming from the interaction between the mechanical resonator and the
intracavity field. The second term gives the noise added by the optical input noise to each output mode. The third term gives the contribution
of the correlations between the intracavity mode and the optical input field, which may cancel the destructive effects of the second noise term
and eventually, even increase the optomechanical entanglement with respect to the intracavity case. We shall analyze this fact in the following
section.

\subsection{A single output mode}

Let us first consider the case when we select and detect only one mode at the cavity output. Just to fix the ideas, we choose the mode specified
by the filter function of Eqs.~(\ref{filterex}) and (\ref{filterex2}), with central frequency $\Omega$ and bandwidth $\tau^{-1}$.
Straightforward choices for this output mode are a mode centered either at the cavity frequency, $\Omega=\omega_c-\omega_0$, or at the driving
laser frequency, $\Omega=0$ (we are in the rotating frame and therefore all frequencies are referred to the laser frequency $\omega_0$), and
with a bandwidth of the order of the cavity bandwidth $\tau^{-1} \simeq \kappa$. However, as discussed above, the motion of the mechanical
resonator generates Stokes and anti-Stokes motional sidebands, consequently modifying the cavity output spectrum. Therefore it may be nontrivial
to determine which is the optimal frequency bandwidth of the output field which carries most of the optomechanical entanglement generated within
the cavity. The cavity output spectrum associated with the photon number fluctuations $S(\omega)= \langle \delta a(\omega)^{\dagger} \delta
a(\omega)\rangle$ is shown in Fig.~\ref{outputspectrum}, where we have considered a parameter regime close to that considered for the
intracavity case, i.e., an oscillator with $\omega _{m}/2\pi =10$ MHz, ${\cal Q}=10^5$, mass $m=50$ ng, a cavity of length $L=1$ mm with finesse
$F=2 \times 10^4$, detuning $\Delta =\omega_m$, driven by a laser with input power $P=30$ mW and wavelength $810$ nm, yielding $G_0=0.43$ kHz,
$G=0.41 \omega_m$, and a cavity bandwidth $\kappa=0.75 \omega_m$. We have again assumed a reservoir temperature for the mirror $T=0.4$ K,
corresponding to $\bar{n}\simeq 833$. This regime is not far but does not corresponds to the best intracavity optomechanical entanglement regime
discussed in Sec.~\ref{Sec:intra}. In fact, optomechanical entanglement monotonically increases with the coupling $G$ and is maximum just at the
bistability threshold, which however is not a convenient operating point. We have chosen instead a smaller input power and a larger mass,
implying a smaller value of $G$ and an operating point not too close to threshold.

\begin{figure}[H]
\centerline{\includegraphics[width=0.7\textwidth]{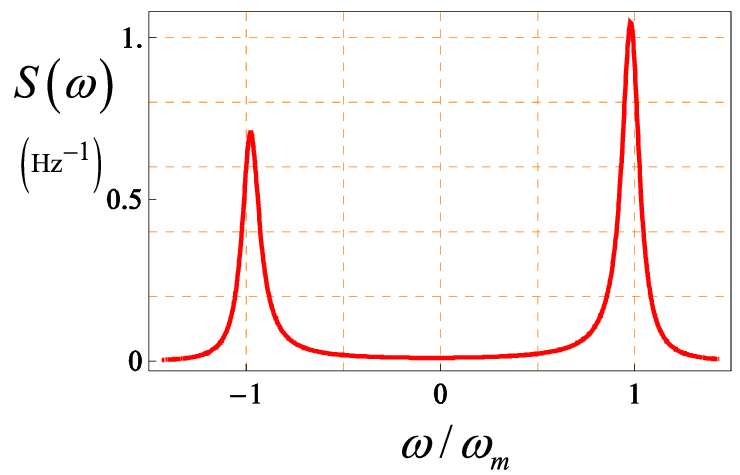}} \caption{(Color online) Cavity output spectrum in the case of an
oscillator with $\omega _{m}/2\pi =10$ MHz, ${\cal Q}=10^5$, mass $m=50$ ng, a cavity of length $L=1$ mm with finesse $F=2 \times 10^4$,
detuning $\Delta =\omega_m$, driven by a laser with input power $P=30$ mW and wavelength $810$ nm, yielding $G_0=0.43$ kHz, $G=0.41 \omega_m$,
and a cavity bandwidth $\kappa=0.75 \omega_m$. We have again assumed a reservoir temperature for the mirror $T=0.4$ K, corresponding to
$\bar{n}\simeq 833$. In this regime photons are scattered only at the two first motional sidebands, at $\omega_0\pm\omega_m$.}
\label{outputspectrum}
\end{figure}

In order to determine the output optical mode which is better entangled with the mechanical resonator, we study the logarithmic negativity
$E_{\mathcal{N}}$ associated with the output CM $V^{out}$ of Eq.~(\ref{Vfinout2}) (for $N=1$) as a function of the central frequency of the mode
$\Omega$ and its bandwidth $\tau^{-1}$, at the considered operating point. The results are shown in Fig.~\ref{output1}, where $E_{\mathcal{N}}$
is plotted versus $\Omega/\omega_m$ at five different values of $\varepsilon=\tau \omega_m$. If $\varepsilon \lesssim 1$, i.e., the bandwidth of
the detected mode is larger than $\omega_m$, the detector does not resolve the motional sidebands, and $E_{\mathcal{N}}$ has a value (roughly
equal to that of the intracavity case) which does not essentially depend upon the central frequency. For smaller bandwidths (larger
$\varepsilon$), the sidebands are resolved by the detection and the role of the central frequency becomes important. In particular
$E_{\mathcal{N}}$ becomes highly peaked around the \emph{Stokes sideband} $\Omega=-\omega_m$, showing that the optomechanical entanglement
generated within the cavity is mostly carried by this lower frequency sideband. What is relevant is that the optomechanical entanglement of the
output mode is significantly larger than its intracavity counterpart and achieves its maximum value at the optimal value $\varepsilon \simeq
10$, i.e., a detection bandwidth $\tau^{-1} \simeq \omega_m/10$. This means that in practice, by appropriately filtering the output light, one
realizes an \emph{effective entanglement distillation} because the selected output mode is more entangled with the mechanical resonator than the
intracavity field.

\begin{figure}[H]
\centerline{\includegraphics[width=0.7\textwidth]{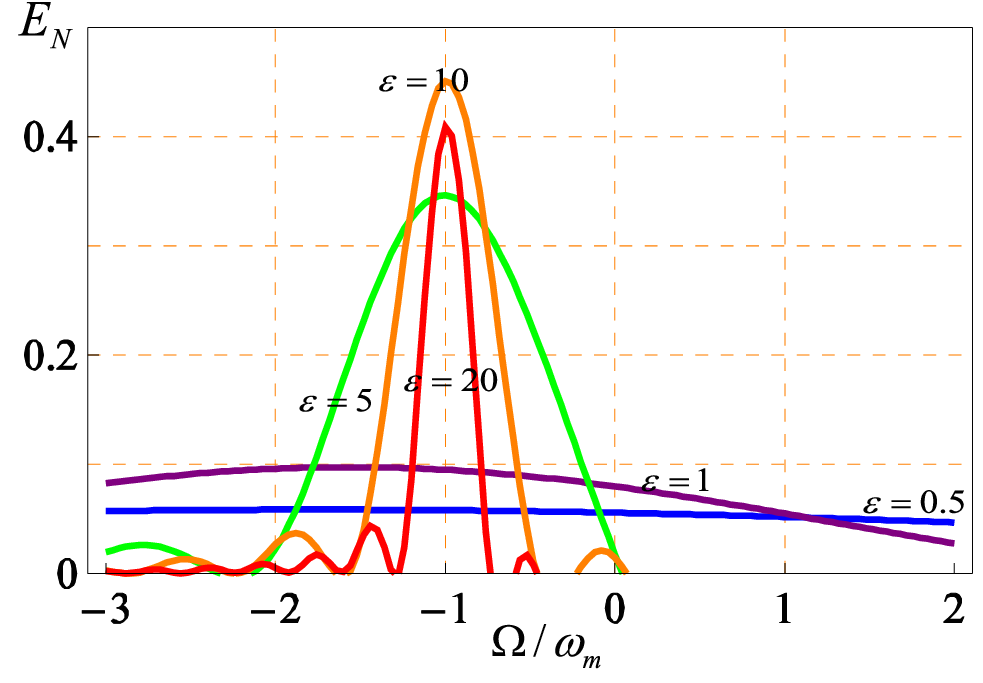}} \caption{(Color online) Logarithmic negativity
$E_{\mathcal{N}}$ of the CV bipartite system formed by mechanical mode and a single cavity output mode versus the central frequency of the
detected output mode $\Omega/\omega_m$ at five different values of its inverse bandwidth $\varepsilon=\omega_m \tau$. The other parameters are
the same as in Fig.~\protect\ref{outputspectrum}. When the bandwidth is not too large, the mechanical mode is significantly entangled only with
the first Stokes sideband at $\omega_0-\omega_m$.} \label{output1}
\end{figure}

The fact that the output mode which is most entangled with the mechanical resonator is the one centered around the Stokes sideband is also
consistent with the physics of two previous models analyzed in Refs.~\cite{fam,prltelep}. In \cite{fam} an atomic ensemble is inserted within
the Fabry-Perot cavity studied here, and one gets a system showing robust tripartite (atom-mirror-cavity) entanglement at the steady state only
when the atoms are resonant with the Stokes sideband of the laser. In particular, the atomic ensemble and the mechanical resonator become
entangled under this resonance condition, and this is possible only if entanglement is carried by the Stokes sideband because the two parties
are only indirectly coupled through the cavity mode. In \cite{prltelep}, a free-space optomechanical model is discussed, where the entanglement
between a vibrational mode of a perfectly reflecting micro-mirror and the two first motional sidebands of an intense laser beam shined on the
mirror is analyzed. Also in that case, the mechanical mode is entangled only with the Stokes mode and it is not entangled with the anti-Stokes
sideband.

By looking at the output spectrum of Fig.~\ref{outputspectrum}, one can also understand why the output mode optimally entangled with the
mechanical mode has a finite bandwidth $\tau^{-1} \simeq \omega_m/10$ (for the chosen operating point). In fact, the optimal situation is
achieved when the detected output mode overlaps as best as possible with the Stokes peak in the spectrum, and therefore $\tau^{-1}$ coincides
with the width of the Stokes peak. This width is determined by the effective damping rate of the mechanical resonator, $\gamma
_{m}^{eff}=\gamma_m+\Gamma$, given by the sum of the intrinsic damping rate $\gamma_m$ and the net laser cooling rate $\Gamma$ of
Eq.~(\ref{netlaser}. It is possible to check that, with the chosen parameter values, the condition $\varepsilon = 10$ corresponds to
$\tau^{-1}\simeq \gamma _{m}^{eff}$.

It is finally important to analyze the robustness of the present optomechanical entanglement with respect to temperature. As discussed above and
shown in \cite{prl07}, the entanglement of the resonator with the intracavity mode is very robust. It is important to see if this robustness is
kept also by the optomechanical entanglement of the output mode. This is shown by Fig.~\ref{robust-optomech}, where the entanglement
$E_{\mathcal{N}}$ of the output mode centered at the Stokes sideband $\Omega=-\omega_m$ is plotted versus the temperature of the reservoir at
two different values of the bandwidth, the optimal one $\varepsilon=10$, and at a larger bandwidth $\varepsilon =2$. We see the expected decay
of $E_{\mathcal{N}}$ for increasing temperature, but above all that also this output optomechanical entanglement is robust against temperature
because it persists even above liquid He temperatures, at least in the case of the optimal detection bandwidth $\varepsilon=10$.

\begin{figure}[H]
\centerline{\includegraphics[width=0.7\textwidth]{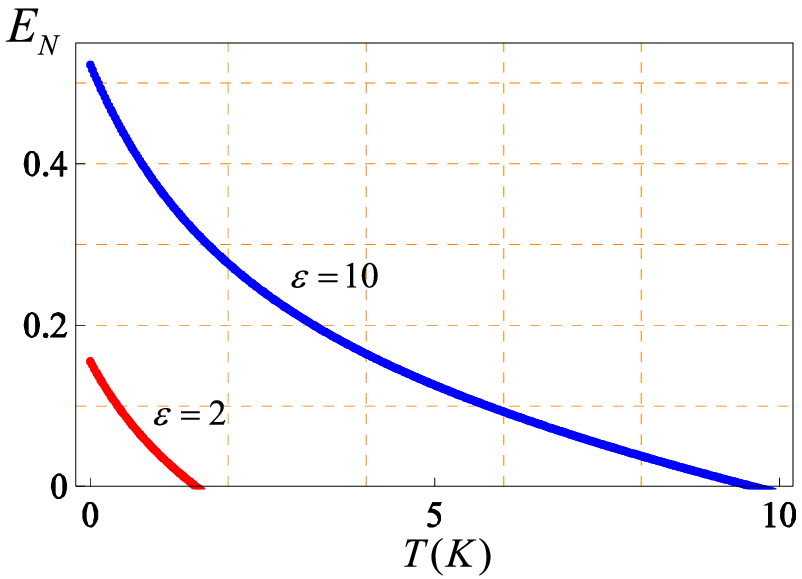}} \caption{(Color online) Logarithmic negativity $E_{\mathcal{N}}$ of the CV
bipartite system formed by mechanical mode and the cavity output mode centered around the Stokes sideband $\Omega=-\omega_m$ versus temperature
for two different values of its inverse bandwidth $\varepsilon=\omega_m \tau$. The other parameters are the same as in
Fig.~\protect\ref{outputspectrum}.} \label{robust-optomech}
\end{figure}

\subsection{Two output modes}

Let us now consider the case where we detect at the output two independent, well resolved, optical output modes. We use again the step-like
filter functions of Eqs.~(\ref{filterex}) and (\ref{filterex2}), assuming the same bandwidth $\tau^{-1}$ for both modes and two different
central frequencies, $\Omega_1$ and $\Omega_2$, satisfying the orthogonality condition of Eq.~(\ref{interfer}) $\Omega_1-\Omega_2=2p \pi
\tau^{-1}$ for some integer $p$, in order to have two independent optical modes. It is interesting to analyze the stationary state of the
resulting tripartite CV system formed by the two output modes and the mechanical mode, in order to see if and when it is able to show i) purely
optical bipartite entanglement between the two output modes; ii) fully tripartite optomechanical entanglement.

The generation of two entangled light beams by means of the radiation pressure interaction of these fields with a mechanical element has been
already considered in various configurations. In Ref.~\cite{giovaEPL01}, and more recently in Ref.~\cite{wipf}, two modes of a Fabry-Perot
cavity system with a movable mirror, each driven by an intense laser, are entangled at the output due to their common ponderomotive interaction
with the movable mirror. In the single mirror free-space model
of Ref.~\cite{prltelep}, the two first motional sidebands are also robustly entangled by the radiation pressure interaction as in a two-mode
squeezed state produced by a non-degenerate parametric amplifier.

The situation considered here is significantly different from that of Refs.~\cite{giovaEPL01,wipf}, which require many driven
cavity modes, each associated with the corresponding output mode. In the present case instead, the different output modes \emph{originate from
the same single driven cavity mode}, and therefore it is much simpler from an experimental point of view. The present scheme can be considered
as a sort of ``cavity version'' of the free-space case of Ref.~\cite{prltelep}, where the reflecting mirror is driven by a single intense laser.
Therefore, as in \cite{prltelep}, one expects to find a parameter region where the two output modes centered around the two motional
sidebands of the laser are entangled. This expectation is clearly confirmed by Fig.~\ref{sideband-ent-sweep}, where the logarithmic negativity
$E_{\mathcal{N}}$ associated with the bipartite system formed by the output mode centered at the Stokes sideband ($\Omega_1=-\omega_m$) and a
second output mode with the same inverse bandwidth ($\varepsilon=\omega_m \tau = 10 \pi$) and a variable central frequency $\Omega$, is plotted
versus $\Omega/\omega_m$. $E_{\mathcal{N}}$ is calculated from the CM of Eq.~(\ref{Vfinout2}) (for $N=2$), eliminating the first two rows
associated with the mechanical mode, and assuming the same parameters considered in the former subsection for the single output mode case. One
can clearly see that bipartite entanglement between the two cavity outputs exists only in a narrow frequency interval around the anti-Stokes
sideband, $\Omega=\omega_m$, where $E_{\mathcal{N}}$ achieves its maximum. This shows that, as in \cite{prltelep}, the two cavity output
modes corresponding to the Stokes and anti-Stokes sidebands of the driving laser are significantly entangled by their common interaction with
the mechanical resonator. The advantage of the present cavity scheme with respect to the free-space case of \cite{prltelep} is that the
parameter regime for reaching radiation-pressure mediated optical entanglement is much more promising from an experimental point of view because
it requires less input power and a not too large mechanical quality factor of the resonator. In Fig.~\ref{sideband-ent}, the dependence of
$E_{\mathcal{N}}$ of the two output modes centered at the two sidebands $\Omega=\pm\omega_m$ upon their inverse bandwidth $\varepsilon$ is
studied. We see that, differently from optomechanical entanglement of the former subsection, the logarithmic negativity of the two sidebands
always increases for decreasing bandwidth, and it achieves a significant value ($\sim 1$), comparable to that achievable with parametric
oscillators, for very narrow bandwidths. This fact can be understood from the fact that quantum correlations between the two sidebands are
established by the coherent scattering of the cavity photons by the oscillator, and that the quantum coherence between the two scattering processes is maximal for output photons with frequencies $\omega_0 \pm \omega_m$.

\begin{figure}[H]
\centerline{\includegraphics[width=0.7\textwidth]{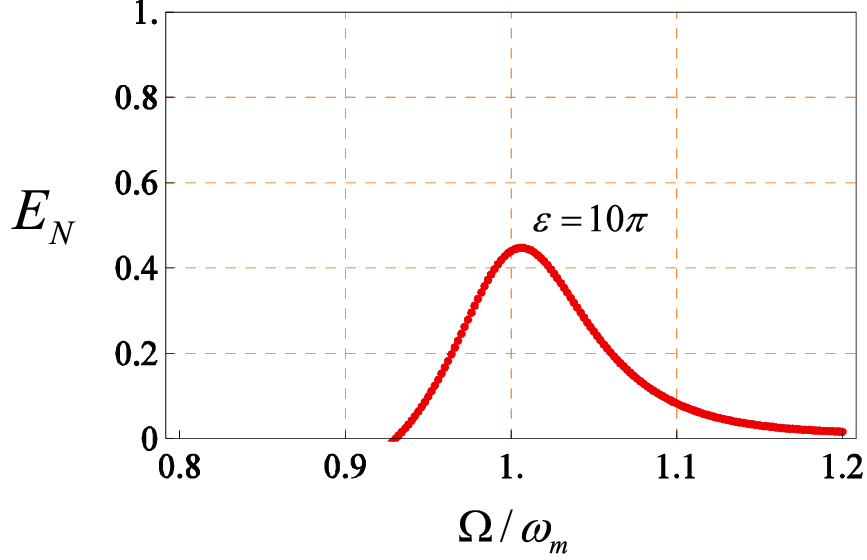}} \caption{(Color online) Logarithmic negativity $E_{\mathcal{N}}$ of the
bipartite system formed by the output mode centered at the Stokes sideband ($\Omega_1=-\omega_m$) and a second output mode with the same inverse
bandwidth ($\varepsilon=\omega_m \tau = 10 \pi$) and a variable central frequency $\Omega$, plotted versus $\Omega/\omega_m$. The other
parameters are the same as in Fig.~\protect\ref{outputspectrum}. Optical entanglement is present only when the second output mode overlaps with
the anti-Stokes sideband.} \label{sideband-ent-sweep}
\end{figure}

\begin{figure}[H]
\centerline{\includegraphics[width=0.7\textwidth]{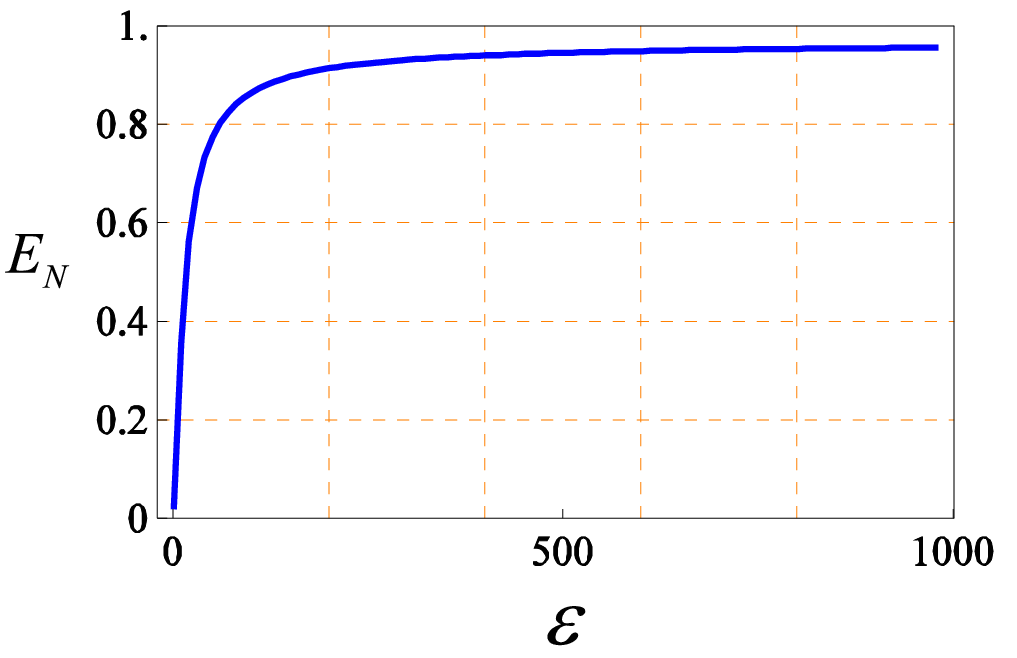}} \caption{(Color online) Logarithmic negativity $E_{\mathcal{N}}$ of the
bipartite system formed by the two output modes centered at the Stokes and anti-Stokes sidebands ($\Omega=\pm\omega_m$) versus the inverse
bandwidth $\varepsilon=\omega_m \tau $. The other parameters are the same as in Fig.~\protect\ref{outputspectrum}.} \label{sideband-ent}
\end{figure}

In Fig.~\ref{robu-sideband} we analyze the robustness of the entanglement between the Stokes and anti-Stokes sidebands with respect to the
temperature of the mechanical resonator, by plotting, for the same parameter regime of Fig.~\ref{sideband-ent}, $E_{\mathcal{N}}$ versus the
temperature $T$ at two different values of the inverse bandwidth ($\varepsilon=10\pi, 100\pi$). We see that this purely optical CV entanglement
is extremely robust against temperature, especially in the limit of small detection bandwidth, showing that the effective coupling provided by
radiation pressure can be strong enough to render optomechanical devices with high-quality resonator a possible alternative to parametric
oscillators for the generation of entangled light beams for CV quantum communication.

\begin{figure}[H]
\centerline{\includegraphics[width=0.7\textwidth]{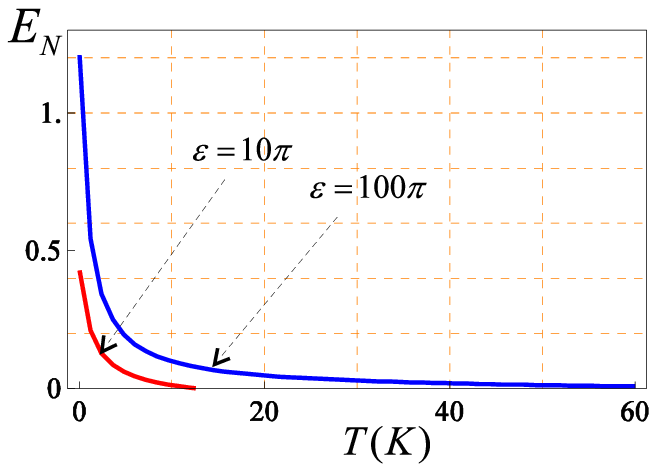}} \caption{(Color online) Logarithmic negativity $E_{\mathcal{N}}$ of
the two output modes centered at the Stokes and anti-Stokes sidebands ($\Omega=\pm\omega_m$) versus the temperature of the resonator reservoir,
at two different values of the inverse bandwidth $\varepsilon=\omega_m \tau $. The other parameters are the same as in
Fig.~\protect\ref{outputspectrum}.} \label{robu-sideband}
\end{figure}

Since in Figs.~\ref{sideband-ent-sweep} and \ref{sideband-ent} we have used the same parameter values for the cavity-resonator system used in
Fig.~\ref{output1}, we have that in this parameter regime, the output mode centered around the Stokes sideband mode shows bipartite entanglement
simultaneously with the mechanical mode and with the anti-Stokes sideband mode. This fact suggests that, in this parameter region, the CV
tripartite system formed by the output Stokes and anti-Stokes sidebands and the mechanical resonator mode could be characterized by a fully
tripartite-entangled stationary state. This is confirmed by Fig.~\ref{tripartite}, where we have applied the classification criterion of
Ref.~\cite{giedke}, providing a necessary and sufficient criterion for the determination of the entanglement class in the case of tripartite CV
Gaussian states, which is directly computable in terms of the eigenvalues of appropriate test matrices \cite{giedke}. These eigenvalues are
plotted in Fig.~\ref{tripartite} versus the inverse bandwidth $\varepsilon$ at $\Delta=\omega_m$ in the left plot, and versus the cavity
detuning $\Delta/\omega_m$ at the fixed inverse bandwidth $\varepsilon=\pi$ in the right plot (the other parameters are again those of
Fig.~\ref{outputspectrum}). We see that all the eigenvalues are negative in a wide interval of detunings and detection bandwidth of the output
modes, showing, as expected, that we have a fully tripartite-entangled steady state.

\begin{figure}[H]
\centerline{\includegraphics[width=0.7\textwidth]{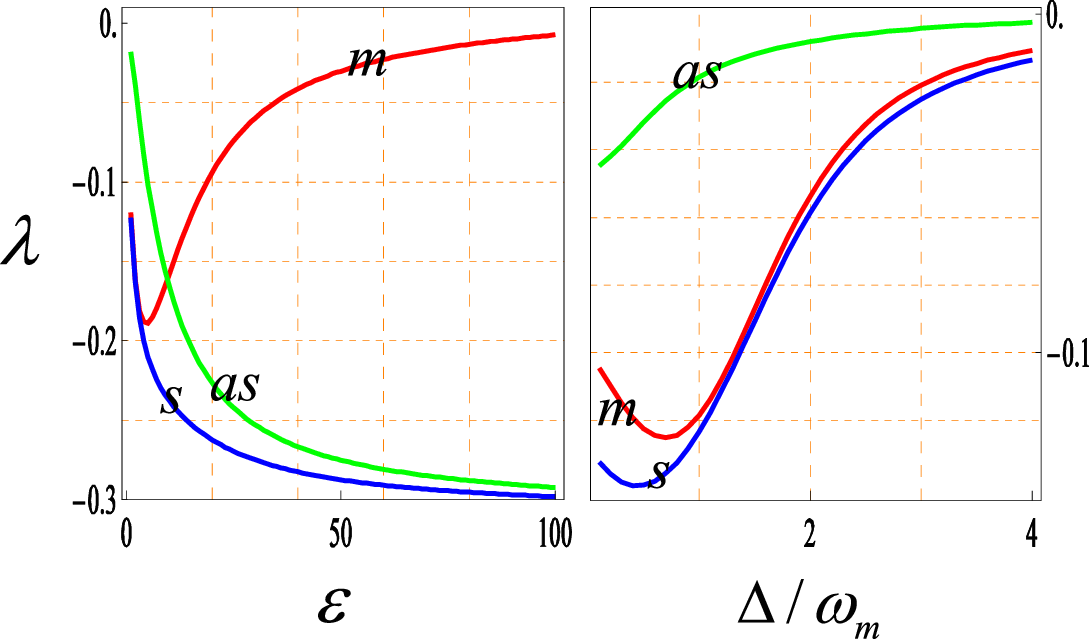}} \caption{(Color online) Analysis of tripartite entanglement. The minimum
eigenvalues after partial transposition with respect to the Stokes mode (blue line), anti-Stokes mode (green line) and the mechanical mode (red
line) are plotted versus the inverse bandwidth $\varepsilon$ at $\Delta=\omega_m$ in the left plot, and versus the cavity detuning
$\Delta/\omega_m$ at the fixed inverse bandwidth $\varepsilon=\pi$ in the right plot. The other parameters are the same as in
Fig.~\protect\ref{outputspectrum}. These eigenvalues are all negative in the studied intervals, showing that one has fully
tripartite-entanglement.} \label{tripartite}
\end{figure}

Therefore, if we consider the system formed by the two cavity output fields centered around the two motional sidebands at $\omega_0\pm \omega_m$
and the mechanical resonator, we find that the entanglement properties of its steady state are identical to those of the analogous tripartite
optomechanical free-space system of Ref.~\cite{prltelep}. In fact, the Stokes output mode shows bipartite entanglement both with the mechanical
mode and with the anti-Stokes mode, the anti-Stokes mode is not entangled with the mechanical mode, but the whole system is in a fully
tripartite-entangled state for a wide parameter regime. What is important is that in the present cavity scheme, such a parameter regime is much
easier to achieve with respect to that of the free-space case.
\end{chapter}

\begin{chapter}{Two driven modes and a movable mirror}\label{2mchapter}
\section{Introduction}
A bichromatic driving of a cavity has been experimentally realized by \cite{Mavalvala}, but their principal aim was to create an optical spring effect useful to trap and cool a heavy ($m\simeq 1 g$) movable mirror. They theoretically predicted also the presence of optical entanglement between the output fields and intracavity light-mirror entanglement \cite{wipf}. However, they computed the output optical entanglement only at the central lasers frequencies.

Here we will use the theory of filters functions, developed in the previous chapter, to select physical output modes centered at different frequencies and with different bandwidths. Moreover the theory of filters will give us the possibility to calculate also the optomechanical entanglement of the output fields with the mirror, in a fully consistent way.

In the previous chapter, we have seen that simply driving a cavity with a single laser we can get good optomechanical  entanglement and also optical entanglement between the Stokes and anti-Stokes sidebands.
That setup has the advantage to be very simple from an experimental point of view. 
However in this chapter we will see that, if we drive a cavity with two different lasers, a \textit{cooling laser} ($\Delta_A=\omega_m$) and a \textit{heating laser} ($\Delta_B=-\omega_m$), we can get better results.

It turns out that optomechanical entanglement is stronger for the output mode associated with the heating laser. This is predicted in the rotating wave approximation, where entanglement is achieved for a field rotating with a frequency opposite to that of mirror. This means that, for every cavity detuning, only the Stokes sideband is entangled with the mirror. Therefore, we can imagine that if the Stokes sidebend is on resonance ($\Delta_B=-\omega_m$), then we can get a better optomechanical entanglement. This is true indeed, but we have already seen that at negative detunings the system is very unstable, since negative detuned lasers tends to \textit{drive} the mirror. The idea is then to use another driving laser, with equal power and opposite detuning $\Delta_A=\omega_m$, to balance the heating/cooling rate and make the system stable. 

In this regime, both Stokes sidebands of the two lasers are entangled with the mirror and the entanglement of the mode driven by the heating laser is stronger, because of the cavity resonance. Moreover we get good results also for the optical entanglement of the output fields, where choosing the detection frequencies resonant with the cavity eigen-modes and using narrow detection bandwidths, we obtain a significant entanglement also at room temperature.

\section{Quantum Langevin equations}

We consider a Fabry-Perot cavity  in which one of the end mirror is a micromechanical oscillator. Two resonant modes of the cavity, $\omega_{cA}$ and $\omega_{cB}$, are driven by two lasers with detuned frequencies $\omega_{0A}$ and $\omega_{0B}$ and powers $P_A$ and $P_B$ respectively.
The Hamiltonian of the system can be written as,
\begin{eqnarray}
&& H=\hbar\omega_{cA}\,a^{\dagger}a+\hbar\omega_{cB}\,b^{\dagger}b+\frac{1}{2}\hbar\omega_{m}(p^{2}+q^{2}) \\
&&-\hbar( G_{0A}\,a^{\dagger}a + G_{0B}\,b^{\dagger}b) q  \notag \\
&& +i\hbar[ E_A(a^{\dagger}e^{-i\omega_{0A}t}-ae^{i\omega_{0A}t})+E_B(b^{\dagger}e^{-i\omega_{0B}t}-be^{i\omega_{0B}t})]. \notag \label{ham0bis}
\end{eqnarray}
In the first line there is the energy of the optical modes and the mirror, which is modeled like a harmonic oscillator of frequency $\omega_m$. The following commutation rules are satisfied, $[a,a^\dag]=[b,b^\dag]=1$ and $[q,p]=i$.  The second line describes the interaction between the mirror and the light due to radiation pressure, with coupling constants $G_{0x}=\sqrt{\hbar/m \omega_m}\,\omega_{cx}/L$, where $m$ is the effective mass of the mechanical oscillator, $L$ is the length of the cavity, and  $x=A,B$. The last line gives the contribution of the two driving lasers, in which $|E_x|=\sqrt{2P_x \kappa/\hbar \omega_{0x}}$, where $\kappa$ is the decay rate of the cavity, supposed to be the same for the two modes. In this treatment, scattering of photons of the driven modes into other cavity modes is neglected; this is acceptable if the mechanical frequency is much smaller then the free spectral range of the cavity, that is $\omega_m << c/L$.

In interaction picture with respect to $\hbar \omega_{0A} a^{\dag}a+\hbar \omega_{0B} b^{\dag}b$, we can derive a set of nonlinear QLE in which optical and mechanical noises are taken into account,
\begin{subequations}\label{QLE}
\begin{eqnarray} 
\dot{q}&=&\omega_m p, \\
\dot{p}&=&-\omega_m q - \gamma_m p + G_{0A} a^{\dag}a+G_{0B}b^\dag b + \xi, \\
\dot{a}&=&-[\kappa+i(\Delta_{0A}-G_{0A}q)]a +E_A +\sqrt{2\kappa} a^{in},\\
\dot{b}&=&-[\kappa+i(\Delta_{0B}-G_{0B}q)]b +E_B +\sqrt{2\kappa} b^{in}.
\end{eqnarray}
\end{subequations}
where $\Delta_{0x}\equiv\omega_{cx}-\omega_{0x}$ are the detunings of the two lasers. The mechanical mode is affected by a
viscous force with damping rate $\gamma_m$ and by a Brownian stochastic
force $\xi(t)$, with zero mean value and the same correlations as given in (\ref{browncorre}).

The optical modes amplitudes instead decay at the rate $\kappa$
and are disturbed by the vacuum radiation input noises $a^{in}(t)$ and $b^{in}(t)$, characterized by zero mean values and the following correlations
\begin{equation}
\langle a^{in}(t)a^{in\dag}(t^{\prime})\rangle=\langle b^{in}(t)b^{in\dag}(t^{\prime})\rangle =
\delta (t-t^{\prime}).
\end{equation}

Setting the time derivatives to zero and solving for the mean values $a_s=\langle a \rangle$, $b_s=\langle b \rangle$, $q_s=\langle q \rangle$, $p_s=\langle p \rangle$, we get
\begin{eqnarray}
&& a_s = \frac{E_A}{\kappa+i \Delta_A}, \label{as}  \\
&& b_s = \frac{E_B}{\kappa+i \Delta_B}, \label{bs}   \\
&& q_s = \frac{G_{0A} |a_s|^2 + G_{0B} |b_s|^2}{\omega_m}, \label{qs}\\
&& p_s = 0.
\end{eqnarray}
where effective detunings $\Delta_{x} \equiv \Delta_{0x}- (G_{0A}^2 |a_s|^2+G_{0B}^2 |b_s|^2)/\omega_m$ have been defined, so that (\ref{as}-\ref{bs}) is actually a nonlinear system, whose solutions gives the stationary amplitudes $a_s$ and $b_s$.

For every operator $O$, we can define the deviation with respect to the steady state as $\delta O = O-\langle O \rangle$. If we apply this transformation to equations (\ref{QLE}) and consider only first order terms in fluctuations, we obtain a set of linearized QLE for the deviation operators,

\begin{eqnarray}  \label{LQLE}
\delta \dot{q}&=&\omega_m \delta p,\\
\delta \dot{p}&=&-\omega_m \delta q - \gamma_m \delta p +G_{0A} \left(a_s\delta a^{\dag}+ a_s ^* \delta a \right) \notag\\
&&+G_{0B} \left(b_s\delta b^{\dag}+ b_s ^* \delta b \right) +\xi,  \notag \\
\delta \dot{a}&=&-(\kappa+i\Delta_A)\delta a +iG_{0A}\, a_s\, \delta q +\sqrt{2\kappa} a^{in},\notag\\
\delta \dot{b}&=&-(\kappa+i\Delta_B)\delta b +iG_{0A}\, b_s\, \delta q +\sqrt{2\kappa} b^{in}.\notag
\end{eqnarray}
To simplify the following steps, it is convenient to define the effective coupling constants $G_A=G_{0A} a_s\sqrt{2}$,   $G_B=G_{0B} b_s\sqrt{2}$ and to choose a reference frame for the cavity modes so that $a_s=|a_s|$, and $b_s=|b_s|$.\footnote{The linearized dynamics of this system is completely equivalent to that of the cavity with a membrane inside. In fact, the only qualitative difference between the two Hamiltonians (\ref{Hmembrane}) and (\ref{ham0bis}) is in the sign of the radiation pressure term of the mode $B$. So, for the ``membrane configuration'', we could repeat the same procedure obtaining the same equations but with the substitution $G_{0B}\rightarrow -G_{0B}$. However, by choosing the reference frame of the mode $B$ such that $b_s=-|b_s|$, the linearized QLE would be exactly equal to (\ref{LQLE}).}

It is better to work with field quadratures, which are defined for the mode ``A" as $\delta X_A\equiv(\delta a+\delta a^{\dag})/\sqrt{2}$, $\delta Y_A\equiv(\delta a-\delta a^{\dag})/i\sqrt{2}$ and similarly for all the other bosonic operators. Now we can finally rewrite equations (\ref{LQLE}) in matrix form as $\dot{\mathbf u} (t) =A \mathbf u(t)+\mathbf n(t) $, where $\mathbf u = (\delta q,\delta p, \delta X_A,\delta Y_A,\delta X_B,\delta Y_B)^T$, $\mathbf n = (0,\xi,\delta X_A^{in},\delta Y_A^{in},\delta X_B^{in},\delta Y_B^{in})^T$, and
\begin{equation}
 A=\left( \begin{array}{cccccc}
0	 &\omega_m&0	    &0	     &0	       &0	\\
-\omega_m&\gamma_m&G_A	    &0	     &G_B      &0	\\
0	 &0	  &-\kappa  &\Delta_A&0	       &0	\\
G_A	 &0	  &-\Delta_A&-\kappa &0	       &0	\\
0	 &0	  &0	    &0	     &-\kappa  &\Delta_B\\
G_B	 &0	  &0	    &0	     &-\Delta_B&-\kappa \\
 \end{array}\right).
\end{equation}
The differential equation can be formally integrated and gives the solution
\begin{equation}\label{u}
\mathbf u (t)=M(t)\mathbf u(0)+\int_0^t ds M(t-s)\mathbf n(s),
\end{equation}
where $M(t)\equiv\exp(At)$.

\section{Stability of the steady state}
The steady state exists and is stable if all the eigenvalues of the $A$ matrix have negative real parts, so that $M(\infty)=0$.
The characteristic polynomial of $A$ is $P(\lambda)=\lambda^6+c_1\lambda^5+c_2\lambda^4+c_3\lambda^3+c_4\lambda^2+c_5\lambda+c_6$, where
\begin{eqnarray}\label{coeff}
c1&=&\gamma_m+4\kappa, \\
c2&=& \Delta_A^2+\Delta_B^2+4 \gamma_m \kappa+6\kappa^2+\omega_m^2,\notag\\
c3&=&\gamma_m (\Delta_A^2+\Delta_B^2+6 \kappa^2)+2 \kappa [\Delta_A^2+\Delta_B^2+2(\kappa^2+\Omega_m^2)], \notag \\
c4&=& \kappa^4+2 \gamma_m \kappa (\Delta_B^2+2 \kappa^2)+6 \kappa^2 \omega_m^2+\Delta_B^2 (\kappa^2+\omega_m^2)+\notag\\
&&\Delta_A^2 (\Delta_B^2+2 \gamma_m \kappa +\kappa^2+\omega_m^2)-\omega_m(G_A^2 \Delta_A+G_B^2 \Delta_B), \notag\\
c5&=& \gamma_m (\Delta_A^2+\kappa^2) (\Delta_B^2+\kappa^2)+2 \kappa\omega_m^2(\Delta_A^2+\Delta^2+2 \kappa^2)\notag\\
&& - 2\kappa\omega_m  (G_A^2 \Delta_A+G_B^2 \Delta_B),\notag\\
c6&=&\omega_m^2 (\Delta_A^2+\kappa^2)(\Delta_B^2+\kappa^2)-\omega_m [G_B^2\Delta_B (\Delta_A^2+\kappa^2) \notag\\
&&  +G_A^2 \Delta_A(\Delta_B^2+\kappa^2)]. \notag
\end{eqnarray}

We may ask if there exists a particular relation between the parameters, such that all the terms containing $G_A$ and $G_B$ in the previous coefficients (\ref{coeff}) cancel out. In this case, the eigenvalues of $A$ would be independent from the power of the two lasers and, in particular, these would be the same eigenvalues of the system with switched off lasers. The stability would be necessarily guaranteed.

Directly from equations (\ref{coeff}), comes out that the eigenvalues of $A$ are independent form $G_A$ and $G_B$ if and only if
\begin{subequations}\label{stab}
\begin{eqnarray}
&&|G_a|=|G_B|=G, \\
&&\Delta_A=-\Delta_B=\Delta.
\end{eqnarray}
\end{subequations}

If the these constraints are satisfied, the system is stable for any values of $G$ and $\Delta$. If instead (\ref{stab}) are not satisfied, we have to apply the Routh-Hurwitz criterion. The inequalities that come out from this criterion are too involved to be reproduced here. 

The particular parameters configuration (\ref{stab}) has a simple physical interpretation. Without loss of generality we can suppose $\Delta>0$, so that the system (\ref{stab}) represents a perfect balance between a cooling laser ($\Delta_A>0$) and a heating laser $\Delta_B<0$. 

This is actually an equilibrium regime very near the stability threshold, in fact it is dangerously broken as soon as the heating rate becomes higher than the cooling one. However we can easily avoid this problem by choosing the power of the heating laser slightly smaller than the cooling one. In this way one remains with certainty in the stable regime, also in the presence of power fluctuations of the driving lasers.

\section{Entanglement measurement simulations using filter functions}
The two principal parameters which characterize a measurement of an output mode are the detection frequency $\Omega$ and the bandwidth $\gamma$ which is also connected with the measuring time $\tau\simeq 1/\gamma$.
Such a measurement can be modelled using the theory given in the previous chapter, where the output field is filtered with a normalized function $g(t)$ oscillating with a mean frequency $\Omega$ and with a bandwidth $\gamma$.

In our particular system, we want to investigate two output modes originating form two different cavity modes, therefore we do not need to choose orthogonal filter functions like that used for the single driving laser.

The simplest choice is then to filter the two output modes $\hat a_A^{out}(t), \hat a_B^{out}(t)$ with two normalized  damped plane waves with frequencies $\Omega_{A,B}$ and the same decay rate $\gamma$,
\begin{eqnarray}
 g_A(t)&=&2 \gamma e^{-\Omega_A i -\gamma t}\theta(t) \,, \\
g_B(t)&=& 2 \gamma e^{-\Omega_B i -\gamma t} \theta(t) \,,
\end{eqnarray}
where $\theta(t)$ is the Heavyside step function.

The two filtered modes are then
\begin{eqnarray}
 a_{\Omega_A}^{out}(t)&=&\int_{-\infty}^{t}ds g_A(t-s) a_A^{out}(s) \,, \\
 a_{\Omega_B}^{out}(t)&=&\int_{-\infty}^{t}ds g_B(t-s) a_B^{out}(s)  \,,
\end{eqnarray}
which Fourier transformed gives 
\begin{eqnarray}
\tilde a_{\Omega_A}^{out}(\omega)&=&\int_{-\infty}^{\infty}\frac{dt}{\sqrt{2\pi}}a_{\Omega_A}^{out}(t)e^{i\omega t}=\sqrt{2\pi}\tilde g_A(\omega) \tilde a_A^{out}(\omega) \,, \\
\tilde a_{\Omega_B}^{out}(\omega)&=&\int_{-\infty}^{\infty}\frac{dt}{\sqrt{2\pi}}a_{\Omega_B}^{out}(t)e^{i\omega t}=\sqrt{2\pi}\tilde g_B(\omega) \tilde a_B^{out}(\omega) \,,  \,
\end{eqnarray}
where $\tilde g_{A,B}(\omega)$ are the Fourier transforms of the filter functions
\begin{eqnarray}
 \tilde g_A(\omega)&=&  \frac{\sqrt{\gamma/\pi}}{\gamma+i(\Omega_A-\omega)}\,, \\
\tilde g_B(\omega)&=&  \frac{\sqrt{\gamma/\pi}}{\gamma+i(\Omega_B-\omega)} \,.
\end{eqnarray}
Now, in analogy with the previous chapter, we define:

\begin{itemize}
 \item a quadrature vector containing the mirror operators and the quadratures of the filtered modes,
\begin{equation}
u^{out}(t)=\left[0,0,\delta x_{\Omega_A}^{out}(t),\delta y_{\Omega_A}^{out}(t),\delta x_{\Omega_B}^{out}(t),\delta y_{\Omega_B}^{out}(t)\right]^T\,,
\end{equation}
\item the corresponding correlation matrix,
\begin{equation}
V^{out}(t)=\frac{1}{2}\langle u^{out}_i(t)u^{out}_j(t)+ u^{out}_j(t) u^{out}_i(t)\rangle \,,
\end{equation}
\item the transformation matrix containing the filter functions,
\begin{small}
\begin{equation}\label{transform2}
  T(t)=\left(\begin{array}{cccccc}
    \delta(t) & 0 & 0 & 0 & 0 & 0 \\
     0 & \delta(t) & 0 & 0 & 0 & 0\\
    0 & 0 & \sqrt{2\kappa}{\rm Re}g_A(t) & -\sqrt{2\kappa}{\rm Im}g_A(t) & 0 & 0\\
    0 & 0 & \sqrt{2\kappa}{\rm Im}g_A(t) & \sqrt{2\kappa}{\rm Re}g_A(t)& 0 & 0 \\
    0 & 0 &  0 & 0 & \sqrt{2\kappa}{\rm Re}g_B(t) & -\sqrt{2\kappa}{\rm Im}g_B(t)\\
    0 & 0 &  0 & 0 & \sqrt{2\kappa}{\rm Im}g_B(t) & \sqrt{2\kappa}{\rm Re}g_B(t)\\
 \end{array}\right),
\end{equation}
\end{small}
\item the matrix associated to the noise operators, which differently form the single mode case is diagonal because the input noise of the two optical modes are uncorrelated,
\begin{equation}\label{dextenbis}
  D(\omega)=\left(\begin{array}{cccccc}
    0 & 0 & 0 & 0 & 0 & 0\\
     0 & \frac{\gamma_m \omega}{\omega_m} \coth\left(\frac{\hbar \omega}{2k_BT}\right) & 0 & 0 & 0 & 0 \\
    0 & 0 & \kappa & 0 & 0 & 0 \\
    0 & 0 & 0 & \kappa & 0 & 0 \\
    0 & 0 & 0 & 0 & \kappa & 0 \\
    0 & 0 & 0 & 0 & 0 & \kappa \\
\end{array}\right).
\end{equation}

\end{itemize}
Repeating the same procedure used in the previous chapter, we find that $V^{out}$ is stationary in time, ad it is given by the same simple formula
\begin{equation}
V^{out}=\int d\omega
\tilde{T}(\omega)\left[\tilde{M}(\omega)+\frac{P_{out}}{2\kappa}\right]D(\omega)\left[\tilde{M}(\omega)^{\dagger}+\frac{P_{out}}{2\kappa}
\right]\tilde{T}(\omega)^{\dagger}, \label{Vfinoutbis}
\end{equation}
where $P_{out}={\rm Diag}[0,0,1,1,1,1]$ is the projector onto the optical quadratures, $\tilde T(\omega)$ is the Fourier transform of (\ref{transform2}), and $\tilde M(\omega)$ if the Fourier transform of $M(t)$ defined in (\ref{u}),
\begin{equation}
\tilde{M}(\omega)=\left(i\omega+A\right)^{-1}.
\end{equation}

From the correlation matrix $V^{out}$ we can extract all the informations about the steady state of the system.
In particular, choosing realistic parameters, we can compute the entanglement between the various $1\times1$ bipartite state obtained tracing out one of the optical output fields or the mechanical mode of the mirror.

\subsection{Cooling mode-mirror optomechanical entanglement}
First of all we consider the entanglement between the cooling mode (A) and the mirror. The situation is very similar to the case that we have already considered with only one driving cooling laser. We choose the parameters very similar to those of the previous chapter, i.e., an oscillator with $\omega _{m}/2\pi =10$ MHz, ${\cal Q}=10^5$, mass $m=50$ ng, a cavity of length $L=1$ mm with finesse
$F=2 \times 10^4$, driven by two lasers with input powers $P_A=15$ mW, $P_B=13$ mW, detunings $\Delta_A=\omega_m$, $\Delta_B=-\omega_m$ and wavelengths both around $810$ nm. We have again assumed a reservoir temperature for the mirror $T=0.4$ K,
corresponding to $\bar{n}\simeq 833$.

We are in the stable regime defined in (\ref{stab}), but the power of the heating laser $P_B$ is chosen slightly smaller than $P_A$ to be far from instabilities. In fact all the eigenvalues $\lambda_i$ of $A$ have negative real parts and the maximum of them is well far away from zero: $\max\{\textrm{Re}(\lambda_i)\}\simeq -2\times 10^{5} \textrm{ Hz}$.

In (Fig.\ref{optomech-cool}) a plot of the entanglement logarithmic negativity is given as a function of the measured frequency, for three different detection bandwidths determined by the parameter $\epsilon=\omega_m / \gamma$. 

\begin{figure}[H]
\centerline{\includegraphics[width=0.7\textwidth]{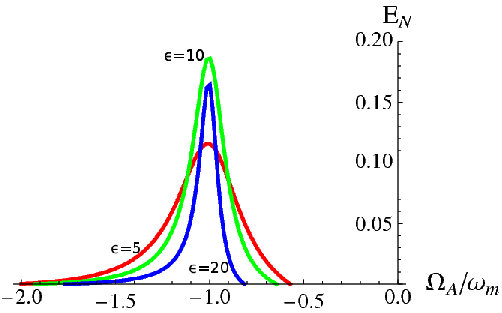}} 
\caption{Logarithmic negativity $E_{\mathcal{N}}$ of the CV bipartite system formed by the mechanical mode and the cooling output beam versus the measured frequency $\Omega_A/\omega_m$ at three different values of the inverse detection bandwidth $\epsilon=\omega_m/\gamma$. As expected only the Stokes sideband $\Omega_A=-\omega_m$ is entangled with the mirror but, since it is off resonance, we obtain a small entanglement. The other parameters are $\omega _{m}/2\pi =10$ MHz, ${\cal Q}=10^5$, $m=50$ ng, $L=1$ mm, $F=2 \times 10^4$, $P_A=15$ mW, $P_B=13$ mW, $\Delta_A=\omega_m$, $\Delta_B=-\omega_m$, $\lambda=810$ nm, $T=0.4$ K.}\label{optomech-cool}
\end{figure}

The behavior of the optomechanical entanglement with respect to temperature is given in (Fig.\ref{optomech-cool-T}).

\begin{figure}[H]
\centerline{\includegraphics[width=0.7\textwidth]{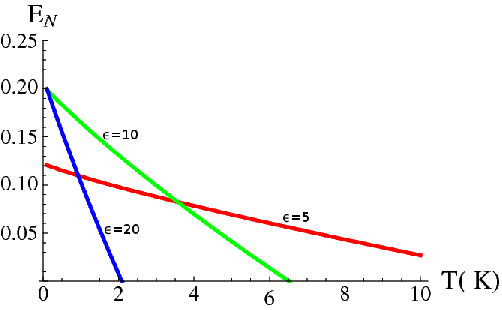}} 
\caption{Logarithmic negativity $E_{\mathcal{N}}$ of the CV bipartite system formed by the mechanical mode and the Stokes sideband ($\Omega_A=-\omega_m$) of the cooling laser versus the reservoir temperature $T$, at three different values of the inverse detection bandwidth $\epsilon=\omega_m/\gamma$. The other system parameters are the same as in Fig. \ref{optomech-cool}.}\label{optomech-cool-T}
\end{figure}

If we compare (Fig.\ref{optomech-cool}) with the monochromatic case (Fig.\ref{output1}), we observe that the addition of a heating driving laser does not improve the situation, but it introduces more noise, decreasing the entanglement of the cooling mode.

\subsection{Heating mode-mirror optomechanical entanglement}
The advantage of the bichromatic driving is instead significant if we consider the entanglement of the heating laser beam (B) with the mirror. The entangled sideband, as predicted in the RWA, is again the Stokes one, but this time it is resonant with the cavity, in fact $\Delta=-\omega_m$ implies $\omega_{cA}=\omega_{0A}-\omega_m=\omega_{Stokes}$.

In (Fig.\ref{optomech-cool}) a plot of the entanglement logarithmic negativity is given as a function of the measured frequency, for three different detection bandwidths. 

\begin{figure}[H]
\centerline{\includegraphics[width=0.7\textwidth]{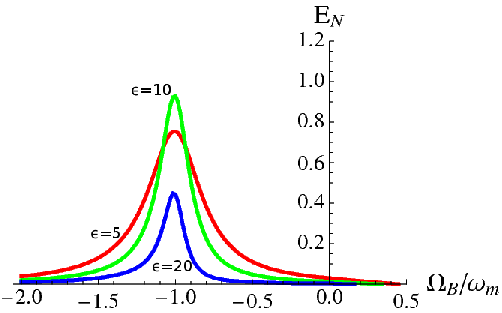}} 
\caption{Logarithmic negativity $E_{\mathcal{N}}$ of the CV bipartite system formed by the mechanical mode and the heating output beam versus the measured frequency $\Omega_B/\omega_m$ at three different values of the inverse detection bandwidth $\epsilon=\omega_m/\gamma$. The Stokes sideband $\Omega_B=-\omega_m$, which is on resonance with the cavity, is very entangled with the mirror. The other system parameters are the same in Fig. \ref{optomech-cool}.}\label{optomech-heat}
\end{figure}
The behavior of the optomechanical entanglement with respect to temperature is given in (Fig.\ref{optomech-heat-T}).
By choosing appropriate detection bandwidths, a significant amount of entanglement persists even above liquid He temperatures.
\end{chapter}

\begin{figure}[H]
\centerline{\includegraphics[width=0.7\textwidth]{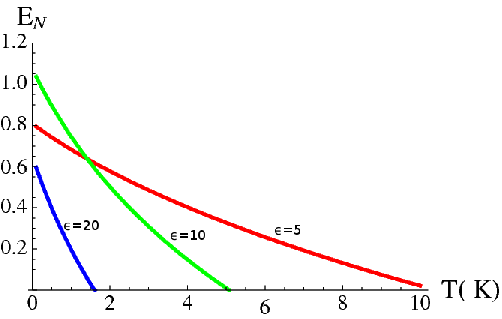}} 
\caption{Logarithmic negativity $E_{\mathcal{N}}$ of the CV bipartite system formed by the mechanical mode and the Stokes sideband ($\Omega_B=-\omega_m$) of the heating laser versus the reservoir temperature $T$, at three different values of the inverse detection bandwidth $\epsilon=\omega_m/\gamma$. The other system parameters are the same as in Fig. \ref{optomech-cool}.}\label{optomech-heat-T}
\end{figure}
\subsection{Heating mode-cooling mode optical entanglement}
Let us consider now the purely optical entanglement between the two output light beams. We optimize the system parameters, by choosing a cavity with a higher finesse $\mathcal F=8 \times 10^{5}$ and a lighter movable mirror $m=10$ ng. We increase also the power of the driving lasers to $P_A=75$ mW and $P_B=65$ mW. We are also in this case near the stable situation of (\ref{stab}), with $P_B$ (heating laser) slightly smaller than $P_A$ (cooling laser).

Stability is confirmed by the eigenvalues of $A$, all having negative real parts with $\max\{\textrm{Re}(\lambda_i)\}\simeq -5\times 10^{5}$.

The RWA suggests us to consider, also in this case, counter-rotating modes to achieve a better entanglement. In fact it comes out that the best choice is to consider the anti-Stokes sideband of the cooling beam and the Stokes sideband of the heating beam, that is $\Omega_A=\omega_m$ and $\Omega_B=-\omega_m$. We observe that also in this case the entangled sidebands are both resonant with the cavity.

In (Fig.\ref{opto-opto}), the logarithmic entanglement negativity is plotted versus the measured central frequency $\Omega_A/\omega_m$ of the cooling mode, while the detected frequency of the heating mode is centered on the Stokes sideband $\Omega_B=-\omega_m$. 

\begin{figure}[H]
\centerline{\includegraphics[width=0.7\textwidth]{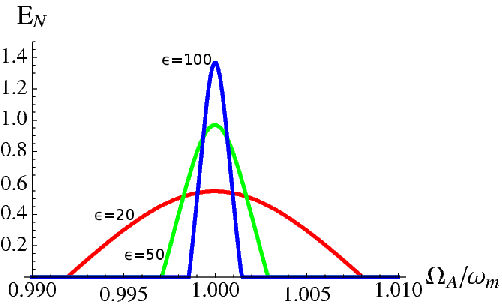}} 
\caption{Logarithmic negativity $E_{\mathcal{N}}$ of the CV bipartite system formed by the cooling and the heating output beams versus the measured frequency $\Omega_A/\omega_m$ of the cooling beam at three different values of the inverse detection bandwidth $\epsilon=\omega_m/\gamma$. The detected frequency of the heating beam is centered on the Stokes sideband. The system parameters are $\omega _{m}/2\pi =10$ MHz, ${\cal Q}=10^5$, $m=10$ ng, $L=1$ mm, $F=8 \times 10^5$, $P_A=75$ mW, $P_B=65$ mW, $\Delta_A=\omega_m$, $\Delta_B=-\omega_m$, $\lambda=810$ nm, $T=0.4$ K.}\label{opto-opto}
\end{figure}

Differently form the optomechanical entanglement, the logarithmic negativity associated to the two optical output beams, always increases for decreasing detection bandwidths. However the more the bandwidths are narrow, the more the measured frequencies must be centered with high precision and this imposes a realistic upper limit for the parameter $\epsilon$.

By measuring narrow bandwidths, thermal noise can be reduced, giving a significant amount of optical entanglement even at room temperature (see Fig.\ref{opto-opto-T}).

\begin{figure}[H]
\centerline{\includegraphics[width=0.7\textwidth]{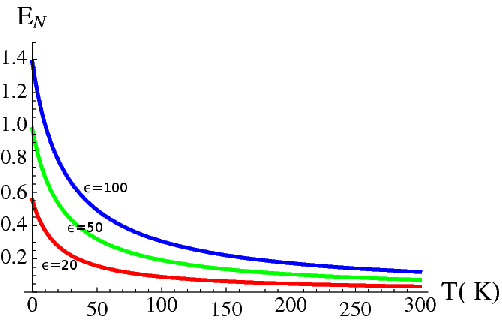}} 
\caption{Logarithmic negativity $E_{\mathcal{N}}$ of the CV bipartite system formed by the anti-Stokes sideband of the cooling mode ($\Omega_A=\omega_m$) and the Stokes sideband ($\Omega_B=-\omega_m$) of the heating laser versus the reservoir temperature $T$, at three different values of the inverse detection bandwidth $\epsilon=\omega_m/\gamma$. The other system parameters are the same as in Fig. \ref{opto-opto}.}\label{opto-opto-T}
\end{figure}

\subsection{Tripartite entanglement}
In analogy with the monochromatic configuration, we check the presence of tripartite entanglement between the Stokes sideband of the heating mode, the anti-Stokes sideband of the cooling mode and the mirror.

With the same procedure used in the previous chapter, we consider the correlation matrices after partial transposition with respect to one of the three sub-systems and then we check if the Heisenberg condition (\ref{bonafide}) is violated.
 In Fig. \ref{tripartite-plots}, the minimum eigenvalue of the three test matrices is plotted as a function of the inverse bandwidth $\epsilon$.
Also in this bichromatic setup, we find a large region in which the system is in a fully tripartite-entangled state (all the three test matrices are not positive definite).
\begin{figure}[H]
\centerline{\includegraphics[width=0.7\textwidth]{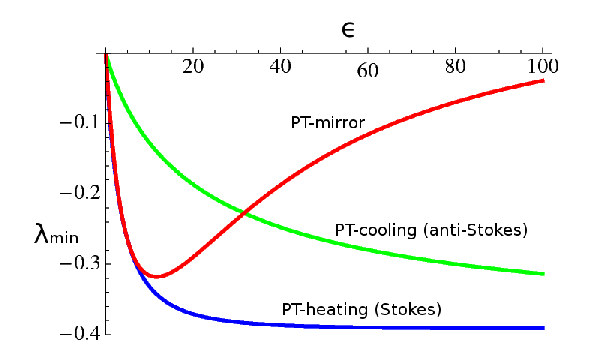}} 
\caption{(Color online) Analysis of tripartite entanglement. The minimum eigenvalues after partial transposition with respect to the Stokes sideband ($\Omega_B=-\omega_m$) of the heating mode (blue line), anti-Stokes sideband ($\Omega_A=\omega_m$) of the cooling mode (green line) and the mechanical mode (red
line) are plotted versus the inverse bandwidth $\epsilon=\omega_m/\gamma$. The other parameters are the same as in
Fig.~\protect\ref{optomech-cool}. These eigenvalues are all negative in the studied interval, showing that one has fully
tripartite-entanglement.}\label{tripartite-plots}
\end{figure}

\backmatter
\chapter{Conclusions}

In the first part of the thesis some aspects of continuous variable quantum information has been investigated.

After a general introduction on the phase space representation of Gaussian states, we have studied the quantum mechanical phenomenon of \textit{entanglement} restricted to CV systems. The most relevant separability criteria and entanglement measures have been given. 

In chapter \ref{cvtele}, the Braunstein-Kimble protocol for the teleportation of Gaussian states has been described in the Heisenberg picture, while in chapter \ref{optimization}, we have presented some original results on the optimization of the teleportation fidelity by local Gaussian operations. We have shown that, for a given shared entanglement, the maximum fidelity of teleportation is bounded below and above by simple expressions depending upon the lowest PT symplectic eigenvalue $\nu$ only (see Eq.~(\ref{lowupp})). We have seen that these bounds are quite tight and that the upper bound of the fidelity is reached if and only if Alice and Bob share a symmetric entangled state. We have also determined the general form of the CM of Alice and Bob state after the optimization procedure. Then we have restricted to local TGCP maps and shown that the optimal TGCP map is composed by a local symplectic map, eventually followed by an attenuation either on Alice or on Bob mode. Finally we have shown how the corresponding value of the maximum fidelity $\mathcal{F}_{opt}$ can be derived from the knowledge of the symplectic invariants of the initial CV entangled state shared by Alice and Bob. The optimization by generic GCP maps (i.e., including also Gaussian measurements on ancillas) or by non-Gaussian local operations are still open questions. 

The second part of the thesis is all focused on optomechanical quantum systems.

In chapter \ref{optomech}, we have derived the quantum Langevin equations related to optical cavities (input-output theory) and to mechanical systems (quantum Brownian motion). The radiation pressure potential between a mechanical resonator and an optical cavity mode has been also derived through different approaches.

In chapter \ref{1mchapter} we have considered an optical cavity mode coupled with an oscillating mirror by radiation pressure. We have studied in details the light-mirror entanglement and its relationship with the laser-cooling of the mechanical mode. We have developed a consistent theory to define filtered output modes, with given central frequencies and detection bandwidths.
We have used this theory to show that the Stokes and anti-Stokes sidebands (photons scattered by the mirror) are optically entangled, and that the Stokes one is also robustly entangled with the mirror vibrational mode. We have seen also that the two sidebands and the mirror mode form a fully tripartite-entangled state.

In chapter \ref{2mchapter} we have extended the previous model by using two driving lasers. We have observed that if the two lasers have equal powers and opposite detunings, then a balance between the cooling and the heating rate can be reached, giving the possibility to achieve interesting stable regimes. In this configuration, the Stokes sideband of the heating laser can be chosen resonant with the cavity and so it can reach a high amount of entanglement with the mirror. Moreover, two output modes associated to counter rotating sidebands (e.g. Stokes of the heating laser and anti-Stokes of the cooling laser) can achieve a significant optical entanglement even at room temperature, while at low temperature they form, together with the mirror vibrational mode, a fully tripartite-entangled state.

\end{document}